\documentclass[11pt, reqno, oneside]{amsart}

\usepackage{amssymb, amsfonts, mathrsfs, verbatim,graphicx,graphics,epsfig,psfrag,mathtools, bm, bbm}

\usepackage{mathabx}
\usepackage{microtype} 

\usepackage{color, xcolor}

\usepackage{accents}
\newcommand{\ubar}[1]{{\underaccent{\bar}{\smash{#1}}}}

\usepackage{abstract}
\usepackage[all]{xy}

\usepackage[left= 3.4 cm, right= 3.4 cm, top= 2.8 cm, bottom=2.7 cm, foot= 1.2 cm, marginparwidth=2.7 cm, marginparsep=0.5 cm]{geometry}

\usepackage[inline]{enumitem}

\usepackage{hyperref}
\hypersetup{colorlinks}
\definecolor{darkred}{rgb}{0.5,0,0}
\definecolor{darkgreen}{rgb}{0,0.5,0}
\definecolor{darkblue}{rgb}{0,0,0.5}
\hypersetup{colorlinks, linkcolor=darkblue, filecolor=darkgreen, urlcolor=darkred, citecolor=darkblue}
\makeatletter 
\@addtoreset{equation}{section}
\makeatother  

\numberwithin{equation}{section}

\usepackage{tikz-cd}
\usepackage{tikz}
\usetikzlibrary{arrows}
\usepackage{float}

\usetikzlibrary{decorations.pathreplacing}

\newtheorem{thm}{Theorem}[section]
\newtheorem{cor}[thm]{Corollary}

\newtheorem{prop}[thm]{Proposition}
\newtheorem{lemma}[thm]{Lemma}

\theoremstyle{definition}
\newtheorem{defn}[thm]{Definition}
\theoremstyle{remark}
\newtheorem{rem}[thm]{Remark}
\newtheorem{hyp}[thm]{Hypothesis}

\newtheorem*{claim}{Claim}

\newcounter{notes}
{\end{list}}

\makeatletter
\@addtoreset{thm}{section}
\makeatother

\newcommand\qu{/\kern-.7ex/} 
\renewcommand{\setminus}{\smallsetminus}

\newcommand{\beq}{\begin{equation}}
\newcommand{\eeq}{\end{equation}}
\newcommand{\beqn}{\begin{equation*}}
\newcommand{\eeqn}{\end{equation*}}
\newcommand{\ov}{\overline}
\newcommand{\mb}{\mathbb}
\renewcommand{\tt}{\texttt}
\newcommand{\mc}{\mathcal}
\newcommand{\mf}{\mathfrak}

\newcommand{\V}{{\rm V}}
\newcommand{\T}{{\rm T}}
\newcommand{\E}{{\rm E}}

\newcommand{\preq}{\preccurlyeq}

\newcommand{\wt}{\widetilde}
\newcommand{\wh}{\widehat}

\renewcommand{\i}{{\bf i}}

\newcommand{\bt}{\blacktriangle}

\newcommand{\bbox}{{\scalebox{0.5}{$\blacksquare$}}}

\newcommand{\ev}{{\rm ev}}
\newcommand{\eev}{\mf{ev}}

\newcommand{\vt}{{\rm vert}}

\newcommand{\Gammait}{{\mathit{\Gamma}}}

\newcommand{\Thetait}{{\mathit{\Theta}}}

\newcommand{\Piit}{{\mathit{\Pi}}}

\newcommand{\uds}[1]{\underline{\smash{#1}}}

\title{Gauged Linear Sigma Model in Geometric Phase. II. The Virtual Cycle}

\begin{document}

\author[Tian]{Gang Tian}
\address{
Beijing International Center for Mathematical Research\\
Beijing University\\
Beijing, 100871, China
}
\email{gtian@math.pku.edu.cn}

\author[Xu]{Guangbo Xu}
\address{
Department of Mathematics\\
Rutgers University--New Brunswick\\
110 Frelinghuysen Road, Piscataway, NJ 08854, USA
}
\email{guangbo.xu@rutgers.edu}

\date{\today}

\thanks{The second named author has been supported by the Simons Foundation through the Homological Mirror Symmetry Collaboration grant, AMS-Simons Travel Grant, and NSF grant DMS-2345030. The majority of this paper was finished when the second named author was in Princeton University.}

\maketitle

\begin{abstract}
We provide the detailed construction of the virtual cycles needed for defining the cohomological field theory associated to a gauged linear sigma model in geometric phase. 
\end{abstract}



\tableofcontents

\section{Introduction}

The {\bf gauged linear sigma model (GLSM)} is a two-dimensional supersymmetric quantum field theory introduced by Witten \cite{Witten_LGCY} in 1993. It has stimulated many important developments in both the mathematical and the physical studies of string theory and mirror symmetry. For example, it plays a fundamental role in physicists argument of mirror symmetry \cite{Hori_Vafa}\cite{Gu_Sharpe_2018}. Its idea is also of crucial importance in the verification of genus zero mirror symmetry for quintic threefold \cite{Givental_96}\cite{LLY_1}. 

Since 2012 the authors have initiated a project aiming at constructing a mathematically rigorous theory for the GLSM, using mainly the method from symplectic and differential geometry. The mathematical theory of the GLSM is the generalization of other known theories including the Gromov--Witten theory for symplectic manifolds, the FJRW theory for orbifold Landau--Ginzburg models (see \cite{FJR_annals}), and the theory of vortices (see \cite{Cieliebak_Gaio_Salamon_2000}\cite{Mundet_thesis, Mundet_2003}). The construction is based on the intersection theory on the moduli space of the {\bf gauged Witten equation}, a first-order equation generalizing the Witten equation in FJRW theory and the vortex equation. In our previous works \cite{Tian_Xu, Tian_Xu_2, Tian_Xu_2021}, we have constructed certain correlation function (i.e. Gromov--Witten type invariants) under certain special conditions: the gauge group is $U(1)$ and the superpotential is a Lagrange multiplier. This correlation function is though rather restricted, as for example, we do not know if they satisfy splitting axioms or not. 

A major difficulty in the study of the (gauged) Witten equation lies in the so-called {\it broad} case, in which the Fredholm property of the equation is problematic. This issue also causes difficulties in various algebraic approaches. Now we have realized that if we restrict to the so-called {\bf geometric phase}, then the issue about broad case disappears in our symplectic setting. Our construction in geometric phase has been outlined in \cite{Tian_Xu_2017}. It ends up at constructing a {\bf cohomological field theory (CohFT)}, namely a collection of correlation functions which satisfy the splitting axioms. The detailed construction in this scenario is the main objective of the current paper.

\subsection{The main theorem}

In this paper we complete the virtual cycle construction required for the definition of the cohomological field theory given in \cite{Tian_Xu_geometric}. Meanwhile we provide the proofs of a few technical results about the gauged Witten equations. To give a short introduction to these results, we first recall a few notations and setups. 

A {\bf GLSM space} (see \cite[Definition 3.1]{Tian_Xu_geometric}) is a quadruple $(V, G, W, \mu)$ where $V$ is a noncompact K\"ahler manifold (usually a vector space), $G$ is a reductive Lie group acting on $V$, $W$ is a $G$-invariant holomorphic function on $V$, and $\mu$ is a moment map of the $G$-action restricted to the maximal compact subgroup $K \subset G$. We also impose other conditions (see details in Section \ref{section3}), such as the existence of an R-symmetry.
The gauged Witten equation is defined over $r$-spin curves (see Definition \ref{defn_r_spin}). An {\it $r$-spin curve} is a quadruple ${\mc C} = (\Sigma, {\bf z}, L, \varphi)$ where $\Sigma$ is a smooth or nodal orbifold Riemann surface, ${\bf z}$ is the set of orbifold marked points, $L$ is an orbifold line bundle over $\Sigma$, and $\varphi$ is an isomorphism from $L^{\otimes r}$ to the log-canonical bundle of $(\Sigma, {\bf z})$.

The variables of the gauged Witten equation are gauged maps. A {\it gauged map} from an $r$-spin curve ${\mc C}$ to the GLSM space $(V, G, W, \mu)$ is a triple $(P, A, u)$, where $P$ is a principal $K$-bundle over $\Sigma$, $A$ is a connection on $P$, and $u$ is a section of the fibre bundle $P(V)$ associated to $P$. For such a triple, the gauged Witten reads (see Section \ref{section3} for details)
\begin{align*}
&\ \ov\partial_A u + \nabla W (u) = 0,\ &\ * F_A + \mu(u) = 0.
\end{align*}

In the companion paper \cite{Tian_Xu_geometric} we proved the compactness result of the moduli space of the gauged Witten equation. Indeed, for each pair of nonnegative integers $(g, n)$ with $2g + n \geq 3$ we defined the moduli space 
\beqn
\ov{\mc M}{}_{g, n}^r (V, G, W, \mu)
\eeqn
of {\it stable} solutions to the gauged Witten equation over genus $g$, $n$-marked $r$-spin curves. Moreover, for a certain combinatorial type ${\sf\Gamma}$ describing possible degenerations of the solutions, there is a compact subset  
\beqn
\ov{\mc M}{}_{\sf\Gamma}(V, G, W, \mu) \subset \ov{\mc M}{}_{g,n}^r(V, G, W, \mu)
\eeqn
consisting of solutions of type $\sf\Gamma$. 

The correlation functions is defined via certain evaluation map, similar to the case of Gromov--Witten invariants. Recall that the {\bf classical vacuum} of a GLSM space $(V, G, W, \mu)$ is defined as 
\beqn
X:= {\rm Crit}W \qu G:= ( {\rm Crit} W \cap \mu^{-1}(0)) / K.
\eeqn
We assume that this quotient is free, hence $X$ is a manifold. There are evaluation maps
\beqn
\ov{\mc M}_{\sf\Gamma} ( V, G, W, \mu; B) \to \ov{\mc M}_\Gamma \times X^n.
\eeqn
Here $\Gamma$  is the underlying curve type and $\ov{\mc M}_\Gamma \subset \ov{\mc M}{}_{g, n}$ is the corresponding stratum in the Deligne--Mumford space. The main theorem of this paper is the following. 

\begin{thm}[Main Theorem]\cite[Theorem 6.1]{Tian_Xu_geometric}\label{thm11}
For each decorated dual graph $\sf\Gamma$, there is a virtual fundamental class $[\ov{\mc M}_{\sf\Gamma}]^{\rm vir} \in H_*(X^n \times \ov{\mc M}_\Gamma; {\mb Q})$ satisfying the following. 

\begin{enumerate}

\item {\bf (Dimension)}	If $\sf\Gamma$ is connected, then the degree of the virtual cycle is 
\beqn
(2-2g) {\rm dim}_{\mb C} X + 2 {\rm deg} {\sf\Gamma} + {\rm dim}_{\mb R} \ov{\mc M}_\Gamma
\eeqn
where ${\rm deg} \sf\Gamma \in {\mb Z}$ is the equivariant first Chern number of the curve class represented by solutions of type $\sf\Gamma$.

\item {\bf (Disconnected Graphs)} Let $\sf\Gamma_1, \ldots, \sf\Gamma_k$ be connected decorated dual graphs and $\sf\Gamma$ be their disjoint union. Then 
\beqn
[\ov{\mc M}_{\sf\Gamma}]^{\rm vir} = \bigoplus_{\alpha=1}^k  [\ov{\mc M}_{\sf\Gamma_\alpha}]^{\rm vir}.
\eeqn

\item {\bf (Cutting Edge)} Let $\sf\Gamma$ be a stable decorated dual graph and let $\sf\Pi$ be the decorated dual graph obtained from $\sf\Gamma$ by shrinking a loop. Then one has
\beqn
(\iota_{\Pi})_* [\ov{\mc M}_{\sf\Pi}]^{\rm vir} = [\ov{\mc M}_{\sf\Gamma}]^{\rm vir} \cap [\ov{\mc M}_\Pi].
\eeqn
Here $\iota_\Pi: \ov{\mc M}_\Pi \hookrightarrow \ov{\mc M}_\Gamma$ is the natural inclusion. 

\item {\bf (Composition)} Let $\sf\Pi$ be the decorated dual graph with one distinguished edge and let $\tilde{\sf\Pi}$ be the decorated dual graph obtained from $\sf\Pi$ by normalization. Let $\Delta \subset X \times X$ be the diagonal. Notice that there is a natural isomorphism $\ov{\mc M}_\Pi \cong \ov{\mc M}{}_{\tilde \Pi}$. Then 
\beqn
[\ov{\mc M}_{\sf\Pi}]^{\rm vir} = [\ov{\mc M}{}_{\tilde{\sf\Pi}}]^{\rm vir} \setminus {\rm PD}[\Delta].
\eeqn
Here we use the slant product
\beqn
\setminus: H_* ( X^{n+2}) \otimes H^*( X \times X) \to H_* (X^n).
\eeqn
\end{enumerate}
\end{thm}

Besides the virtual cycle construction, in this paper we also provide the complete proofs of two important facts used in the construction, namely \cite[Theorem 4.1, Theorem 4.2]{Tian_Xu_geometric} (see Section \ref{section4} for details). 

\subsection{Topological construction of the virtual cycle}\label{subsection12}

The construction of the correlation functions in the CohFT is based on the theory of virtual cycles. In symplectic geometry, the virtual cycle theory arose in mathematicians' efforts in defining Gromov--Witten invariants for general symplectic manifolds. Jun Li and the first named author have their approach in both the algebraic case \cite{Li_Tian_2} and the symplectic case \cite{Li_Tian}. Meanwhile there were also other people's approach such as the method of Kuranishi structure of Fukaya--Ono \cite{Fukaya_Ono}, further developed by Fukaya--Oh--Ohta--Ono \cite{FOOO_Kuranishi}. Until recently, there have been various new developments in the virtual cycle theory, such as the polyfold method of Hofer--Wysocki--Zehnder \cite{HWZ1, HWZ2, HWZ3} and the algebraic topological approach of Pardon \cite{Pardon_virtual} and Abouzaid \cite{Abouzaid_axiomatic}.

Our approach follows the topological nature of Li--Tian's method, which is based on the fact that transversality can also be achieved in the topological category. The concrete treatment, though, looks similar to the Kuranishi approach. The main difference from the Kuranishi approach, besides that we are in the topological rather than the smooth category, is that we directly construct a finite good coordinate system while bypassing the Kuranishi structure.

\subsection{Outline}

In Section \ref{section2} we review basic facts about the moduli space of $r$-spin curves, especially the notion of resolution data. In Section \ref{section3} we review the setup of the gauged Witten equation. In Section \ref{section4} we prove two technical theorems about properties of the gauged Witten equation which were asserted in \cite{Tian_Xu_geometric}. In Section \ref{section5} we describe the compactified moduli spaces. In Section \ref{section6} we recall the abstract framework of the virtual cycle theory. In Section \ref{section7} we state our main theorem in a rigorous form. In Section \ref{section8}---\ref{section10} we construct virtual orbifold atlases on the moduli space of the gauged Witten equation. In Section \ref{section11} we prove the properties of the virtual cycle and hence our main theorem (Theorem \ref{thm11}) follows. 

\subsection{Acknowledgement}

We thank Wei Gu, Melissa Liu, Mauricio Romo, Jake Solomon for stimulating discussions. We thank Alexander Kupers for kindly answering questions about topological transversality.

\section{Local Models of Moduli Spaces of $r$-Spin Curves}\label{section2}

In this section we recall the notions of $r$-spin curves and their moduli spaces. For the purpose of the gluing construction we also set up notations for the local models, especially the notion of resolution data.

\subsection{Marked nodal curves}

Suppose $g, n\geq 0$ and $2g + n \geq 3$. The moduli space of smooth genus $g$ Riemann surfaces with $n$ marked points has a well-known Deligne--Mumford compactification, denoted by $\ov{\mc M}_{g, n}$. It is a compact complex orbifold, which is an effective orbifold except for $(g, n) = (1, 1)$ and $(2, 0)$ (see \cite{Robbin_Salamon_2006} for related facts from an analytical point of view). 

We denote a representative of a point of $\ov{\mc M}_{g, n}$ by $(\Sigma, \vec{\bf z} )$ where $\Sigma$ is a compact smooth or nodal curve of genus $g$, and $\vec{\bf z} = (z_1, \ldots, z_n)$ is an ordered set of marked points which are distinct and disjoint from the nodes. Given such $\Sigma$, let $\V_\Sigma$ be the set of irreducible components, whose elements are denoted by $v$. Let $\E_\Sigma$ be the set of its nodal points, whose elements are denoted by $w_1, \ldots, w_m$. Let $\pi: \tilde  \Sigma \to \Sigma$ be the normalization, which is a possibly disconnected smooth Riemann surface with components $\{ \Sigma_v \ | v \in \V_\Sigma \}$. $\pi$ is generically one-to-one and the preimages of the marked points $z_a$ are still denoted by $z_a$. On the other hand, each node of $\Sigma$ has two preimages and we denote by $\tilde w_a$ to be a preimage of some node. The notion of canonical bundles extends to nodal curves. Indeed, if $\Sigma$ is nodal, then there is an isomorphism 
\beqn
\pi^* K_\Sigma \cong K_{\tilde \Sigma} \otimes \bigotimes_{\tilde w_a} {\mc O}(\tilde w_a).
\eeqn

\subsubsection{Universal unfoldings}

The Deligne--Mumford spaces (as well as the moduli spaces of stable $r$-spin curves) have purely algebraic geometric descriptions. However, to fit in the differential geometric construction, we would like to use a differential geometric description. Here we follow the approach of Robbin--Salamon \cite{Robbin_Salamon_2006} for the Deligne--Mumford space $\ov{\mc M}_{g, n}$ and then modify it for the case of $r$-spin curves.

A {\it holomorphic family} of complex curves consists of open complex manifolds ${\mc U}$, ${\mc V}$ and a proper holomorphic map $\pi: {\mc U} \to {\mc V}$ with relative complex dimension one. It is called a {\it nodal} family if for 
every critical point $p \in {\mc U}$ of $\pi$, there exist holomorphic coordinates $(w_0, w_1, \ldots, w_s)$ around $p$ and holomorphic coordinates around $\pi(p)$ such that locally 
\beqn
\pi(w_0, w_1, \ldots, w_s) = (w_0 w_1, w_2, \ldots, w_s). 
\eeqn
Given $n \geq 0$, an $n$-{\it marked nodal family} consists of a nodal family $\pi: {\mc U} \to {\mc V}$ together with holomorphic sections $Z_1, \ldots, Z_n: {\mc V} \to {\mc U} \setminus {\rm Crit} \pi$ whose images are mutually disjoint. Then for each $b \in {\mc V}$, $(\pi^{-1}(b), Z_1(b), \ldots, Z_n(b))$ is a marked nodal curve. 

Given two $n$-marked nodal families $(\pi_i: {\mc U}_i \to {\mc V}_i, Z_{i,1}, \ldots, Z_{i, n})$, $(i=1,2)$, we abbreviate the data as ${\mc U}_1$ and ${\mc U}_2$. A {\it morphism} from ${\mc U}_1$ to ${\mc U}_2$ is a commutative diagram
\beqn
\vcenter{ \xymatrix{ {\mc U}_1 \ar[r]^{\tilde \varphi} \ar[d]^{\pi_1} & {\mc U}_2 \ar[d]_{\pi_2} \\ 
           {\mc V}_1 \ar[r]^{\varphi} \ar@/^/[u]^{Z_{1,i}}        & {\mc V}_2 \ar@/_/[u]_{{Z_{2,i}} } }},\ i = 1, \ldots, n.
\eeqn

\begin{defn}\label{defn21}
Let $(\Sigma, \vec{\bf z})$ be a marked nodal curve. An {\it unfolding} of $(\Sigma, \vec{\bf z})$ consists of a nodal family $\pi: {\mc U} \to {\mc V}$ with markings $Z_1, \ldots, Z_n$, a base point $b \in {\mc V}$ and an isomorphism 
\beqn
(\Sigma, z_1, \ldots, z_n) \cong ({\mc U}_b, Z_1(b), \ldots, Z_n(b) ).
\eeqn
We can define morphisms of unfoldings as morphisms of marked nodal families that respect the central fibre identifications. We can also define germs of unfoldings and germs of morphisms of unfoldings. An unfolding is {\it universal} if for every other unfolding consisting of a nodal family $\pi': {\mc U}' \to {\mc V}'$, markings $Z_1', \ldots, Z_n'$, a base point $b ' \in {\mc V}'$, and for every isomorphism between central fibres $f: {\mc U}_{b'}' \cong {\mc U}_b$, there is a unique germ of morphism from $({\mc U}', b')$ to $({\mc U}, b)$ whose restriction to ${\mc U}_{b'}'$ coincides with $f$.
\end{defn}

\begin{thm}\cite{Robbin_Salamon_2006} A marked nodal curve $(\Sigma, \vec{\bf z})$ has a unique germ of universal unfolding if and only if it is stable. 
\end{thm}

For any universal unfolding $\pi: {\mc U}\to {\mc V}$, for any automorphism $\gamma$ of $(\Sigma, \vec{\bf z})$, $\gamma$ induces an action on ${\mc U}$ and ${\mc V}$ that preserves the markings. Denote the actions by 
\begin{align*}
&\ \gamma^{\mc U}: {\mc U} \to {\mc U},\ &\ \gamma^{\mc V}: {\mc V} \to {\mc V}.
\end{align*}
Hence $\pi: {\mc U} \to {\mc V}$ becomes an ${\rm Aut}(\Sigma, \vec{\bf z})$-equivariant family. Hence from now on, for all universal unfoldings ${\mc U} \to {\mc V}$ of a stable marked curve $(\Sigma, \vec{\bf z})$, we always assume it is ${\rm Aut}(\Sigma, \vec{\bf z})$-equivariant. 

On the other hand, there is a natural continuous injection 
\beqn
{\mc V}/ {\rm Aut}(\Sigma, \vec{\bf z}) \hookrightarrow \ov{\mc M}{}_{g, n}
\eeqn
hence a local universal unfolding provides a local orbifold chart of the Deligne--Mumford space. It also implies the following facts. Suppose there is an isomorphism
\beqn
\varphi: ({\mc U}_{b_1}, Z_1(b_1), \ldots, Z_n(b_1)) \cong ({\mc U}_{b_2}, Z_1(b_2), \ldots, Z_n(b_2))
\eeqn
where $b_1, b_2 \in {\mc V}$. Then there must be a unique automorphism $\gamma \in {\rm Aut}(\Sigma, \vec{\bf z})$ such that $\gamma^{\mc V}(b_1) = b_2 $ and $\varphi = \gamma^{\mc U}|_{{\mc U}_{b_1}}$.

\subsubsection{Generalized marked curves}

In the study of pseudoholomorphic curves, one usually needs to add marked points to unstable components to obtain a stable domain curve. The newly added marked points are treated as  unordered. It is then convenient to generalize the notion of marked curves. 

\begin{defn}\label{defn23}
A {\it generalized marked curve} is a triple $(\Sigma, \vec{\bf z}, {\bf y})$ where $(\Sigma, \vec{\bf z})$ is marked curve and ${\bf y}$ is a set of points which does not intersect with $\vec{\bf z}$ or nodes. Two generalized marked curves $(\Sigma_1, \vec{\bf z}_1, {\bf y}_1)$ and $(\Sigma_2, \vec{\bf z}_2, {\bf y}_2)$ are {\it isomorphic} if there is an isomorphisms between $(\Sigma_1, \vec{\bf z}_1)$ and $(\Sigma_2, \vec{\bf z}_2)$ that maps ${\bf y}_1$ bijectively onto ${\bf y}_2$. A generalized marked curve is {\it stable} if it becomes a stable marked curve after arbitrarily ordering $\vec{\bf z} \cup {\bf y}$. If the genus of $\Sigma$ is $g$, $\# \vec{\bf z} = n$ and $\# {\bf y} = l$, then we say it is of type $(g, n, l)$.
\end{defn}

Let the moduli space of generalized marked curves of type $(g, n, l)$ be $\ov{\mc M}_{g, n, l}$. It can be viewed as the quotient 
\beqn
\ov{\mc M}_{g, n, l}:= \ov{\mc M}_{g, n+l}/ S_l.
\eeqn
More generally, we can consider the moduli 
\beqn
\ov{\mc M}_{g, n, m_1, \ldots, m_r}:= \ov{\mc M}_{g, n+m_1 + \cdots + m_r}/ S_{m_1} \times \cdots \times S_{m_r}.
\eeqn

From universal unfoldings of stable marked curves one can easily obtain universal unfoldings of stable generalized marked curves. Indeed, let $(\Sigma, \vec{\bf z}, {\bf y})$ be a stable generalized marked curve. Give an arbitrary order of ${\bf y}$, this produces a stable marked curve $(\Sigma, \vec{\bf z} \cup \vec {\bf y})$, hence admits a universal unfolding ${\mc U} \to {\mc V}$ with sections $Z_1, \ldots, Z_n, Y_1, \ldots, Y_l$. We also have the relations between the automorphism groups
\beqn
{\rm Aut}(\Sigma, \vec{\bf z} \cup \vec{\bf y})  \subset {\rm Aut}(\Sigma, \vec{\bf z},  {\bf y})
\eeqn
where the former contains biholomorphic maps that preserves the ordering on ${\bf y}$.

Now for each $h \in S_l$, one can choose a universal unfolding ${\mc U}_h \to {\mc V}_h$ of $(\Sigma, \vec{\bf z} \cup h \vec{\bf y})$. Moreover, we can choose ${\mc U}_h$'s in such a way that the following conditions are satisfied.
\begin{enumerate}
\item If $(\Sigma, \vec{\bf z} \cup h_1 \vec{\bf y})$ and $(\Sigma, \vec{\bf z} \cup h_2 \vec{\bf y})$ are not isomorphic as stable marked curves, then the underlying sets $[{\mc V}_{h_1}]$ and $[{\mc V}_{h_2}]$ are disjoint in $\ov{\mc M}_{g, n+l}$. 

\item If $\gamma: (\Sigma, \vec{\bf z} \cup h_1\vec {\bf y}) \cong (\Sigma, \vec{\bf z} \cup h_2 \vec{\bf y})$ is an isomorphism, then there is an isomorphism of nodal families whose restriction to the central fibre is $\gamma$
\beqn
\vcenter{ \xymatrix{  {\mc U}_{h_1} \ar[r] \ar[d] & {\mc U}_{h_2} \ar[d] \\
                      {\mc V}_{h_1} \ar[r]        &  {\mc V}_{h_2}    }    }.
\eeqn
(The existence of the above commutative diagram on the germ level is guaranteed by the universality of the unfoldings). 
\end{enumerate} 
Then the original ${\mc U} \to {\mc V}$ can be viewed as a universal unfolding of the generalized marked curve, which is ${\rm Aut}(\Sigma, \vec{\bf z}, {\bf y})$-equivariant.

\subsection{Resolution data and gluing parameters}\label{subsection22}

(See \cite[Remark 5.2.3]{Tehrani_Fukaya}) Consider a (possibly) nodal stable generalized marked $ (\Sigma, \vec{\bf z}, {\bf y})$ of type $(g, n, l)$. Its normalization is another generalized marked $(\tilde \Sigma, \vec{\bf z}, {\bf w}, {\bf y})$ of type $(g, n, m, l)$ where $m/2$ is the number of nodes of ${\mc C}$. Notice that the preimages of the nodes are unordered. $\tilde \Sigma$ is not necessarily connected. For every vertex $v \in \V_\Sigma$, denote the smooth generalized marked curve by 
\beqn
(\tilde \Sigma_v, \vec{\bf z}_v, {\bf w}_v, {\bf y}_v).
\eeqn
which is of type $(n_v, m_v, l_v)$ (see Definition \ref{defn23} and discussion afterwards). Then there exist universal unfoldings $\pi_v: {\mc U}_v \to  {\mc V}_v$ where we suppressed the data of an identification of ${\mc C}_v$ with the central fibre. The product of these unfoldings parametrizes deformations of ${\mc C}$ that do not resolve the nodes. Moreover, the automorphism group $\Gammait_{\mc C}$ of $(\Sigma_{\mc C}, \vec{\bf z}, {\bf y})$ acts on the normalization 
\beqn
\bigsqcup_{v \in \V_{\mc C}} ( \tilde \Sigma_v, \vec{\bf z}_v, {\bf w}_v, {\bf y}_v).
\eeqn
Hence can shrink ${\mc V}_v$ so that $\Gammait_{\mc C}$ also acts on the disjoint union of the universal unfoldings 
\beq\label{eqn21}
\bigsqcup_{v\in \V_{\mc C}} {\mc U}_v  \to \bigsqcup_{v \in \V_{\mc C}} {\mc V}_v.
\eeq

For the purpose of gluing, we need to have more explicit description of how to resolve the nodes. 

\begin{lemma}\label{lemma24} \cite[Lemma 5.24]{Tehrani_Fukaya} 
For all components $v \in \V_{\Sigma}$, there exist a collection of universal unfoldings $\pi_v: {\mc U}_v \to {\mc V}_v$, a collection of open subsets ${\mc N}_v \subset {\mc U}_v$, which we call {\it nodal neighborhoods}, and a collection of holomorphic functions $\theta_v: {\mc N}_v \to {\mb C}$ for all $v \in \V_\Sigma$ satisfying the following conditions. 
 
\begin{enumerate}
\item ${\mc N}_v$ is the disjoint union of open neighborhoods of points in ${\bf w}_v$, denoted by 
\beqn
{\mc N}_v = \bigcup_{\tilde w \in {\bf w}_v} {\mc N}_{v, \tilde w}.
\eeqn

\item For any automorphism $\gamma$ of ${\mc C}$\footnote{Remember $\gamma$ acts on the set of irreducible components.} there is an isomorphism
\beqn
\varphi_\gamma: {\mc U}_v \cong {\mc U}_{\gamma v}
\eeqn
such that $\varphi_{\gamma} \circ \varphi_{\gamma'} = \varphi_{\gamma \gamma'}$.

\item For any automorphism $\gamma$ of ${\mc C}$ and $\tilde w \in {\bf w}_v$ \footnote{Remember that $\gamma$ acts on the set of preimages of nodes.} we have 
\beqn
\varphi_\gamma ( {\mc N}_{v, \tilde w}) = {\mc N}_{\gamma v, \gamma \tilde w}.
\eeqn

\item Denote the restriction of $\theta_v$ to ${\mc N}_{v, \tilde w}$ by $\theta_{v, \tilde w}$. Then the restriction of $\theta_{v, \tilde w}$ to each fibre is one-to-one and maps $\tilde w$ to $0\in {\mb C}$. Moreover, 
\beq\label{eqn22}
\theta_{\gamma v, \gamma \tilde w} \circ \varphi_\gamma = \theta_{v, \tilde w}.
\eeq

\item For any node $w$ of ${\mc C}$ with preimages $\tilde w_-, \tilde w_+$ belonging to components $v_-, v_+$ (which could be equal if $w$ is non-separating), for every automorphism $\gamma$ of ${\mc C}$ of order $m_\gamma$, there is an $m_\gamma$-th root of unity $\mu_{\gamma, w}$ such that 
\beqn
\left( \theta_{\gamma v_-, \gamma \tilde w_-} \circ \varphi_\gamma\right) \left( \theta_{\gamma v_+, \gamma \tilde w_+} \circ \varphi_\gamma \right) = \mu_{\gamma, w} \theta_{v_-, \tilde w_-} \theta_{v_+, \tilde w_+}.
\eeqn
\end{enumerate}
\end{lemma}
\begin{proof}
Only \eqref{eqn22} is not stated in \cite[Lemma 5.24]{Tehrani_Fukaya}, but it is a consequence of the construction in their proof. 
\end{proof}

The functions $\theta_{v, w}$ are viewed as fibrewise coordinates around the nodes. Choosing the data $\pi_v: {\mc U}_v \to {\mc V}_v$, $\theta_v: {\mc N}_v \to {\mb C}$ and $\mu_{\gamma, w}$ that satisfy conditions in Lemma \ref{lemma24}, we can produce a special type of universal unfolding of ${\mc C}$. Define 
\begin{align*}
&\ {\mc V}_{\rm def} = \prod_{v\in \V_\Sigma } {\mc V}_v,\ &\ {\mc V}_{\rm res} = \prod_{w \in \E(\Sigma)} {\mb C}.
\end{align*}
Then the automorphism group $\Gammait_{\mc C}$ acts on ${\mc V}_{\rm def}$. It also acts on ${\mc V}_{\rm res}$ in the following way. The coordinates of any $\zeta \in {\mc V}_{\rm res}$ are denoted by $\zeta_w$. For each $\gamma \in \Gammait_{\mc C}$ and $\zeta \in {\mc V}_{\rm res}$, define 
\beqn
(\gamma \zeta)_w = \mu_{\gamma, w} \zeta_{\gamma w}.
\eeqn
Denote variables in ${\mc V}_{\rm def}$ by $\eta$ and variables in ${\mc V}_{\rm res}$ by $\zeta$, call them {\it deformation parameters} and {\it gluing parameters}. We construct a universal unfolding of $(\Sigma, \vec{\bf z}, {\bf y})$ as follows. 

For each deformation parameter $\eta$, one has the corresponding fibre $(\Sigma_\eta, \vec{\bf z}_\eta, {\bf y}_\eta)$ obtained by identifying the preimages of the nodes. We may regard $\Sigma_\eta$ as a fiber of \eqref{eqn21}. Denote
\beqn
N_{v, \eta}:= {\mc N}_v \cap \Sigma_\eta,\ N_{v, \tilde w, \eta}:= {\mc N}_{v, \tilde w} \cap \Sigma_\eta.
\eeqn
Now for each (small) gluing parameter $\zeta$, define an object $(\Sigma_{\eta, \zeta}, \vec{\bf z}_{\eta, \zeta}, {\bf y}_{\eta, \zeta})$ as follows. Consider a node $w$, whose preimages in the normalization, $\tilde w_-, \tilde w_+$, belong to components $v_-$, $v_+$ respectively. If the corresponding gluing parameter $\zeta_w$ is nonzero, then replace the union 
\beqn
N_{v_-, \tilde w_-, \eta} \cup N_{v_+, \tilde w_+, \eta}
\eeqn
by the annulus 
\beqn
N_{w, \eta, \zeta}:= \Big\{ (\xi_-, \xi_+) \in N_{v_-, \tilde w_-, \eta} \times N_{v_+, \tilde w_+, \eta} \ |\ \theta_{v_-, \tilde w_-}(\xi_-) \theta_{v_+, \tilde w_+}(\xi_+) = \zeta_w \Big\}.
\eeqn
This provides a nodal (or smooth) curve 
\beqn
\Sigma_{\eta, \zeta}:= \Big( \Sigma_\eta \setminus \bigcup_{\zeta_w \neq 0} N_{v_-, \tilde w_-, \eta} \cup N_{v_+, \tilde w_+, \eta} \Big) \cup \bigcup_{\zeta_w \neq 0} N_{w, \eta, \zeta}.
\eeqn
The positions of $\vec{\bf z}_{\eta, \zeta}$ and ${\bf y}_{\eta, \zeta}$ are the same as $\vec{\bf z}_\eta$ and ${\bf y}_\eta$. 

It is a well-known fact that the union of all $\Sigma_{\eta, \zeta}$ forms a universal unfolding ${\mc U} \to {\mc V}$ of $(\Sigma, \vec{\bf z}, {\bf y})$. We then define the union of all ${\mc N}_{v, \tilde w}$ and the neck region $N_{w, \eta, \zeta}$ as the {\it thin} part, denoted by 
\beqn
{\mc U}^{\rm thin} \subset {\mc U};
\eeqn
the closure of the complement of ${\mc U}^{\rm thin}$ is called the {\it thick} part, denoted by ${\mc U}^{\rm thick}$. Then 
\beqn
{\mc U} = {\mc U}^{\rm thin} \cup {\mc U}^{\rm thick}.
\eeqn
From the construction we know that ${\mc V}$ is contractible. Hence there exists a smooth trivialization 
\beqn
{\mc U}^{\rm thick} \cong {\mc V} \times \Sigma^{\rm thick},\ {\rm where}\ \Sigma^{\rm thick}:= \Sigma \cap {\mc U}^{\rm thick}.
\eeqn

\begin{defn}\label{defn25} {\rm (Resolution data for stable curves)} Let $(\Sigma, \vec{\bf z}, {\bf y})$ be a stable generalized marked curve. A {\rm resolution data} of $(\Sigma, \vec{\bf z}, {\bf y})$ consists of the following objects.
\begin{enumerate}

\item A universal unfolding $\tilde {\mc U} \to \tilde {\mc V}$ of its normalization of the form \eqref{eqn21}.

\item A collection of holomorphic functions $\zeta_v: {\mc N}_v \to {\mb C}$ and $\mu_{\gamma, w}$ as in Lemma \ref{lemma24}

\item A smooth trivialization 
\beqn
\tilde {\mc U} \setminus \tilde {\mc N} \cong \tilde {\mc V}  \times\big( \Sigma \setminus (\tilde {\mc U}_b \cap \tilde {\mc N} ) \big).
\eeqn
which is holomorphic near the boundary of $\tilde {\mc N}$. 

\end{enumerate}
\end{defn}

A resolution datum is a much more enhanced object than a universal unfolding. We actually have shown the existence of resolution data of stable generalized marked curves. As we have seen, a resolution datum provides the following objects.
\begin{enumerate}

\item A universal unfolding ${\mc U} \to {\mc V}$ of $(\Sigma, \vec{\bf z}, {\bf y})$. 

\item Fibrewise holomorphic coordinates near nodes and markings that are invariant under automorphisms.

\item Smooth trivializations of the thick part of ${\mc U}$. 

\end{enumerate}

\subsection{$r$-spin curves}\label{subsection23}

\subsubsection{Orbifold curves and orbifold bundles}

A one-dimensional complex orbifold is called an orbifold Riemann surface or an orbifold curve. In this paper, we impose the following conditions and conventions for orbifold curves. 
\begin{enumerate}

\item Orbifold curves are effective orbifolds with finitely many orbifold points. 

\item We always assume that an orbifold curve is marked, and denoted by $(\Sigma, \vec{\bf z})$. Moreover, the set of orbifold points are contained in the set of markings $\vec{\bf z}$, however, each $z_a \in \vec{\bf z}$ may not be a strict orbifold point.
\end{enumerate}

Near each $z_a \in \vec{\bf z}$ there exists an orbifold chart of the form 
\beqn
(U_a, \Gammait_a) \cong ({\mb D}_\epsilon, {\mb Z}_{r_a}),\ r_a \geq 1.
\eeqn
The ${\mb Z}_{r_a}$ action on ${\mb D}_\epsilon$ is the standard action, by viewing ${\mb Z}_{r_a}$ as a subgroup of $U(1)$. 

The notion of nodal orbifold curves is more complicated. At a node $w$ the nodal orbifold curve has a chart of the form
\beqn
(U_w, \Gammait_w) \cong \big( \{ (\xi_-, \xi_+) \in {\mb D}_\epsilon^2\ |\ \xi_- \xi_+ = 0 \big\}, {\mb Z}_{r_w} \big),\ r_w \geq 1,
\eeqn
where the ${\mb Z}_{r_w}$-action on $U_w$ is given by 
\beqn
\gamma (\xi_-, \xi_+) = (\gamma^{-1} \xi_-, \gamma \xi_+).
\eeqn

Given a smooth orbifold curve $(\Sigma, \vec{\bf z})$, an orbifold line bundle is $L$ is a complex orbifold with a holomorphic map $\pi: L \to \Sigma$ which, over non-orbifold points has local trivializations as an ordinary holomorphic line bundle, while over an orbifold chart $(U_a, \Gammait_a)$ there is a chart of the form 
\beqn
(\tilde U_a, \Gammait_a) \cong ( U_a \times {\mb C}, {\mb Z}_{r_a})
\eeqn
where the ${\mb Z}_{r_a}$-action on the ${\mb C}$-factor is given by a weight $n_a$, i.e. 
\beqn
\gamma( z, t) = (\gamma z, \gamma^{n_a} t).
\eeqn
We call the element $e^{2\pi q {\bf i}} \in {\mb Z}_r$, where $q = \frac{ n_a}{r_a}$ the {\it monodromy} of the line bundle at $z_a$. When $\Sigma$ is nodal, we require that an orbifold line bundle $L \to \Sigma$ has local charts at a node $w$ of the form 
\beq\label{eqn23}
(\tilde U_w, \Gammait_w) \cong \big( \{ (\xi_-, \xi_+, t) \in {\mb D}_\epsilon \times {\mb D}_\epsilon \times {\mb C}\ |\ \xi_- \xi_+ = 0\}, {\mb Z}_{r_w} \big)
\eeq
where the ${\mb C}$-factor still has the linear action by the local group ${\mb Z}_{r_w}$. Then with respect to the normalization map $\pi: \tilde \Sigma \to \Sigma$, the pull-back bundle
\beqn
\tilde L:= \pi^* L \to \tilde \Sigma
\eeqn
is an orbifold line bundle whose monodromies at two preimages of a node are opposite. 

\begin{defn}{\rm (Log-canonical bundle)}
Let $(\Sigma, \vec{\bf z})$ be a marked smooth or nodal curve with $n$ markings $z_1, \ldots, z_n$. Its {\it log-canonical bundle} is the holomorphic line bundle
\beqn
K_{\Sigma, {\rm log}} = K_\Sigma \otimes \bigotimes_{z_a \in \Sigma_v} {\mc O}( z_a).
\eeqn

For an orbifold curve $(\Sigma, \vec{\bf z})$, its log-canonical bundle $K_{\Sigma, {\rm log}} \to \Sigma$ is defined to be the pullback of the log-canonical bundle of the underlying marked coarse curve.
\end{defn}

\subsubsection{$r$-spin curves}

We recall the definition of $r$-spin curves.

\begin{defn}\label{defn_r_spin} {\rm ($r$-spin curve)} Let $r$ be a positive integer.
\begin{enumerate}
\item A smooth or nodal $r$-spin curve of type $(g, n)$ is denoted by 
\beqn
{\mc C} = (\Sigma_{\mc C}, \vec{\bf z}_{\mc C}, L_{\mc C}, \varphi_{\mc C})
\eeqn
where $(\Sigma_{\mc C}, \vec{\bf z}_{\mc C})$ is an $n$-marked smooth or nodal orbifold curve, $L_{\mc C} \to \Sigma_{\mc C}$ is a holomorphic orbifold line bundle, and $\varphi_{\mc C}$ is an isomorphism 
\beqn
\varphi_{\mc C}: L_{\mc C}^{\otimes r} \cong K_{{\mc C}, {\rm log}}:= K_{\Sigma_{{\mc C}}, {\rm log}}.
\eeqn

\item Let ${\mc C}_i = (\Sigma_{{\mc C}_i}, \vec{\bf z}_{{\mc C}_i}, L_{{\mc C}_i}, \varphi_{{\mc C}_i})$, $i = 1, 2$ be two $r$-spin curves. An {\it isomorphism} from ${\mc C}_1$ to ${\mc C}_2$ consists of an isomorphism of orbifold bundles (represented by the following commutative diagram)
\beqn
\vcenter{ \xymatrix{ L_{ {\mc C}_1} \ar[r]^{\tilde\rho} \ar[d]_{\pi_1} & L_{{\mc C}_2} \ar[d]^{\pi_2} \\
             \Sigma_{{\mc C}_1}   \ar[r]^\rho    & \Sigma_{{\mc C}_2}     } }
\eeqn
such that the following induced diagram commutes
\beq\label{eqn24}
\vcenter{ \xymatrix{ L_{{\mc C}_1}^{\otimes r} \ar[r]^{(\tilde\rho)^{\otimes r}}  \ar[d]_{\varphi_{{\mc C}_1}} & L_{{\mc C}_2}^{\otimes r} \ar[d]^{\varphi_{{\mc C}_2}   } \\ 
             K_{{\mc C}_1, {\rm log}}\ar[r]^{\rho} & K_{{\mc C}_2, {\rm log}}} }.
\eeq
\end{enumerate}

\end{defn}

\begin{rem} {\rm (Minimality of the local group)}
We require that the orbifold structures at the markings or nodes of $\Sigma_{\mc C}$ are {\it minimal} in the following sense. Take a local chart of $L_{\mc C}$ at a marking $z_a$
\beqn
({\mb D}_\epsilon \times {\mb C}, {\mb Z}_{r_a})
\eeqn
where the action of ${\mb Z}_{r_a}$ on the ${\mb C}$-factor has weight $n_a$. Then the isomorphism $\varphi_{\mc C}: L_{\mc C}^{\otimes r} \cong K_{{\mc C}, {\rm log}}$ implies that $rm_a/ r_a$ is an integer, hence there exists an integer $m_a \in \{0, 1, \ldots, r-1\}$ such that
\beqn
\frac{ n_a}{r_a} = \frac{m_a}{r}.
\eeqn
We require that $n_a$ and $r_a$ are coprime. Similar requirement is imposed for nodes. In other words, the group generated by the monodromies of $L_{\mc C}$ at a marking or a node is the same as the local group of the orbifold curve at that marking or node. In particular, when the monodromy of $L_{\mc C}$ is trivial, the marking or node is not a strict orbifold point. 
\end{rem}

\subsubsection{Infinite cylinders}

Denote the infinite cylinder by 
\beqn
\Thetait:= {\mb R} \times S^1.
\eeqn
Let the standard cylindrical coordinates be $s + {\bf i}t$ where we also regard $t \in [0, \pi]$. Hence there is a complex coordinate $z = e^{s + {\bf i} t}$ which identifies $\Thetait$ with ${\mb C}^*$. For any cyclic group ${\mb Z}_k$, it acts on $\Thetait$ by rotating the $t$-coordinate. The orbifold $\Thetait/ {\mb Z}_k$ is still isomorphic to $\Thetait$, hence we may either regard the infinite cylinder as a smooth object, or regard it as an orbifold with local group ${\mb Z}_k$ at $\pm \infty$ for certain $k \geq 1$. When we are in the latter perspective, we call it an orbifold cylinder.

The isomorphism classes of $r$-spin structures on an infinite are classified by their monodromies. Indeed, for any $m \in \{0, 1, \ldots, r-1\}$, there is an orbifold line bundle $L_m \to \Thetait$ whose monodromy at $-\infty$ is $e^{2\pi q {\bf i} }$ where $q = \frac{m}{r}$, together with a holomorphic isomorphism $\varphi: L_m^{\otimes r} \cong K_\Thetait$. If we regard $z$ as the local coordinate near $-\infty$, then over each contractible open subset of $\Thetait$, there is a holomorphic section $e$ of $L_m$ such that
\beqn
\varphi(e^{\otimes r}) = z^m \frac{dz}{z}.
\eeqn
The monodromy at $+\infty$ of $L_m$ is $e^{- 2\pi q {\bf i}}$.

\subsubsection{Automorphisms}

Given an $r$-spin curve ${\mc C}$, its automorphisms form a group, denoted by ${\rm Aut}{\mc C}$. ${\mc C}$ is called {\it stable} if ${\rm Aut}{\mc C}$ is finite. Every automorphism of ${\mc C}$ induces an automorphism of the underlying orbifold nodal curve $\Sigma_{\mc C}$. Let ${\rm Aut}_0 (L_{\mc C}) \subset {\rm Aut} {\mc C}$ be the subgroup of automorphisms which induce the identity of $(\Sigma_{\mc C}, \vec{\bf z})$, namely the set of bundle automorphisms. Then there is a tautological exact sequence of groups
\beqn
\xymatrix{ 1 \ar[r] & {\rm Aut}_0  L_{\mc C}  \ar[r] & {\rm Aut}{\mc C} \ar[r] & {\rm Aut} (\Sigma_{\mc C}, \vec{\bf z}).}
\eeqn
It is easy to see that when ${\mc C}$ is smooth, ${\rm Aut}_0 L_{\mc C} \cong {\mb Z}_r$ (see Definition \ref{defn_r_spin} and \eqref{eqn24}). More generally, suppose ${\mc C}$ is not smooth. Then for every irreducible component $v \in \V_{\mc C}$ and every preimage $\tilde w \in \tilde  \Sigma_v$ of a node, there is a restriction map 
\beqn
r_{\tilde w}: {\rm Aut}_0 L_{{\mc C}_v}  \to {\mb Z}_r.
\eeqn
as the automorphism of the fibre of $L_{{\mc C}_v}$ at $\tilde w$. For each node $w \in \E(\Sigma_{\mc C})$, let their preimages in the normalization be $\tilde w_-$ and $ \tilde w_+$ which belong to components $v_-$ and $v_+$ respectively. Then define 
\begin{align*}
&\ r: \prod_{v\in \V_{\mc C}} {\rm Aut}_0 L_{\tilde {\mc C}_v} \to \prod_{w \in \E(\Sigma_{\mc C})} {\mb Z}_r / {\mb Z}_{r_w},\ &\ \big( \gamma_v \big)_v \mapsto \big( r_{\tilde w_+}(\gamma_{v_-} ) r_{\tilde w_-}(\gamma_{v_-})^{-1}\big)_w
\end{align*}
where ${\mb Z}_{r_w}$ is the local group at the node $w$. Then there is an exact sequence
\beq\label{eqn25}
\xymatrix{ 1 \ar[r] & {\rm Aut}_0 L_{\mc C} \ar[r] & \displaystyle \prod_{v \in \V_{\mc C}} {\rm Aut}_0 L_{{\mc C}_v} \ar[r] & \displaystyle \prod_{w \in \E_{\mc C}}  {\mb Z}_r}/ {\mb Z}_{r_w}.
\eeq


\subsubsection{Generalized stable $r$-spin curves}

As in Definition \ref{defn23}, a {\it generalized $r$-spin curve} is an object 
\beqn
{\mc C} = (\Sigma_{\mc C}, L_{\mc C}, \varphi_{\mc C}, \vec{\bf z}, {\bf y})
\eeqn
where the object without ${\bf y}$ is an $r$-spin curve and ${\bf y}$ is an unordered set of points on $\Sigma_{\mc C}$ that are disjoint from special points. If the genus of $\Sigma_{\mc C}$ is $g$, $\# \vec{\bf z} = n$ and $\# {\bf y} = l$, then we call it a generalized $r$-spin curve of type $(g, n, l)$. 

\subsection{Unfoldings of $r$-spin curves}\label{subsection24}

We do not define the notion of unfoldings of (generalized) $r$-spin curves in the algebraic-geometric fashion. Instead, we construct universal unfoldings of stable $r$-spin curves from universal unfoldings of their underlying coarse curves. Let ${\mc C} = (\Sigma_{\mc C}, L_{\mc C}, \varphi_{\mc C}, \vec{\bf z}, {\bf y})$ be a stable generalized $r$-spin curve. Consider the underlying coarse curve, denoted by $(\Sigma{}_{\mc C},\vec{\bf z}, {\bf y})$, which is a stable generalized marked curve. Given a resolution datum ${\bm r}$ (see Definition \ref{defn25}), which induces a universal unfolding $\ubar {\mc U} \to {\mc V} \cong {\mc V}_{\rm def}\times {\mc V}_{\rm res}$. We define an orbifold ${\mc U}$ out of $\ubar {\mc U}$ as follows. Indeed, we can decompose the thin part $\ubar{\mc U}^{\rm thin}$ as 
\beqn
\ubar{\mc U}^{\rm thin} = \bigsqcup_{a = 1}^n \ubar {\mc N}_{z_a} \sqcup \bigsqcup_{b=1}^m \ubar {\mc N}_{w_b}.
\eeqn
The resolution data provides holomorphic identifications
\begin{align*}
&\ \ubar {\mc N}_{z_a} \cong {\mb D}_r \times {\mc V},\ &\ \ubar {\mc N}_{w_b} \cong \big\{ (\xi_-, \xi_+) \in {\mb D}_r^2\ |\ |\xi_-||\xi_+| \leq \epsilon \big\} \times {\mc V}_{\rm def}.
\end{align*}
We glue in orbifold charts ${\mc N}_{z_a}$ and ${\mc N}_{w_b}$ as follows. Suppose the local group of $z_a$ is ${\mb Z}_{r_a}$. Then ${\mc N}_{z_a}$ is just the $r_a$-fold cover of ${\mb D}_r$ multiplying ${\mc V}$. Suppose the local group of $w_b$ is ${\mb Z}_{r_b}$. Then define 
\beqn
{\mc N}_{w_b} = \big\{ ( \tilde \xi_-, \tilde \xi_+) \in {\mb D}_{r}^2\ |\ |\tilde \xi_-||\tilde \xi_+| \leq \epsilon \big\} \times {\mc V}_{\rm def}.
\eeqn
Define the ${\mb Z}_{r_b}$-action on ${\mc N}_{w_b}$ by 
\beqn
\gamma( \tilde \xi_-, \tilde \xi_+, \eta) = ( \gamma^{-1} \tilde \xi_-, \gamma \tilde \xi_+, \eta).
\eeqn
Hence we obtained a complex orbifold ${\mc U}$ together with a holomorphic map ${\mc U} \to {\mc V}$. The thick-thin decomposition on $\ubar {\mc U}$ induce a thick-thin decomposition 
\beq\label{eqn26}
{\mc U} = {\mc U}^{\rm thick} \cup {\mc U}^{\rm thin}.
\eeq

Now we construct a holomorphic orbifold line bundle ${\mc L}_{\mc U} \to {\mc U}$ and an isomorphism 
\beqn
\varphi_{\mc C}: {\mc L}_{\mc U}^{\otimes r} \to K_{{\mc U} / {\mc V}, {\rm log}}.
\eeqn
Indeed, the central fibre is already equipped with the line bundle $L_{\mc C}$ which is an $r$-th root of the log-canonical bundle of the central fibre. Choose a collection of coordinate charts $\{(W_\beta, w_\beta)\}$ on $\Sigma_{\mc C}$ which cover the thick part $\Sigma_{\mc C}^{\rm thick}$. Moreover, assume that for each two charts in this collection, the overlap is contractible. The transition function of $\omega_{\mc C}$ are holomorphic functions 
\beqn
g_{\beta \beta'}: W_\beta \cap W_{\beta'} \to {\mb C}^*
\eeqn
while the transition functions of $L_{\mc C}$ are holomorphic functions 
\beqn
h_{\beta \beta'}: W_\beta \cap W_{\beta'} \to {\mb C}^*
\eeqn
such that $g_{\beta \beta'} = h_{\beta \beta'}^r$. The smooth trivialization \eqref{eqn26} makes $w_\beta$ smooth coordinates on each fibre, but not necessarily holomorphic. But still there are unique sections of the canonical bundle of each fibre that over each $W_\beta$ of the form 
\beqn
dw_\beta + \epsilon d\bar w_\beta
\eeqn
where $\epsilon$ is a smooth function defined on $W_\beta \times {\mc V}$. Hence $g_{\beta \beta'}$ extends to a smooth function on $(W_\beta \cap W_{\beta'}) \times {\mc V}$ that is fibrewise holomorphic. 

The bundle $L_{\mc C}$ also extends to a smooth complex line bundle over ${\mc U}^{\rm thick}$ via \eqref{eqn26}. Over $W_{\beta} \cap W_{\beta'}$, we can take $r$-th root of the transition function $g_{\beta \beta'}$ which is fibrewise holomorphic. Since the overlap is contractible, there is a unique $r$-th root that continuously extends the value on the central fibre. Therefore we obtain a fibrewise holomorphic line bundle ${\mc L}$ over the thick part. It is also naturally a holomorphic bundle since the canonical bundle is holomorphic over the family. 

On the other hand, one can define the bundle over the thin part. Indeed, take a local chart of $L_{\mc C}$ near a node $w$ by 
\beqn
\big\{ (\xi_-, \xi_+, t) \in {\mb D}_r \times {\mb D}_r \times {\mb C}\ |\ \xi_- \xi_+ = 0 \big\}
\eeqn
where the local group is ${\mb Z}_{r_w}$. Then gluing a bundle chart over ${\mc U}$ by 
\beqn
\big\{ (\xi_- , \xi_+, t) \in {\mb D}_r \times {\mb D}_r \times {\mb C}\ |\ |\xi_-| |\xi_+| < \epsilon \big\}
\eeqn
where the action of ${\mb Z}_{r_w}$ on the three coordinates has the same weights as before. This defines an orbifold line bundle ${\mc L}_{\mc U} \to {\mc U}$. The isomorphism $\varphi_{{\mc U}}: {\mc L}_{\mc U}^{\otimes r} \to K_{{\mc U}/ {\mc V}, {\rm log}}$ is immediate.

\begin{rem}
As in the case of universal unfoldings of stable marked curves, after appropriate shrinkings of the family ${\mc U} \to {\mc V}$ the automorphism group ${\rm Aut}{\mc C}$ also acts on this family. More precisely, there is a $\Gamma_{\mc C}$-action on the line bundle ${\mc L}_{\mc U}$ that descends to actions on ${\mc U}$ and ${\mc V}$. We only remark on the somewhat unusual way of the action of ${\rm Aut}_0 L_{\mc C} \subset {\rm Aut}{\mc C}$ on this family. Given an element $\gamma \in {\rm Aut}_0 L_{\mc C}$ which over each irreducible component $v \in \V_{\mc C}$ is an element $\gamma_v \in {\mb Z}_r$. For each node $w \in \E_{\mc C}$ where a local chart of ${\mc L}_{\mc U}$ has coordinates $(\xi_-, \xi_+, t)$ and the ${\mb Z}_{r_w}$-action has weights $-1, 1, n_w$. Then we define
\beqn
\gamma( \xi_-, \xi_+, t) = ( \xi_-, \gamma_{v_+} \gamma_{v_-}^{-1} \xi_+, \gamma_{v_-} t ) = \gamma_{v_+} \gamma_{v_-}^{-1} ( \gamma_{v_+} \gamma_{v_-}^{-1} \xi_-, \xi_+, \gamma_{v_+} t). 
\eeqn
Notice that by the exact sequence \eqref{eqn25}, $\gamma_w:= \gamma_{v_+} \gamma_{v_-}^{-1} \in {\mb Z}_{r_w}$. This is a well-defined map on the total space of ${\mc L}_{\mc U}$. Notice that this action will move the fibre over the gluing parameter $\zeta_w = \xi_- \xi_+$ to the fibre over another gluing parameter $\gamma_{v_+} \gamma_{v_-}^{-1}  \zeta_w$.

\end{rem}

Because of the above construction, a {\it resolution datum} of a stable generalized $r$-spin curve is identified with a resolution datum of its underlying coarse curve. 

The above construction also provides a way to define the topology on the moduli space of stable $r$-spin curves.

Now given a stable generalized $r$-spin curve ${\mc C}$ and a resolution data ${\bm r}_{\mc C}$ which contains the fibration ${\mc U} \to {\mc V}$, there is a smooth part 
\beqn
{\mc U}^* \subset {\mc U}
\eeqn
which is the complement of the orbifold markings and the nodes.

\subsubsection{Bundles over $r$-spin curves}

We would like to include other principal bundles over $r$-spin curves. Let $K$ be a connected compact Lie group. In this paper, a principal $K$-bundle over an $r$-spin curve ${\mc C}$ always means a smooth principal $K$-bundle over the punctured smooth Riemann surface $\Sigma_{\mc C}^*$. However we still denote it as $P \to {\mc C}$.

\begin{defn} ({\it Resolution data})
Let ${\mc C}$ be a generalized $r$-spin curve and $P \to {\mc C}$ be a smooth $K$-bundle. A {\it  resolution data} of $({\mc C}, P)$, denoted by ${\bm r} = ({\bm r}_{\mc C}, {\bm r}_{\mc P})$, consists of a resolution data ${\bm r}_{\mc C}$ of ${\mc C}$ and a resolution data ${\bm r}_P$, where the latter means smooth $K$-bundle ${\mc P} \to {\mc U}^*$ and an isomorphism of $K$-bundles
\beqn
\vcenter{ \xymatrix{ P \ar[r] \ar[d] & {\mc P}_b \ar[d] \\
           \Sigma_{\mc C}^* \ar[r]  & {\mc U}_b^*} },
\eeqn
a smooth trivialization over the thin part,
\beqn
t_P^{\rm thin}: {\mc P}^{\rm thin} = {\mc P}|_{{\mc U}^{\rm thin}} \cong {\mc U}^{\rm thin} \times K,
\eeqn
and a trivialization over the thick part
\beqn
t_P^{\rm thick}: {\mc P}^{\rm thick} = {\mc P}|_{{\mc U}^{\rm thick}} \cong {\mc V} \times P|_{\Sigma_{\mc C}^{\rm thick}}.
\eeqn
Moreover, they satisfy the following conditions. For any automorphism $(\gamma, \gamma_P) \in \Gammait = {\rm Aut}({\mc C}, P)$ where $\gamma: {\mc C} \to {\mc C}$ is an automorphism of the generalized $r$-spin curve ${\mc C}$ inducing an action on the family ${\mc U}$, and $\gamma_P: {\mc P} \to {\mc P}$ is a smooth $K$-bundle isomorphism covering $\gamma: {\mc U} \to {\mc U}$. Then we require that
\begin{align}\label{eqn27}
&\ \vcenter{ \xymatrix{ {\mc P}^{\rm thin} \ar[r]^-{\gamma_P} \ar[d]_{t_P^{\rm thin}} & {\mc P}^{\rm thin}\ar[d]^{t_P^{\rm thin}} \\
           {\mc U}^{\rm thin} \times K \ar[r]^-\gamma &   {\mc U}^{\rm thin} \times K}   },\ &\ \vcenter{ \xymatrix{  {\mc P}^{\rm thick} \ar[r]^-{\gamma_P} \ar[d]_{t_P^{\rm thick}} &  {\mc P}^{\rm thick} \ar[d]^{ t_P^{\rm thick} } \\
                       {\mc V} \times P|_{\Sigma_{\mc C}^{\rm thick}} \ar[r]^\gamma & {\mc V} \times P|_{\Sigma_{\mc C}^{\rm thick}}   }    }.
                       \end{align}
\end{defn}

\section{Vortex Equation and Gauged Witten Equation}\label{section3}

In this section we briefly recall the setup of the gauged Witten equation which was already used in the companion paper \cite{Tian_Xu_geometric}.

\subsection{The GLSM space}

We recall the definition of GLSM space used in \cite{Tian_Xu_geometric}.

\begin{defn}
A {\bf GLSM space} is a quadruple $(V, G, W, \mu)$ where
\begin{enumerate}

\item $V$ is a K\"ahler manifold with a holomorphic ${\mb C}^*$-action (called the {\it R-symmetry}). 

\item $G$ is a connected complex reductive Lie group acting holomorphically on $V$ such that: a) the $G$-action commutes with the R-symmetry; b) the restriction of the $G$-action to its maximal compact subgroup $K\subset G$ is Hamiltonian. 

\item $W: V \to {\mb C}$ is a $G$-invariant holomorphic function and is homogeneous of degree $r \geq 1$ with respect to the R-symmetry, namely
\beqn
W(\xi x) = \xi^r W(x),\ \forall \xi \in {\mb C}^*,\ x\in V.
\eeqn

\item $\mu: V \to {\mf k}^*$ is a moment map for the $K$-action with $0$ being a regular value.
\end{enumerate}
\end{defn}

Since the $G$-action commutes with the ${\mb C}^*$-action, they induce an action by $\ubar G = G \times {\mb C}^*$ whose maximal compact subgroup is $\ubar K:= K \times U(1)$. For any vector $\xi$ in ${\mf k}$, $\mf{u}(1)$, or the Lie algebra $\ubar {\mf k}$ of $\ubar K$, let 
\beqn
{\mc X}_\xi \in \Gammait(TV)
\eeqn
denote the infinitesimal action of $\xi$. 

Given a GLSM space $(V, G, W, \mu)$, define the semi-stable locus of $V$ as 
\beqn
V^{\rm ss}:= G \mu^{-1}(0) \subset V.
\eeqn
Denote by ${\rm Crit}W\subset V$ the critical locus of $W$ and denote
\beqn
({\rm Crit} W)^{\rm ss}:= {\rm Crit} W \cap V^{\rm ss}.
\eeqn

\begin{defn}\label{defn32}
A GLSM space $(V, G, W, \mu)$ is in the {\bf geometric phase} if $W$ is Morse--Bott in $V^{\rm ss}$, i.e., if $dW|_{V^{\rm ss}}$ intersects cleanly with the zero section of $T^*V $.
\end{defn}

We also make the following technical assumptions.
\begin{hyp}\label{hyp33}
\begin{enumerate}

\item There is a homomorphism $\iota_W: {\mb C}^* \to G$ such that
\beqn
\xi \cdot x = \iota_W (\xi) \cdot x,\ \forall x\in \ov{({\rm Crit} W)^{\rm ss}}.
\eeqn

\item The action of $e^{\frac{2\pi \i}{r}} \iota_W( e^{-\frac{2\pi\i}{r}}) \in U(1) \times K$ on $V$ is the identity.

\item The restriction $\mu$ to $\ov{({\rm Crit} W)^{\rm ss}}$ is proper. 

\item The $K$-action on $\mu^{-1}(0)$ is free. 

\item $V$ is symplectically aspherical: for any smooth map $f: S^2 \to V$ there holds 
\beqn
\int_{S^2} f^* \omega_V = 0.
\eeqn 

\end{enumerate}
\end{hyp}

All above hypothesis allow us to consider the symplectic reductions
\begin{align*}
&\ V \qu G:= \mu^{-1}(0)/K \cong V^{\rm ss}/G, &\ X:= ({\rm Crit} W \cap \mu^{-1}(0))/K \cong ( {\rm Crit} W)^{\rm ss} /G.
\end{align*}
Our hypothesis implies that $V \qu G$ is a K\"ahler manifold and $X \subset V \qu G$ is a closed submanifold. The manifold $X$ is called the {\bf classical vacuum} of $(V, G, W, \mu)$. The third item of Hypothesis \ref{hyp33} implies that $X$ is compact.

We also have the following assumption on the geometry at infinity. 

\begin{hyp}
There exists $\xi_W$ in the center $Z({\mf k})\subset {\mf k}$, and a continuous function $\tau \mapsto c_W (\tau)$ (for $\tau \in Z({\mf k})$) satisfying the following condition. If we define ${\mc F}_W:= \mu\cdot \xi_W$, then ${\mc F}_W|_{\ov{({\rm Crit} W)^{\rm ss}}}: \ov{({\rm Crit} W)^{\rm ss}} \to {\mb R}$ is proper, bounded from below, and
\beq\label{eqn31}
\begin{array}{c} x \in \ov{({\rm Crit} W)^{\rm ss}}\\
\xi \in T_x V\\
  {\mc F}_W(x) \geq c_W(\tau)
\end{array} \Longrightarrow \left\{ \begin{array}{c} \langle \nabla_\xi \nabla {\mc F}_W(x), \xi\rangle + \langle \nabla_{J \xi} \nabla {\mc F}_W(x), J \xi \rangle \geq 0, \\
\langle \nabla {\mc F}_W(x), J {\mc X}_{\mu(x) - \tau}(x) \rangle \geq 0. \end{array}\right. 
\eeq
\end{hyp}

As shown in \cite[Subsection 3.2]{Tian_Xu_geometric} there is a class of GLSM spaces satisfying these assumptions with classical vacuum being smooth projective hypersurfaces $ X\subset \mb{CP}^n$.

\subsection{Gauged maps and vortices}\label{subsection32}

We have given a detailed review of gauged maps and vortices in the companion paper \cite[Section 3.3]{Tian_Xu_geometric}. For the convenience of the reader we briefly recap the basic notions. 

\begin{defn} 
A {\it gauged map} from a surface $\Sigma$ to $V$ is a triple ${\bm v} = (P, A, u)$ where $P \to \Sigma$ is a principal $K$-bundle, $A \in {\mc A}(P)$ is a connection (or called a {\it gauge field}) and $u\in {\mc S}(Y)$ is a section of the associated bundle $P(V) = (P\times V)/K$. When $P$ is fixed in the context, we also call the pair ${\bm v} = (A, u)$ a gauged map from $P$ to $V$. 
\end{defn}

To any gauged map ${\bm v} = (P, A, u)$ there are the following objects associated.
\begin{enumerate}

\item The {\it covariant derivative} of $u$ with respect to $A$
\beqn
d_A u  \in \Gamma( \Sigma, \Lambda^1 \otimes u^* P(TV)) = \Gamma( \Sigma, {\rm Hom}( T\Sigma, u^* P(TV)) ),
\eeqn

\item the $(0, 1)$-part of the covariant derivative
\beqn
\ov\partial_A u  \in \Omega^{0,1} ( u^* P(TV) ) \cong \Gamma(\Sigma, {\rm Hom}^{0,1}(T\Sigma, u^* P(TV))),
\eeqn

\item the curvature of $A$ 
\beqn
F_A \in \Omega^2 (\Sigma, {\rm ad} P),
\eeqn

\item and the composition of $u$ with the moment map $\mu$
\beqn
\mu(u) \in \Omega^0(\Sigma, {\rm ad} P^*).
\eeqn
\end{enumerate}

Now we recall the vortex equation. Choose an adjoint-invariant metric on ${\mf k}$ so that we can identify ${\mf k} \cong {\mf k}^*$. Choose an area form on the surface $\Sigma$ so we can identify a zero-form with a two-form by the Hodge star operator $*$. A {\it vortex} is a gauged map $(P, A ,u)$ satisfying the following equation 
\begin{align*}
&\ \ov\partial_A u = 0,\ &\ * F_A + \mu(u) = 0.
\end{align*}
This equation is gauge-invariant. Moreover, vortices are minimizers of the {\it Yang--Mills--Higgs functional}, defined by
\beqn
E(A, u) = \frac{1}{2} \Big[ \| d_A u\|_{L^2}^2 + \| F_A \|_{L^2}^2 + \| \mu(u)\|_{L^2}^2 \Big].
\eeqn
Here the energy density is computed using the metric on $X$ induced from the symplectic form and the almost complex structure, and the metric on $\Sigma$ induced from the complex structure and the area form.

\subsection{Gauged Witten equation}

We recall the formulation of the gauged Witten equation given in \cite[Section 3.4]{Tian_Xu_geometric}. Let ${\mc C}= (\Sigma_{\mc C}, L_{\mc C}, \varphi_{\mc C}, \vec{\bf z})$ be a smooth $r$-spin curve. Let $\Sigma_{\mc C}^* \subset \Sigma_{\mc C}$ be the complement of the markings and nodes, which is a smooth open Riemann surface. Recall that in \cite[Subsection 2.3]{Tian_Xu_geometric} we have chosen a Hermitian metric on $L_{\mc C}$ which only depends on the isomorphism class of ${\mc C}$. Then $P_{\mc C} \to \Sigma_{\mc C}^*$ is the associated $U(1)$-bundle and $A_{\mc C}\in {\mc A}(P_{\mc C})$ is the Chern connection. 

In order to define gauged maps from $r$-spin curves, we need to combine a principal $K$-bundle and the principal $U(1)$-bundle induced from the $r$-spin structure. If $P \to \Sigma_{\mc C}^*$ is a principal $K$-bundle, denote by 
\beqn
\ubar P_{\mc C} \to \Sigma_{\mc C}^*
\eeqn
the principal $\ubar K = K \times U(1)$-bundle constructed from $P_{\mc C}$ and $P$. When ${\mc C}$ is fixed from the context, we abbreviate $\ubar P_{\mc C}$ by $\ubar P$. As $\ubar K$ acts on the target space $V$, one can formulate the associated bundle $\ubar P (V) = (\ubar P \times V)/ \ubar K$.

\begin{defn}
Let ${\mc C} = (\Sigma_{\mc C}, L_{\mc C}, \varphi_{\mc C}, \vec{\bf z})$ be a smooth or nodal $r$-spin curve. A {\it gauged map} from ${\mc C}$ to $V$ is a triple ${\bm v} = (P, A, u)$ where $P \to \Sigma_{\mc C}^*$ is a smooth principal $K$-bundle, $A \in {\mc A}(P)$ is a connection, and $u \in {\mc S}( \ubar P (V))$ is a section of $\ubar P (V)$.
\end{defn}

Compared to the vortex equation, the gauged Witten equation has one more term coming from the superpotential. Once the bundle $P \to \Sigma_{\mc C}^*$ is fixed, one can lift $W : V \to {\mb C}$ to a section 
\beqn
{\mc W}_{\mc C}\in \Gamma( \ubar P(V), \pi^* K_{{\mc C}, {\rm log}}) \cong \Gamma( \ubar P(V), \pi^* K_{\Sigma_{\mc C}^*})
\eeqn
(see details in \cite[Section 3.4]{Tian_Xu_geometric}). Using the Hermitian metric on $ \ubar P(TV) \to \ubar P(V)$ induced from the $\uds K$-invariant K\"ahler metric on $V$, one can dualize the differential $d{\mc W}_{\mc C} \in \Gamma( \ubar P(V), \pi^*K_{\Sigma_{\mc C}^*} \otimes \ubar P(T^*V) )$, obtaining the gradient
\beqn
\nabla {\mc W}_{\mc C} \in \Gamma( \ubar P(V), \pi^* \ov{ K_{\Sigma_{\mc C}^*}} \otimes \ubar P(TV)).
\eeqn

Now we write down the gauged Witten equation. For each $A \in {\mc A}(P)$ and $u \in {\mc S}( \ubar P(V))$, one can take the covariant derivative of $u$ with respect to $\ubar A$, and take its $(0,1)$ part. We regard this as an operation on the pair $(A, u)$ and hence denote the covariant derivative by $d_A u$ and its $(0,1)$ part as 
\beqn
\ov\partial_A u \in \Gamma( \Sigma_{\mc C}^*, K_{\Sigma_{\mc C}^*} \otimes u^*  \ubar P(TV))=: \Omega^{0,1}(\Sigma_{\mc C}^*, u^* \ubar P(TV)  ).
\eeqn
It lies in the same vector space as $u^* \nabla {\mc W}_{\mc C} =: \nabla {\mc W}_{\mc C}(u)$. We have the {\it Witten equation}
\beq\label{eqn32}
\ov\partial_A u + \nabla {\mc W}_{\mc C} (u) = 0.
\eeq

We need to fixed the complex gauge by imposing a curvature condition. Let ${\rm ad} P$ be the adjoint bundle. Then the curvature form $F_A$ of $A$ is a section of ${\rm ad} P $ with two form coefficients. On the other hand, since the R-symmetry commutes with the $K$-action, for any section $u$ of $\ubar P(V)$, $\mu (u)$ is a well-defined section of ${\rm ad} P $. The {\it vortex equation} is 
\beq
* F_A + \mu (u) = 0.
\eeq
The {\it gauged Witten equation} over ${\mc C}$ is the system on $(P, A, u)$
\begin{align}\label{eqn34}
&\ \ov\partial_A u + \nabla {\mc W}_{\mc C}(u) = 0,\ &\ * F_A + \mu (u) = 0.
\end{align}

Basic local properties of solutions to the gauged Witten equation still hold as for vortices. For example, there is a gauge symmetry of \eqref{eqn34} for gauge transformations $g: P \to K$. Moreover, for any weak solution, there exists a gauge transformation making it a smooth solution. We have an analogue of energy. For a gauged map ${\bm v} = (P, A, u)$ from an $r$-spin curve ${\mc C}$ to $X$, define its {\it energy} as
\beq
E(P, A, u) =\frac{1}{2} \left[ \int_{\Sigma_{{\mc C}}^*} \Big( | d_A u|^2 + |\mu (u)|^2 + |F_A |^2  + 2 | \nabla {\mc W}_{\mc C} (u) |^2 \Big) \sigma_c \right].
\eeq
We say that a gauged map ${\bm v} = (P, A, u)$ is {\it bounded} if it has finite energy and if there is a $\ubar K$-invariant compact subset $Z \subset V$ such that 
\beqn
u(\Sigma_{\mc C}^* ) \subset \ubar P{} (Z) = (\ubar P \times Z)/\ubar K. 
\eeqn

\section{Analytical Properties of Solutions}\label{section4}

We prove several properties of solutions to the gauged Witten equation.

\subsection{Holomorphicity}

First we show that solutions are all holomorphic and are contained in the critical locus of the superpotential. 

\begin{thm}(cf. \cite[Theorem 4.1]{Tian_Xu_geometric})\label{thm41}
Let ${\bm v} = (P, A, u)$ be a bounded solution over a smooth $r$-spin curve ${\mc C}$. Then 
\begin{align*}
&\ \ov\partial_A u \equiv 0,\ &\ \nabla {\mc W} (u) \equiv 0
\end{align*}
\end{thm}

Now we start to prove Theorem \ref{thm41}. For each puncture $z_a \in \Sigma_{\mc C}$, there exists a local coordinate $w$ such that with respect to the cylindrical metric, near $z_a$, $|d \log w| = |\frac{dw}{w}| = 1$. Denote $z = s + \i t = - \log w$, the cylindrical coordinate on a cylindrical end $U_a^*$. By the definition of $r$-spin structure, there exists a holomorphic section $e$ of $L_{\mc C} |_{U_a^*}$ such that 
\beqn
\varphi_{\mc C} ( e^{\otimes r}) = w^{m_a} dz.
\eeqn
Temporarily omit the index $a$ and denote $q = m_a /r$. If the metric of $L_{\mc C}$ is 
\beqn
\| e \| = |w|^q e^{h},
\eeqn
then by our choice of the metric, $h$ is bounded from below and we have estimates
\beq\label{eqn41}
\sup_{U_a^*} | \nabla^l h | \leq C_l,\ l=0, 1, 2, \ldots.
\eeq

On the other hand, choose an arbitrary trivialization of $P$ over the cylindrical end. Then we obtain trivializations of $\ubar P$ and $\ubar P(V)$. Under these trivializations, we write $A = d + \phi ds + \psi dt$. Then 
\beqn
\ubar A = d + \ubar \phi ds + ({\bf i} q + \ubar \phi) dt = d + ( {\bf i} \partial_t h + \phi ) + ( - {\bf i} q - {\bf i} \partial_s h + \psi ) dt.
\eeqn
Moreover, we have
\begin{align*}
&\ {\mc W}(z, x) =  e^{ - {\bf i} rq t + r h} W(x)  dz,\ & \ \nabla {\mc W}(z, x) = e^{ {\bf i} rq t + r h} \nabla W(x) d \bar{z}. 
\end{align*}

Define the covariant derivatives $\nabla_s^A, \nabla_t^A, \nabla_z^A, \nabla_{\bar z}^A: \Gamma(U_a^*, u^* TV) \to \Gamma(U_a^*, u^* TV)$ by
\begin{align*}
&\ \nabla_s^A \xi = \nabla_s \xi + \nabla_\xi {\mc X}_{\ubar \phi},\ &\ \nabla_t^A \xi = \nabla_t \xi + \nabla_\xi {\mc X}_{\ubar \phi}.
\end{align*}
\begin{align*}
&\ \nabla_z^A \xi  = \frac{1}{2} \Big( \nabla_s^A \xi - J \nabla_t^A \xi\Big),\ &\ \nabla_{\bar z}^A \xi = \frac{1}{2} \Big( \nabla_s^A \xi + J \nabla_t^A \xi \Big).
\end{align*}
Here ${\mc X}$ is the infinitesimal action. Also introduce 
\beqn
{\bm v}_s = \partial_s u + {\mc X}_{\ubar \phi} (u),\ {\bm v}_t = \partial_t u + {\mc X}_{\i q + \ubar \psi} (u),\ {\bm v}_z = \frac{1}{2} \big( {\bm v}_s + J {\bm v}_t \big),\ {\bm v}_{\bar z} = \frac{1}{2} \big( {\bm v}_s - J {\bm v}_t \big).
\eeqn
On the other hand, using the trivializations, we can view sections of ${\rm ad} P$ as maps to ${\mf k}$. We define the covariant differentials of a map $\eta: U_a^* \to {\mf k}$ to be 
\begin{align}\label{eqn42}
&\ \nabla_s^A \eta = \partial_s \eta + [\phi, \eta],\ &\ \nabla_t^A \eta = \partial_t \eta + [\psi, \eta].
\end{align}
The curvature of $\ubar A$ is written as 
\beqn
F_{\ubar A} = \ubar \kappa ds dt,\ {\rm where}\ \ubar \kappa = \partial_s \ubar \psi - \partial_t \ubar \phi + [\ubar \phi, \ubar \psi].
\eeqn
Let $\kappa_{\mc C}$ be the curvature of the connection $A_{\mc C}$ which is independent of $A$.

\begin{lemma}\label{lemma42}
 In the above notations,
\begin{align*}
&\ \nabla_z^A \nabla {\mc W}(u) = \nabla_{{\bm v}_{\bar z}} \nabla {\mc W}(u),\ &\ \nabla_{\bar z}^A \nabla {\mc W}(u) = \nabla_{{\bm v}_z} \nabla {\mc W} (u) + 2r \frac{\partial h}{\partial \bar z}\nabla {\mc W}.
\end{align*}
\end{lemma}

\begin{proof}
We have the following straightforward calculations. 
\beqn
\begin{split}
\nabla_s^A \nabla {\mc W} = &\  \nabla_s \Big( e^{  {\bf i} rq t + r h} \nabla W \Big) + e^{  {\bf i} r qt + r h} \nabla_{\nabla W} {\mc X}_{\ubar \phi} \\
= &\  e^{  {\bf i} r q t + r h} \Big( r  \partial_s h \nabla W + \nabla_{{\bm v}_s} \nabla W + [\nabla W, {\mc X}_{\ubar \phi} ] \Big)\\
= &\  e^{  {\bf i} r q t + r h} \Big( r  \partial_s h \nabla W + \nabla_{{\bm v}_s} \nabla W + [\nabla W, {\mc X}_{{\bf i} \partial_t h} ] \Big)\\
= &\  e^{  {\bf i} r q t + r h} \Big( r  \partial_s h \nabla W + \nabla_{{\bm v}_s} \nabla W + ( {\bf i} r \partial_t h) \nabla W \Big) \\
 = &\ \nabla_{{\bm v}_s}  e^{ {\bf i} r qt  + r h} \nabla W + 2 r \frac{\partial h}{\partial \bar z} e^{ {\bf i} r q + t r h} \nabla W.
 \end{split}
 \eeqn
For the third equality above we used the fact that $W$ is $K$-invariant. Similarly,
\beqn
\begin{split}
\nabla_t^A \nabla {\mc W} = &\  \nabla_t \Big( e^{ {\bf i} rq t + r h} \nabla W \Big) + e^{ {\bf i} r qt + r h} \nabla_{\nabla W} {\mc X}_{\ubar \psi} \\
= &\  e^{ {\bf i} r q t + r h} \Big( ( {\bf i} r q + r  \partial_t h) \nabla W + \nabla_{{\bm v}_t} \nabla W + [\nabla W, {\mc X}_{\ubar \psi} ] \Big)\\
= &\  e^{  {\bf i} r q t + r h} \Big( (  {\bf i} rq + r  \partial_t h) \nabla W + \nabla_{{\bm v}_t} \nabla W + [\nabla W, {\mc X}_{ - {\bf i} q - {\bf i}\partial_s h} ] \Big)\\
= &\  e^{  {\bf i} r q t + r h} \Big( (  {\bf i} rq +  r  \partial_t h) \nabla W + \nabla_{{\bm v}_s} \nabla W + ( - {\bf i} rq - {\bf i} r \partial_s h) \nabla W \Big) \\
 = &\ \nabla_{{\bm v}_s}  e^{ {\bf i} r qt  + r h} \nabla W + 2 r \frac{\partial h}{\partial \bar z} e^{ {\bf i} r qt + r h} \nabla W.
 \end{split}
 \eeqn
Since $W$ is holomorphic, we have
\beqn
\nabla_z^A \nabla {\mc W} = \frac{1}{2} \Big( \nabla_{{\bm v}_s} e^{{\bf i} rq t + r h} \nabla W - J \nabla_{{\bm v}_t} e^{{\bf i} r qt + rh} \nabla W \Big) = \nabla_{{\bm v}_{\bar z}} \nabla {\mc W},
\eeqn
and
\beqn
\nabla_{\bar z}^A \nabla {\mc W} = \nabla_{ {\bm v}_z } \nabla {\mc W} + 2r \frac{\partial h}{\partial \bar z} \nabla {\mc W}.
\eeqn
This finishes the proof of Lemma \ref{lemma42}.
\end{proof}

The last lemma is used in proving the following result.

\begin{lemma}{\rm (}cf. \cite[Proposition 4.5]{Tian_Xu}{\rm )}
Using the above notations, one has
\beq\label{eqn43}
\lim_{s \to +\infty} | {\bm v}_s(s, t) | = \lim_{s \to +\infty} | {\bm v}_t(s, t) | = \lim_{s \to +\infty} |\nabla {\mc W} (u(s, 	t))| = 0.
\eeq
\end{lemma}

\begin{proof}
We have the following calculation.
\beqn
\begin{split}
&\ (\nabla_s^A)^2 {\bm v}_s + (\nabla_t^A)^2 {\bm v}_s  \\
= &\ \nabla_s^A \Big( \nabla_s^A {\bm v}_s  + \nabla_t^A {\bm v}_t \Big) - \big[ \nabla_s^A, \nabla_t^A \big] {\bm v}_t - \nabla_t^A \Big( \nabla_s^A {\bm v}_t - \nabla_t^A {\bm v}_s \Big)\\
= &\ \nabla_s^A \Big( \nabla_s^A ( -J {\bm v}_t - 2 \nabla {\mc W} (u) )\Big) + \nabla_s^A \Big( \nabla_t^A ( J {\bm v}_s + 2 J \nabla {\mc W}(u)) \Big) \\
&\ - R({\bm v}_s, {\bm v}_t) {\bm v}_t - \nabla_{{\bm v}_t} {\mc X}_{\ubar \kappa} - \nabla_t^A {\mc X}_{\ubar \kappa}\\[0.1cm]
= &\ - J \nabla_s^A {\mc X}_{\ubar \kappa} - \nabla_t^A {\mc X}_{\ubar \kappa} - \nabla_{{\bm v}_t} {\mc X}_{\ubar \kappa} - R({\bm v}_s, {\bm v}_t) {\bm v}_t  - 4 \nabla_s^A \nabla_{\bar z}^A \nabla {\mc W} (u) \\
= &\ J \nabla_z^A \Big( - {\mc X}_{\kappa_{\mc C}} + {\mc X}_{\mu (u) } \Big) + \nabla_{{\bm v}_t} \Big( - {\mc X}_{\kappa_{\mc C}} + {\mc X}_{\mu (u)} \Big) \\
&\ - R({\bm v}_s, {\bm v}_t) {\bm v}_t - 4 \nabla_s^A \Big( \nabla_{{\bm v}_z} \nabla {\mc W} (u) + 2 r \frac{\partial h}{\partial \bar z} \nabla {\mc W} \Big)\\
= &\ J \nabla_z^A \Big( J {\mc X}_{\Delta h} + {\mc X}_{\mu (u) } \Big) + \nabla_{{\bm v}_t} \Big( J {\mc X}_{\Delta h} + {\mc X}_{\mu (u)} \Big) \\
&\ - R({\bm v}_s, {\bm v}_t) {\bm v}_t - 4 \nabla_s^A \Big( \nabla_{ {\bm v}_z} \nabla {\mc W} (u) + 2 r \frac{\partial h}{\partial \bar z} \nabla {\mc W} \Big).
\end{split}
\eeqn
Using the fact that the solution is bounded and \eqref{eqn41}, one has the following estimates of each terms of the last express. First, since the infinitesimal actions are holomorphic vector fields and behave equivariantly, we have 
\beqn
\nabla_z^A {\mc X}_{\Delta h} = \nabla_{{\bm v}_z} {\mc X}_{\Delta h} + {\mc X}_{ \frac{\partial \Delta h}{\partial z}} = O(|{\bm v}_s|) + O(|{\bm v}_t|) + O(1);
\eeqn
\beqn
\nabla_z^A {\mc X}_{\mu (u)} = \nabla_{{\bm v}_z} {\mc X}_{\mu (u)} + {\mc X}_{\nabla_z^A \mu (u)} = O(|{\bm v}_s|) + O(|{\bm v}_t|) + O(1). 
\eeqn
We also have
\beqn
\nabla_{{\bm v}_t}\Big( J {\mc X}_{\Delta h} + {\mc X}_{\mu (u)} \Big) = O(|{\bm v}_t|) + O(1);
\eeqn
and 
\beqn
R({\bm v}_s, {\bm v}_t) {\bm v}_t = O(|{\bm v}_s| |{\bm v}_t|^2);
\eeqn
and
\beqn
\nabla_s^A \nabla_{{\bm v}_z} \nabla {\mc W}(u) = O( |\nabla_s^A {\bm v}_s|) + O(|\nabla_s^A {\bm v}_t|) + O(|{\bm v}_s|^2) + O(|{\bm v}_s||{\bm v}_t|) + O(1);
\eeqn
and
\beqn
\nabla_s^A \frac{\partial h}{\partial \bar z}\nabla {\mc W}(u) = O(|{\bm v}_s|) + O(1). 
\eeqn

Therefore, since $\nabla^A$ preserves the metric, one has
\beqn
\begin{split}
\frac{1}{2} \Delta | {\bm v}_s|^2 &\ = \Big\langle  (\nabla_s^A)^2 {\bm v}_s + (\nabla_t^A)^2 {\bm v}_s, {\bm v}_s \Big\rangle + \big| \nabla_s^A {\bm v}_s \big|^2 + \big| \nabla_t^A {\bm v}_s \big|^2 \\
                           &\ \geq -C \Big( 1 + (|{\bm v}_s|^2 + |{\bm v}_t|^2 )^2 \Big) + |\nabla_s^A {\bm v}_s|^2 + |\nabla_t^A {\bm v}_s|^2 - C \Big( |\nabla_s^A {\bm v}_s| + |\nabla_s^A {\bm v}_t|\Big).
\end{split}
\eeqn
Changing $s$ to $t$, one has a similar estimate
\beqn
\frac{1}{2} \Delta |{\bm v}_t|^2 \geq - C \Big( 1 + ( |{\bm v}_s|^2 + |{\bm v}_t|^2)^2 \Big)  + |\nabla_s^A {\bm v}_t|^2 + |\nabla_t^A {\bm v}_t|^2 - C \Big( |\nabla_t^A {\bm v}_t| + |\nabla_t^A {\bm v}_s|\Big).
\eeqn
Then adding the above two inequalities together and changing the value of $C$, one obtains
\beqn
\Delta \Big( |{\bm v}_s|^2 + |{\bm v}_t|^2 \Big) \geq - C \Big( 1 + (|{\bm v}_s|^2 + |{\bm v}_t|^2)^2 \Big).
\eeqn

By the mean value estimate, there exists $\epsilon > 0$ (depending on the above $C$) such that for $r < 1$ and $(s_0, t_0) \in U_j^*$, denoting the radius $r$ disk centered at $(s_0, t_0)$ by $B_r$, 
\beqn
\int_{B_r} \Big( |{\bm v}_s|^2 + |{\bm v}_t|^2 \Big) ds dt \leq \epsilon \Longrightarrow {\bm v}_s(s_0, t_0) \leq \frac{8}{\pi r^2} \int_{B_r} \Big( |{\bm v}_s|^2 + |{\bm v}_t|^2 \Big) ds dt + \frac{ C r^2}{4}.
\eeqn
One can then derive that $|{\bm v}_s|$ and $|{\bm v}_t|$ converge to $0$ as $s \to +\infty$. The last equality of \eqref{eqn43} follows from the Witten equation \eqref{eqn32}.
\end{proof}
Since ${\bm v}$ is bounded, as a consequence, we know that 
\beq\label{eqn44}
\lim_{s \to +\infty} {\rm dist} (u(s, t), {\rm Crit} W) = 0.
\eeq
Then by the holomorphicity of ${\mc W}$, we have
\begin{multline*}
 \| \ov\partial_A u \|_{L^2}^2 + \| \nabla {\mc W} (u) \|_{L^2}^2 = \| \ov\partial_A u + \nabla {\mc W} (u)\|_{L^2}^2 - 2 \int_{\Sigma_{\mc C}^*} \langle \ov\partial_A u, \nabla {\mc W} (u) \rangle \sigma_c \\
                = -\i \int_{\Sigma_{\mc C}^* } d{\mc W} (u) \cdot \ov\partial_A u  = \i \int_{\Sigma_{\mc C}^*} \ov\partial \big[ {\mc W}(u) \big] =  \i \int_{\Sigma_{\mc C}^*} d \big( {\mc W} (u) \big) = \i \sum_{j=1}^n \lim_{s \to +\infty}  \oint_{C_s(z_j)} {\mc W} (u).
\end{multline*}
As the only critical value of $W$ is zero (see \cite[Lemma 3.4]{Tian_Xu_geometric}), by \eqref{eqn44}, the integrals of ${\mc W} (u)$ along the loops appeared above all converge to zero. It follows that $\| \ov\partial_A u \|_{L^2} = \| \nabla {\mc W} (u) \|_{L^2} = 0$. This finishes the proof of Theorem \ref{thm41}.

\subsection{Asymptotic behavior}

Another part of the $C^0$ asymptotic behavior for a solution ${\bm v} = (A, u)$ is prove that $\mu(u)$ converges to zero at punctures. 

\begin{lemma}\label{lemma44}
One has $\displaystyle \lim_{s \to +\infty} \mu (u(s, t)) = 0$. 
\end{lemma}

\begin{proof}
Recall the definition of the covariant differential of maps into ${\mf k}$ given by \eqref{eqn42}. By the holomorphicity of $u$, one has the following calculation.
\beqn
\begin{split}
(\nabla^A_s)^2 \mu(u) +  (\nabla^A_t)^2 \mu(u) &\ = \nabla^A_s (d\mu \cdot {\bm v}_s) +  \nabla^A_t (d\mu \cdot {\bm v}_t)  \\
&\ = \nabla^A_t ( d\mu \cdot J {\bm v}_s ) - \nabla^A_s ( d\mu \cdot J {\bm v}_t )\\
&\ = -2 \rho_{\ubar K} ( {\bm v}_s, {\bm v}_t) - d\mu \cdot J {\mc X}_{\ubar \kappa}\\
&\ = d\mu\cdot \Big[ J {\mc X}_{\mu(u)} - J {\mc X}_{\kappa_{\mc C}} \Big] - 2\rho_{\ubar K} ( {\bm v}_s, {\bm v}_t).
\end{split}
\eeqn
Here $\rho_{\ubar K}$ is a 2-form taking value in the Lie algebra of $\ubar K$. Therefore, for $C>0$ depending on the solution ${\bm v}$, 
\begin{multline*}
\frac{1}{2} \Delta \big| \mu(u) \big|^2 = \big\langle (\nabla^A_s)^2 \mu(u) + (\nabla^A_t)^2  \mu(u), \mu(u) \big\rangle + \big| \nabla^A_s \mu(u) \big|^2 + \big| \nabla^A_t \mu(u) \big|^2 \\
                                \geq \big\langle d\mu \cdot \big( J {\mc X}_{\mu(u)} - J {\mc X}_{\kappa_{\mc C}} \big) - 2 \rho_{\ubar K} ( {\bm v}_s, {\bm v}_t), \mu(u) \big\rangle \geq - C \big( 1 + |\mu(u)|^4 \big).
                                \end{multline*}
Again, this lemma follows from a mean-value estimate and the finite energy condition. 
\end{proof}

Another corollary of above results is the following result on removal of singularity. 

\begin{cor}\label{cor45}
Let ${\bm v} = (P, A, u)$ be a bounded solution of the gauged Witten equation over a smooth $r$-spin curve ${\mc C}$. Let $z_a$ be a puncture of ${\mc C}$ at which the monodromy of $L_{\mc C}$ is $e^{2\pi q_a {\bf i}}$. Then there exist $x_a \in {\rm Crit} W \cap \mu^{-1}(0)$, $\eta_a \in {\mf k}$, and a smooth trivialization of $\ubar P$ over the cylindrical end at $z_a$ such that, with respect to this trivialization, we can write $A = d + \phi ds + \psi dt$ and regard $u$ as a map $u_a: U_a^* \to X$, satisfying 
\beqn
\lim_{s \to +\infty} \phi = 0,\ \lim_{s \to +\infty} \psi = \eta_a,\ \lim_{s \to +\infty} e^{- ( {\bf i} q_a + \eta_a) t} u(s, t) = x_a. 
\eeqn
\end{cor}

\begin{proof}
Lemma \ref{lemma44} imply that over a punctured neighborhood $U_a^*$ of $z_a$, the value of $u$ is very close to the level set $\mu^{-1}(0)$. As there is a holomorphic projection from a neighborhood of $\mu^{-1}(0)$ to the GIT quotient $\mu^{-1}(0)/K$ which annihilates infinitesimal $G$-actions, near the puncture $u$ projects to a holomorphic map $\bar{u}$ into $\mu^{-1}(0)/K$. It is easy to see that the $L^2$-energy of $\ov{u}$ is finite. Hence by Gromov's theorem on removable singularities, $\bar{u}$ extends to a holomorphic map over the puncture. Then by choosing suitable gauge on the cylindrical end, it is easy to find a smooth trivialization of $\ubar P|_{U_a^*}$ which satisfies the prescribed properties. 
\end{proof}

Then over each cylindrical end of $\Sigma_{\mc C}^*$ we can regard a solution to the gauged Witten equation over ${\mc C}$ as a special solution to the symplectic vortex equation with target $({\rm Crit} W)^{\rm ss}$. We can view the gauge group as either $K$ or $\ubar K$, because Hypothesis \ref{hyp33} says that over $({\rm Crit} W)^{\rm ss}$ the $K$-action and the R-symmetry merge together. Then we can use the results about asymptotic behavior of vortices (see for example \cite{Chen_Wang_Wang}\cite{Venugopalan_quasi}) to obtain the following corollary.

\begin{thm}\label{thm46}
There exists $\delta>0$ satisfying the following conditions. Let ${\bm v} = (P, A, u)$ be a bounded solution to the gauged Witten equation over a smooth $r$-spin curve ${\mc C}$. Let $z_a$ be a puncture and $U_a^*$ be a cylindrical end around $z_a$ with cylindrical metric $s + {\bf i} t$. 
\begin{enumerate}

\item Let $\tt{e}: \Sigma_{\mc C}^*  \to [0, +\infty)$ be the energy density function. Then
\beqn
\limsup_{s \to +\infty} e^{\delta s} \tt{e}(s, t) < +\infty.
\eeqn

\item There exist a trivialization of $\ubar P|_{U_a^*}$ and a point $x_a \in {\rm Crit} W \cap \mu^{-1}(0)$ satisfying the conditions of Corollary \ref{cor45}. Moreover, if we write $\ubar A = d + \ubar \phi ds + \ubar\psi dt$ and $u$ as a map $u_a: U_a^* \to X$, then
\beqn
\limsup_{s \to +\infty}  \Big( |\ubar\phi | + |\nabla \ubar\phi | + | \ubar \psi | + |\nabla \ubar\psi | \Big) e^{\delta s} < +\infty;
\eeqn
if we write $u_a (s, t) = \exp_{x_a} \xi_a(s, t)$, then 
\beqn
\limsup_{s \to +\infty }  \Big(  \| \xi_a \| + \| \nabla \xi_a \| \Big) e^{\delta s} < +\infty.
\eeqn

\end{enumerate}
\end{thm}

As we have explained in \cite[Section 4]{Tian_Xu_geometric}, Theorem \ref{thm41} and Corollary \ref{cor45} imply that a bounded solution represents an equivariant homology class in $V$ for the action by $\ubar K = K \times U(1)$. Given $g\geq 0$, $n \geq 0$ with $2g+n \geq 3$ and a class $\ubar B \in H_2^{\ubar K} ( V; {\mb Z})$, let 
\beqn
\tilde {\mc M}_{g, n}^r ( V, G, W, \mu; \ubar B)
\eeqn
be the set of pairs $({\mc C}, {\bm v})$ where ${\mc C}$ is a smooth $r$-spin curve of type $(g, n)$ and ${\bm v} = (P, A, u)$ is a smooth bounded solution to the gauged Witten equation over ${\mc C}$ which represents $\ubar B$. 

Moreover, as explained in \cite[Remark 4.6]{Tian_Xu_geometric}, a solution to the gauged Witten equation corresponds naturally to a solution to the vortex equation (with a small perturbation). Using the group homomorphism $\iota_W: U(1) \to K$, the bundle $\ubar P \to {\mc C}$ induces a $K$-bundle $P' \to {\mc C}$ and the connection $\ubar A$ induces a connection $A'$ on $P'$. The section $u$ becomes a section of $P'(V)$ and the triple $(P', A', u')$ solves the equation 
\begin{align}\label{induced_vortex}
&\ \ov\partial_{A'} u' = 0,\ &\ * F_{A'} + \mu(u') = * F_{\iota_W(A_{\mc C})}
\end{align}
where the term $* F_{\iota_W(A_{\mc C})}$ decays exponentially along cylindrical ends. 

\section{Moduli Space of Stable Solutions}\label{section5}

In the companion paper \cite{Tian_Xu_geometric} we defined the topology of the moduli space of solutions to the $r$-spin equation. We defined a natural compactification of the moduli space by allowing the so-called soliton components and proved the corresponding compactness theorem. In this section we quickly recall the related notions before we state the main theorem on the virtual cycle.


\subsection{Solitons}

In our context, solitons are bounded solutions to the gauged Witten equation over the infinite cylinder 
\beqn
\Thetait:= (-\infty, +\infty) \times S^1.
\eeqn
This notion was defined in \cite[Definition 5.1]{Tian_Xu_geometric} so we only briefly recall here. The set of isomorphism classes of $r$-spin structures over $\Thetait$ is a torsor of ${\mb Z}_r$. By our convention (see Subsection \ref{subsection23}) the line bundle $L_{\mc C}$ which is trivial has a flat Hermitian metric with Chern connection written as 
\beqn
A_{\mc C} = d + {\bf i} q dt,\ q \in \frac{1}{r} \{ 0, 1, \ldots, r-1\}.
\eeqn
A $q$-{\it soliton} is a bounded solution to gauged Witten equation over the infinite cylinder $\Thetait = {\mb R} \times S^1$ equipped with translation invariant metric. More precisely, a soliton is a gauged map ${\bm v}= (u, \phi, \psi)$ from $\Thetait$ to $V$ which solves the equation
\begin{align}\label{eqn51}
&\ \left\{ \begin{array}{r} \partial_s u + {\mc X}_{\phi} + J  (\partial_t u + {\mc X}_{{\bf i} q + \psi}) = 0,\\
\partial_s \psi - \partial_t \phi + [\phi, \psi ] + \mu(u) = 0.
\end{array}
\right. \ &\ u(\Thetait) \subset \ov{{\rm Crit}W \cap V^{\rm ss}}.
\end{align}
By Corollary \ref{cor45} and Theorem \ref{thm46} one can gauge transform a soliton so that
\begin{enumerate}

\item as $s\to \pm\infty$, $\phi$ and $\iota_W({\bf i} q) + \psi$ are asymptotic to zero (exponentially fast).

\item There exists $x_\pm \in {\rm Crit} W \cap \mu^{-1}(0)$ such that
\beq\label{eqn52}
\lim_{s \to \pm \infty} u(s, t) = x_\pm.
\eeq
\end{enumerate}
So as before a soliton represents an equivariant curve class $\ubar B \in H_2^{\ubar K} (V; {\mb Z})$, and we have the following energy identity
\beqn
E (u, \phi, \psi) = \langle [\omega_V^{\ubar K}, \ubar B] \rangle.
\eeqn
On the other hand, using a compactness argument one can prove that the energy of nontrivial solitons are bounded from below. We omit the details.

\subsection{Decorated dual graphs}

We briefly recall the notion of dual graphs modelling $r$-spin curves. The details can be found in \cite[Section 5.2]{Tian_Xu_geometric}. Recall that a marked smooth or nodal curve can be described by a dual graph $\uds \Gamma$ which contains a set of vertices $\V(\uds\Gamma)$ (corresponding to irreducible components), a set of edges $\E(\uds\Gamma)$ (corresponding to nodes), a set of tails $\T(\uds \Gamma)$ (corresponding to markings), and a genus function on $\V(\uds\Gamma)$. The stability condition can be described in terms of this dual graph. To model $r$-spin curves, we need certain extra structures. Given a dual graph $\uds \Gamma$, let $\uds {\tilde \Gamma}$ be the dual graph obtained by cutting off all edges (including loops), i.e., the graph of the normalization. So $\E( \uds{\tilde \Gamma}) = \emptyset$ and there is a natural inclusion $\T (\uds \Gamma) \hookrightarrow \T( \uds{\tilde \Gamma})$. Denote 
\beqn
\tilde \E(\uds\Gamma) = \T (\uds{\tilde \Gamma}) \setminus \T(\uds \Gamma).
\eeqn
An {\it $r$-spin dual graph} (see \cite[Definition 5.4]{Tian_Xu_geometric}), denoted by $\Gamma$, consists of an underlying dual graph $\uds \Gamma$ and a map
\beqn
\tt{m}: \T(\uds{\tilde \Gamma}) \to {\mb Z}_r
\eeqn
satisfying the following conditions.
\begin{enumerate}

\item If $\tilde e_-, \tilde e_+ \in \tilde \E(\uds\Gamma) \subset \T(\uds{\tilde \Gamma})$ are obtained by cutting off $e \in \E( \uds\Gamma)$, then $\tt{m}(\tilde e_-) \tt{m}( \tilde e_+) = 1$. 

\item If the genus of a vertex $v \in \V(\uds\Gamma)$ is zero, then its valence is at least two.

\item If the genus of a vertex $v\in \V(\uds\Gamma)$ is zero and the valence is two, the two edges connected to $v$ do not form a loop. Moreover, for the two tails $t, t' \in \T(
\uds{\tilde \Gamma})$ that are attached to $v$, we have $\tt{m}(t) \tt{m}(t') = 1$. 
\end{enumerate}
We also had the notion of maps between $r$-spin dual graphs (see \cite[Definition 5.5]{Tian_Xu_geometric}. We do not repeat it here. 

Given an $r$-spin dual graph $\Gamma$, there is a notion of {\it stabilization}, denoted by $\Gamma^{\rm st}$. 

The combinatorial type of an $r$-spin curve can be described by an $r$-spin dual graph. For a stable $r$-spin dual graph $\Gamma$, let ${\mc M}{}_{\Gamma}^r \subset \ov{\mc M}{}_{g, n}^r$ be the subset of points corresponding to stable $r$-spin curves that have $r$-spin dual graph $\Gamma$. There is a partial order $\preq$ among $r$-spin dual graphs and the closure of ${\mc M}{}_{\Gamma}^r$ can be described as 
\beqn
\ov{\mc M}{}_{\Gamma}^r = \bigsqcup_{\Pi \preq \Gamma, \Pi = \Pi^{\rm st} } {\mc M}{}_{\Pi}^r. 
\eeqn

The combinatorial type of solutions to the gauged Witten equation over smooth or nodal $r$-spin curves can be described by {\it decorated dual graphs}. Such an object is typically denoted by $\sf\Gamma$, which consists of an $r$-spin dual graph $\Gamma$ and a curve class $\ubar B_v\in H_2^{\ubar K}( V; {\mb Z})$ for each vertex $v$ (see \cite[Definition 5.7]{Tian_Xu_geometric}). The partial order of $r$-spin dual graphs can be lifted to a partial order among decorated dual graphs, which is still denoted by $\preq$.

\subsection{Stable solutions}\label{subsection53}

We recall the definition of stable solutions.

\begin{defn}\cite[Definition 5.8]{Tian_Xu_geometric}
Let $\sf \Gamma$ be a decorated dual graph. A {\it solution to the gauged Witten equation of combinatorial type $\sf\Gamma$} consists of a smooth or nodal $r$-spin curve ${\mc C}$ of combinatorial type $\Gamma$, and a collection of objects
\beqn
{\bm v}:= \Big[ ({\bm v}_v)_{ v \in \V(\Gamma)} \Big]
\eeqn
Here for each stable vertex $v \in \V (\Gamma)$, ${\bm v}_v = (P_v, A_v, u_v)$ is a solution to the gauged Witten equation over the smooth $r$-spin curve ${\mc C}_v$; for each unstable vertex $v \in \V (\Gamma)$, ${\bm v}_v = (P_v, A_v, u_v)$ is a soliton. They satisfy the following conditions.
\begin{enumerate}
\item For each vertex $v \in \V (\Gamma)$, the equivariant curve class represented by ${\bm v}_v$ coincides with $\ubar B_v$ (which is contained in the data $\sf \Gamma$).


\item For each edge $e \in \E(\Gamma)$ corresponding to a node in ${\mc C}$, let $\tilde e_-$ and $\tilde e_+$ be the two tails of $\tilde \Gamma$ obtained by cutting off $e$, which are attached to vertices $v_-$ and $v_+$ (which could be equal), corresponding to preimages $\tilde w_-$ and $\tilde w_+$ of $w$ in $\tilde {\mc C}$. 
Then the evaluations of ${\bm v}_{v_-}$ at $w_-$ and ${\bm v}_{v_+}$ at $w_+$ (which exist by Corollary \ref{cor45}) are the same point of the classical vacuum $X$. 
\end{enumerate}

The solution is called {\it stable} if $\sf\Gamma$ is stable. The energy (resp. homology class) of a solution ${\bm v}$ is the sum of energies (resp. homology classes) of each component.
\end{defn}

One can define an equivalence relation among all stable solutions, and define the corresponding moduli spaces of equivalence classes. 

\begin{defn} \cite[Definition 5.9]{Tian_Xu_geometric}
Given a decorated dual graph $\sf\Gamma$, let 
\beqn
{\mc M}_{\sf\Gamma}:= {\mc M}{}_{\sf\Gamma} (V, G, W, \mu ) 
\eeqn
be the set of equivalence classes of stable solutions of combinatorial type $\sf\Gamma$. Given $g, n$ with $2g + n \geq 3$, define 
\beqn
\ov{\mc M}{}_{g, n}^r (V, G, W, \mu):= \bigsqcup_{{\sf \Gamma}} {\mc M}_{\sf \Gamma}(V, G, W, \mu)
\eeqn
where the disjoint union is taken over all stable decorated dual graphs. 
\end{defn}

In \cite[Section 5]{Tian_Xu_geometric} we defined the topology on $\ov{\mc M}{}_{g,n}^r(V, G, W, \mu)$. If we define 
\beqn
\ov{\mc M}{}_{g,n}^r(V, G, W, \mu; \uds B)
\subset \ov{\mc M}{}_{g,n}^r(V, G, W, \mu)
\eeqn
be the subspace of stable solutions in curve class $\uds B$, then we proved (see \cite[Theorem 5.14]{Tian_Xu_geometric}) that it is a compact space, with each $\ov{\mc M}{}_{\sf \Gamma}$ a closed subset.

\section{Topological Virtual Orbifolds and Virtual Cycles}\label{section6}

We recall our framework of constructing virtual fundamental cycles associated to moduli problems. Such constructions, usually called ``virtual technique'', has a long history since it first appeared in algebraic Gromov--Witten theory by \cite{Li_Tian_2}. The current method is based on the topological approach of \cite{Li_Tian}.	Since our approach has been detailed in \cite[Appendix]{Tian_Xu_2021}, we will not repeat the same discussions but frequently refer to corresponding texts.

\subsection{Topological manifolds and transversality}

Recall that an $n$-dimensional topological manifold $M$ is a second countable Hausdorff space $M$ which is locally homeomorphic to ${\mb R}^n$; a topological submanifold of $M$ is a subset $S \subset M$ which is a topological manifold with the subspace topology. A map $f: N \to M$ between topological manifold is an embedding if it is a homeomorphism onto its image; it is called a {\it locally flat embedding} if for any $p \in f(N)$, there is a local coordinate $\varphi_p: U_p \to {\mb R}^m$ where $U_p \subset {\mb R}^n$ is an open neighborhood of $p$ such that 
\beqn
\varphi_p (f(N) \cap U_p) \subset {\mb R}^n \times \{0\}.
\eeqn
In this paper, without further clarification, all embeddings of topological manifolds are assumed to be locally flat (in fact having normal microbundles).

\subsubsection{Microbundles}

The discussion of topological transversality needs the concept of microbundles, which was introduced by Milnor \cite{Milnor_micro_1}.

\begin{defn}{\rm (Microbundles)}
Let $B$ be a topological space. 
\begin{enumerate}

\item A {\it microbundle} over a $B$ is a triple $(E, i, p)$ where $E$ is a topological space, $i: M \to E$ (the zero section map) and $p: E \to M$ (the projection) are continuous maps, satisfying the following conditions.
\begin{enumerate}
\item $p \circ i = {\rm Id}_M$.

\item For each $b \in B$ there exist an open neighborhood $U \subset B$ of $b$ and an open neighborhood $V \subset E$ of $i(b)$ with $i(U) \subset V$, $j(V) \subset U$, such that there is a homeomorphism $V \cong U \times {\mb R}^n$ which makes the following diagram commutes.
\beqn
\vcenter{ \xymatrix{  &  V \ar[dd] \ar[rd]^p & \\
                     U \ar[ru]^i \ar[rd]^i & & U \\
                      & U \times {\mb R}^n \ar[ru]^p & } }.
\eeqn
\end{enumerate}

\item Two microbundles $\xi = (E, i, p)$ and $\xi' = (E', i', p')$ over $B$ are {\it equivalent} if there are open neighborhoods of the zero sections $W \subset E$, $W' \subset E'$ and a homeomorphism $\rho: W \to W'$ which is compatible with the structures of the two microbundles.
\end{enumerate}
\end{defn}


\begin{defn}{\rm (Normal microbundles)}
Let $f: S \to M$ be a topological embedding.

\begin{enumerate}

\item  A {\it normal microbundle} of $f$ is a pair $\xi = (N, \nu)$ where $N \subset M$ is an open neighborhood of $f(S)$ and $\nu: N \to S$ is a continuous map such that together with the natural inclusion $S \hookrightarrow N$ they form a microbundle over $N$. A normal microbundle is also called a {\it tubular neighborhood}.

\item Two normal microbundles $\xi_1 = (N_1, \nu_1)$ and $\xi_2 = (N_2, \nu_2)$ are {\it equivalent} if there is another normal microbundle $(N, \nu)$ with $N \subset N_1 \cap N_2$ and 
\beqn
\nu_1|_N = \nu_2|_N = \nu.
\eeqn
An equivalence class is called a {\it germ}.
\end{enumerate}
\end{defn}

For example, for a smooth submanifold $S \subset M$ in a smooth manifold, there is always a normal microbundle. Its equivalence class is not unique though, as we need to choose the projection map. However, different normal microbundles are always isotopic. As a result, intersection numbers will be independent of the choice of normal microbundles.

\subsubsection{Transversality}

We first recall the notion of microbundle transversality. Let $Y$ be a topological manifold, $X \subset Y$ be a submanifold and $\xi =  (N, \nu)$ be a normal microbundle of $X$. Let $f: M \to Y$ be a continuous map. 

\begin{defn} \label{defn63} {\rm (Microbundle transversality)}
Let $Y$ be a topological manifold, $X \subset Y$ be a submanifold and $\xi_X$ be a normal microbundle of $X$. Let $f: M \to Y$ be a continuous map. We say that $f$ is {\it transverse to $\xi$} if the following conditions are satisfied.
\begin{enumerate}

\item $f^{-1}(X)$ is a submanifold of $M$.

\item $f^{-1}(X)$ has a normal microbundle $\xi' = (N', \nu')$ making the diagram
\beqn
\xymatrix{ N' \ar[d] \ar[r]^f &  N \ar[d] \\
           f^{-1}(X) \ar[r]   &  X }
\eeqn
commute and the inclusion $f: N' \to N$ induces an equivalence of microbundles.
\end{enumerate}
More generally, if $C \subset M$ is any subset, then we say that $f$ is transverse to $X$ {\it near} $C$ if the restriction of $f$ to an open neighborhood of $C$ is transverse to $X$. 
\end{defn}
It is easy to see that the notion of being transverse to $\xi$ only depends on the germ of $\xi$. 


The following theorem, which is of significant importance in our virtual cycle construction, shows that one can achieve transversality by arbitrary small perturbations. The reader can find references in \cite[Remark A.8]{Tian_Xu_2021}.

\begin{thm} \label{thm_transversality} {\rm (Topological transversality theorem)} Let $Y$ be a topological manifold and $X \subset Y$ be a proper submanifold. Let $\xi$ be a normal microbundle of $X$. Let $C \subset D \subset Y$ be closed sets. Suppose $f: M \to Y$ is a continuous map which is microbundle transverse to $\xi$ near $C$. Then there exists a homotopic map $g: M \to Y$ which is transverse to $\xi$ over $D$ such that the homotopy between $f$ and $g$ is supported in an open neighborhood of $f^{-1}( (D \setminus C) \cap X)$. 
\end{thm}


In most of the situations of this paper, the notion of transversality is about sections of vector bundles. Suppose $f: M \to {\mb R}^n$ is a continuous map. The origin $0 \in {\mb R}^n$ has a canonical normal microbundle. Therefore, one can define the notion of transversality for $f$ as a special case of Definition \ref{defn63}. Now suppose $E \to M$ is a vector bundle and 
\beqn
\varphi_U: E|_U \to U \times {\mb R}^n
\eeqn
is a local trivialization. Each section $s: M \to E$ induces a map $s_U: U \to {\mb R}^n$. Then we say that $s$ is transverse over $U$ if $s_U$ is transverse to the origin of ${\mb R}^n$. This notion is clearly independent of the choice of local trivializations. Then $s$ is said to be transverse if it is transverse over a sufficiently small neighborhood of every point of $M$. 


\begin{thm}\label{thm65}\cite[Theorem A.9]{Tian_Xu_2021}
Let $M$ be a topological manifold and $E \to M$ be an ${\mb R}^n$-bundle. Let $C\subset D \subset M$ be closed subsets. Let $s: M \to E$ be a continuous section which is transverse near $C$. Then there exists another continuous section which is transverse near $D$ and which agrees with $s$ over a small neighborhood of $C$. 
\end{thm}

\begin{cor}
Let $M$ be a topological manifold without boundary, and let $s_1, s_2: M \to E$ be two transverse sections which are homotopic. Then the two submanifolds $S_1:= s_1^{-1}(0)$ and $S_2:= s_2^{-1}(0)$ are cobordant. 
\end{cor}

Another scenario where we will apply the transversality theorem is the transversality to a smooth submanifold. If $Y \subset M$ is a smooth submanifold, then $Y$ admits a tubular neighborhood and a unique isotopy class of normal microbundles. Towards an intersection theory, we do not need to specify a normal microbundle to smooth submanifolds, although the notion of transversality is certainly different for different choices.

\subsection{Topological orbifolds and orbibundles}

For simplicity we use Satake's notion of V-manifolds \cite{Satake_orbifold} instead of groupoids to treat orbifolds, and only discuss it in the topological category.\footnote{When discussing non-effective orbifolds, the usual way is to take a further stabilization. Namely, one can replace a non-effective orbifold by the total space of an orbifold vector bundle over it, which can be chosen to be effective.}

\begin{defn}
Let $M$ be a second countable Hausdorff topological space.

\begin{enumerate}

\item Let $x \in M$ be a point. A topological orbifold chart (with boundary) of $x$ consists of a triple $(\tilde U_x, \Gammait_x, \varphi_x)$, where $\tilde U_x$ is a topological manifold with possibly empty boundary $\partial \tilde U_x$, $\Gammait_x$ is a finite group acting continuously  on $( \tilde U_x, \partial \tilde U_x)$ and
\beqn
\varphi_x: \tilde U_x/\Gammait_x \to M
\eeqn
is a continuous map which is a homeomorphism onto an open neighborhood of $x$. Denote the image $U_x = \varphi_x(\tilde U_x/ \Gammait_x) \subset M$ and denote the composition 
\beqn
\tilde \varphi_x: \xymatrix{ \tilde U_x \ar[r] & \tilde U_x/ \Gammait_x \ar[r]^-{\varphi_x} & M}.
\eeqn

\item If $p \in \tilde U_x$, take $\Gammait_p = (\Gammait_x)_p \subset \Gammait_x$ the stabilizer of $p$. Let $\tilde U_p \subset \tilde U_x$ be a $\Gammait_p$-invariant neighborhood of $p$. Then there is an induced chart (which we call a subchart) $( \tilde U_p, \Gammait_p, \varphi_p)$, where $\varphi_p$ is the composition
\beqn
\varphi_p: \xymatrix{ \tilde U_p/\Gammait_p \ar@{^{(}->}[r] & \tilde U_x/\Gammait_x \ar[r]^-{\varphi_p} & M
}.
\eeqn

\item Two charts $(\tilde U_x, \Gammait_x, \varphi_x)$ and $( \tilde U_y, \Gammait_y, \varphi_y)$ are {\it compatible} if for any $p \in \tilde U_x$ and $q \in \tilde U_y$ with $\varphi_x(p) = \varphi_y (q)\in M$, there exist an isomorphism $\Gammait_p \to \Gammait_q$, subcharts $ \tilde U_p \ni p$, $\tilde U_q \ni q$ and an equivariant homeomorphism $\varphi_{qp}: ( \tilde U_p, \partial \tilde U_p) \cong ( \tilde U_q, \partial \tilde U_q)$.

\item A {\it topological orbifold atlas} of $M$ is a {\it set} $\{ ( \tilde U_\alpha, \Gammait_\alpha, \varphi_\alpha)\ |\ \alpha \in I \}$ of topological orbifold charts of $M$ such that $M = \bigcup_{\alpha\in I} U_\alpha$ and for each pair $\alpha, \beta \in I$, $( \tilde U_\alpha, \Gammait_\alpha, \varphi_\alpha), ( \tilde U_\beta, \Gammait_\beta, \varphi_\beta)$ are compatible. Two atlases are {\it equivalent} if the union of them is still an atlas. A structure of topological orbifold (with boundary) is an equivalence class of atlases. A topological orbifold (with boundary) is a second countable Hausdorff space with a structure of topological orbifold (with boundary).
\end{enumerate}
\end{defn}
We will often skip the term ``topological'' in the rest of this paper. Similarly one can define orbifold vector bundles using bundle charts. The details can be found in \cite[Definition A.11]{Tian_Xu_2021} and skipped here.

\subsubsection{Embeddings}

Now we consider embeddings for orbifolds and orbifold vector bundles. First we consider the case of manifolds. Let $S$ and $M$ be topological manifolds and $E \to S$, $F \to M$ be continuous vector bundles. Let $\phi: S \to M$ be a topological embedding. A {\it bundle embedding} covering $\phi$ is a continuous map $\wh \phi: E \to F$ which makes the diagram
\beqn
\xymatrix{ E \ar[r]^{\wh\phi} \ar[d] & F \ar[d] \\
           S \ar[r]^\phi         & M }
\eeqn
commute and which is fibrewise a linear injective map. Since $\wh \phi$ determines $\phi$, we also call $\wh \phi: E \to F$ a bundle embedding. 

\begin{defn} {\rm (Orbifold embedding)}
Let $S$, $M$ be orbifolds and $f: S \to M$ is a continuous map which is a homeomorphism onto its image. $\phi$ is called an {\it embedding} if for any pair of orbifold charts, $(\tilde U, \Gammait, \varphi)$ of $S$ and $(\tilde V, \Piit, \psi)$ of $M$, any pair of points $p\in \tilde U$, $q \in  \tilde V$ with $\phi (\varphi(p)) = \psi(q)$, there are subcharts $( \tilde U_p, \Gammait_p, \varphi_p) \subset (\tilde U, \Gammait, \varphi)$ and $(\tilde V_q, \Piit_q, \psi_q) \subset ( \tilde V, \Piit, \psi)$, an isomorphism $\Gammait_{p} \cong \Piit_q$ and an equivariant locally flat embedding $\tilde \phi_{pq}: \tilde U_p \to \tilde V_q$ such that the following diagram commutes.
\beqn
\vcenter{
\xymatrix{ \tilde U_p \ar[d]_{\tilde \varphi_p} \ar[r]^{\tilde \phi_{pq}} & \tilde V_q \ar[d]^{\tilde \psi_q}\\
S \ar[r]^{\phi } & M } }
\eeqn
\end{defn}

\subsubsection{Multisections and perturbations}

The equivariant feature of the problem implies that transversality can only be achieved by multi-valued perturbations. We briefly recall the notion of multisections  originally used in \cite{Fukaya_Ono} and \cite{Li_Tian}. The reader can find the detailed treatment in the current context in our previous paper \cite[Appendix A]{Tian_Xu_2021}.

Multisections are modelled on multi-valued maps between sets (see \cite[Definition A.13]{Tian_Xu_2021}). If $A$, $B$ are sets, a multimap (with $l$ branches) is a map from $A$ to the symmetric product ${\mc S}^l(B)$. One can define a natural equivalence relation among multimaps with different numbers of branches. If both $A$ and $B$ are acted by a group $\Gamma$, then one can consider equivariant multimaps. If $A$ and $B$ are topological spaces, then we only consider continuous multimaps. A continuous multimap $f: A \to {\mc S}^l(B)$ is called {\it liftable} if there are continuous maps $f_1, \ldots, f_l: A \to B$ such that
\beqn
f(x) = [f^1(x), \ldots, f^l(x)] \in {\mc S}^l(B),\ \forall x \in A.
\eeqn
$f^1, \ldots, f^l$ are called {\it branches} of $f$. 

Multisections of orbifold vector bundles are locally equivariant multimaps which match on overlapping charts (see \cite[Definition A.14]{Tian_Xu_2021}). If $M$ is a topological orbifold and $E \to M$ is a vector bundle, a multisection, denoted by $s: M \overset{m}{\to} E$, is called {\it transverse} if it is ``locally liftable'' and for any liftable local representative $s_p: \tilde U_p \to {\mc S}^{l}({\mb R}^k)$, all branches are transverse to the origin of ${\mb R}^k$.

The space of continuous multisections is a module over the ring of continuous functions of the base orbifold. Continuous multisections also form a commutative monoid but not necessarily an abelian group: one can only invert the operation of adding a single-valued section. It is enough, though, since we have the notion of being transverse to a single-valued section which is not necessarily the zero section. 


In \cite{Tian_Xu_2021} we proved the analogue of the transversality theorem for multisections of orbifold vector bundles.

\begin{lemma}\label{lemma69}\cite[Lemma A.15]{Tian_Xu_2021}
Let $M$ be an orbifold and $E \to M$ be an orbifold vector bundle. Let $C \subset D \subset M$ be closed subsets. Let $S: M \to E$ be a single-valued continuous function and $t_C: M \overset{m}{\to} E$ be a multisection such that $S + t_C$ is transverse over a neighborhood of $C$. Then there exists a multisection $t_D: M \overset{m}{\to} E$ satisfying the following condition. 
\begin{enumerate}
\item $t_C = t_D$ over a neighborhood of $C$.

\item $S + t_D$ is transverse over a neighborhood of $D$. 
\end{enumerate}
Moreover, if $E$ has a continuous norm $\| \cdot \|$, then for any $\epsilon>0$, one can require that 
\beqn
\| t_D \|_{C^0} \leq \| t_C \|_{C^0} + \epsilon.
\eeqn
\end{lemma}

\subsection{Virtual orbifold atlases}

Now we introduce the notion of virtual orbifold atlases. This notion plays a role as a general structure of moduli spaces we are interested in, and is a type of intermediate objects in concrete constructions. The eventual objects we need are {\it good coordinate systems}, which are special types of virtual orbifold atlases. 

\begin{defn}\label{defn610}\cite[Definition A.16]{Tian_Xu_2021} Let ${\mc X}$ be a compact and Hausdorff space.
\begin{enumerate}

\item A {\it virtual orbifold chart} (chart for short) on ${\mc X}$ is a tuple 
\beqn
C:= (U, E, S, \psi, F)
\eeqn
where
\begin{enumerate}

\item $U$ is a topological orbifold (with boundary).

\item $E \to U$ is a continuous orbifold vector bundle.

\item $S: U \to E$ is a continuous section.

\item $F \subset {\mc X}$ is an open subset. 

\item $\psi: S^{-1}(0) \to F$ is a homeomorphism.
\end{enumerate}
$F$ is call the {\it footprint} of this chart $C$, and the integer ${\rm dim} U - {\rm rank} E$ is called the {\it virtual dimension} of $C$.

\item Let $C = (U, E, S, \psi, F)$ be a chart and $U' \subset U$ be an open subset. The {\it restriction} of $C$ to $U'$ is the chart 
\beqn
C' = C|_{U'} = (U', E', S', \psi', F')
\eeqn
where $E' = E|_{U'}$, $S' = S|_{U'}$, $\psi' = \psi|_{(S')^{-1}(0)}$, and $F' = {\rm Image}(\psi')$. Any such chart $C'$ induced from an open subset $U' \subset U$ is called a {\it shrinking} of $C$. A shrinking $C' = C|_{U'}$ is called a {\it precompact shrinking} if $U' \sqsubset U$, denoted by $C' \sqsubset C$.
\end{enumerate}
\end{defn}



\begin{defn}\label{defn611}\cite[Definition A.18]{Tian_Xu_2021}
Let $C_i:= (U_i, E_i, S_i, \psi_i, F_i)$, $i=1,2$ be two charts of $X$. An {\it embedding} of $C_1$ into $C_2$ consists of a bundle embedding $\wh{\phi}_{21}$ satisfying the following conditions.
\begin{enumerate}

\item The following diagrams commute;
\begin{align*}
&\ \xymatrix{E_1 \ar[r]^{\wh\phi_{21}}   \ar[d]^{\pi_1} & E_2 \ar[d]_{\pi_2} \\
          U_1 \ar@/^1pc/[u]^{S_1} \ar[r]^{\phi_{21}} & U_2 \ar@/_1pc/[u]_{S_2}       }\ &\  \xymatrix{  S_1^{-1}(0) \ar[r]^{\phi_{21}} \ar[d]^{\psi_1} & S_2^{-1}(0) \ar[d]^{\psi_2}  \\     X \ar[r]^{{\rm Id}} & X }
\end{align*}

\item {\bf (Tangent Bundle Condition)} There exists an open neighborhood $N_{21} \subset U_2$ of $\phi_{21}(U_1)$ and a subbundle $E_{1;2}\subset E_2|_{N_{21}}$ which extends $\wh\phi_{21}(E_1)$ such that $S_2|_{N_{21}}: N_{21} \to E_2|_{N_{21}}$ is transverse to $E_{1;2}$ and $S_2^{-1}(E_{1;2}) = \phi_{21}(U_1)$. 
\end{enumerate}
\end{defn}



\begin{defn}\label{defn612}\cite[Definition A.20]{Tian_Xu_2021}
Let $C_i = (U_i, E_i, S_i, \psi_i, F_i)$, $(i=1, 2)$ be two charts. A {\it coordinate change} from $C_1$ to $C_2$ is a triple $T_{21} = (U_{21}, \phi_{21}, \wh\phi_{21})$, where $U_{21}\subset U_1$ is an open subset and $(\phi_{21}, \wh\phi_{21})$ is an embedding from $C_1|_{U_{21}}$ to $C_2$. They should satisfy the following conditions.
\begin{enumerate}

\item $\psi_1(U_{21} \cap S_1^{-1}(0)) = F_1 \cap F_2$.

\item If $x_k \in U_{21}$ converges to $x_\infty \in U_1$ and $y_k = \phi_{21}(x_k)$ converges to $y_\infty \in U_2$, then $x_\infty \in U_{21}$ and $y_\infty = \phi_{21}(y_\infty)$. 
\end{enumerate}
\end{defn}

\begin{lemma}\label{lemma613}\cite[Lemma A.21]{Tian_Xu_2021}
Let $C_i = (U_i, E_i, S_i, \psi_i, F_i)$, $(i=1, 2)$ be two charts and let $T_{21} = (U_{21}, \wh\phi_{21})$ be a coordinate change from $C_1$ to $C_2$. Suppose $C_i' = C_i|_{U_i'}$ be a shrinking of $C_i$. Then the restriction $T_{21}':= T_{21}|_{U_1' \cap \phi_{21}^{-1}(U_2')}$ is  a coordinate change from $C_1'$ to $C_2'$. 
\end{lemma}

Now we introduce the notion of atlases.

\begin{defn}\label{defn_atlas}\cite[Definition A.22]{Tian_Xu_2021} Let ${\mc X}$ be a compact Hausdorff space. A {\it virtual orbifold atlas} of virtual dimension $d$ on ${\mc X}$ is a collection
\beqn
{\mc A}:= \Big( \Big\{ C_I:= (U_I, E_I, S_I, \psi_I, F_I)\ |\ I \in {\mc I} \Big\},\ \Big\{ T_{JI} = \big( U_{JI}, \phi_{JI}, \wh\phi_{JI}\big) \ |\ I \preq J \Big\}\Big),
\eeqn
where
\begin{enumerate}

\item $({\mc I}, \preq)$ is a finite, partially ordered set.

\item For each $I\in {\mc I}$, $C_I$ is a virtual orbifold chart of virtual dimension $d$ on ${\mc X}$.

\item For $I \preq J$, $T_{JI}$ is a coordinate change from $C_I$ to $C_J$.
\end{enumerate}
They are subject to the following conditions.
\begin{itemize}
\item {\bf (Covering Condition)} ${\mc X}$ is covered by all the footprints $F_I$. 

\item {\bf (Cocycle Condition)} For $I \preq J \preq K \in {\mc I}$, denote $U_{KJI} = U_{KI} \cap \phi_{JI}^{-1} (U_{KJ}) \subset U_I$. Then we require that 
\beqn
\wh\phi_{KI}|_{U_{KJI}} = \wh\phi_{KJ} \circ \wh \phi_{JI}|_{U_{KJI}}
\eeqn
as bundle embeddings.

\item {\bf (Overlapping Condition)} For $I, J \in {\mc I}$, we have
\beqn
\ov{F_I} \cap \ov{F_J} \neq \emptyset \Longrightarrow I \preq J\ {\rm or}\ J \preq I.
\eeqn
\end{itemize}
\end{defn}

\begin{rem}
All virtual orbifold atlases considered in this paper have definite virtual dimensions, although sometimes we do not explicitly mention it. The discussion of orientations of virtual orbifold atlases has also been provided in \cite[Appendix A]{Tian_Xu_2021}. We do not repeat here. 
\end{rem}

\subsection{Good coordinate systems}

Now we introduce the notion of shrinkings of virtual orbifold atlases. 

\begin{defn}\cite[Definition A.24]{Tian_Xu_2021}
Let ${\mc A} = ( \{ C_I | I \in {\mc I} \}, \{ T_{JI} | I \preq J \})$ be a virtual orbifold atlas on ${\mc X}$.
\begin{enumerate}

\item  A {\it shrinking} of ${\mc A}$ is another virtual orbifold atlas ${\mc A}' = ( \{ C_I' | I \in {\mc I} \}, \{ T_{JI}' | I \preq J \})$ indexed by elements of the same partially ordered set ${\mc I}$ such that for each $I \in {\mc I}$, $C_I'$ is a shrinking $C_I|_{U_I'}$ of $C_I$ and for each $I \preq J$, $T_{JI}'$ is the induced shrinking of $T_{JI}$ given by Lemma \ref{lemma613}.

\item If for every $I \in {\mc I}$, $U_I'$ is a precompact subset of $U_I$, then we say that ${\mc A}'$ is a {\it precompact shrinking} of ${\mc A}$ and denote ${\mc A}' \sqsubset {\mc A}$. 

\end{enumerate}
\end{defn}

Given a virtual orbifold atlas ${\mc A}= ( \{ C_I | I \in {\mc I} \}, \{ T_{JI} | I \preq J \} )$, 
we define a relation $\curlyvee$ on the disjoint union $\bigsqcup_{I \in {\mc I}} U_I$ as follows. $U_I \ni x \curlyvee y\in U_J$ if one of the following holds.
\begin{enumerate}
\item $I = J$ and $x = y$;

\item $I \preq J$, $x \in U_{JI}$ and $y = \phi_{JI}(x)$;

\item $J \preq I$, $y \in U_{IJ}$ and $x = \phi_{IJ}(y)$.
\end{enumerate}
For an atlas ${\mc A}$, if $\curlyvee$ is an equivalence relation, we can form the quotient space
\beqn
|{\mc A}|:= \Big( \bigsqcup_{I \in {\mc I}} U_I \Big)/ \curlyvee.
\eeqn
with the quotient topology. There is a natural injective map 
\beqn
{\mc X} \hookrightarrow |{\mc A}|.
\eeqn
We call $|{\mc A}|$ the {\it virtual neighborhood} of ${\mc X}$ associated to the atlas ${\mc A}$ with the quotient map
\beq\label{eqn61}
\pi_{\mc A}: \bigsqcup_{I \in {\mc I}} U_I \to |{\mc A}|
\eeq
which induces continuous injections $U_I \hookrightarrow |{\mc A}|$. A point in $|{\mc A}|$ is denoted by $|x|$, which has certain representative $x \in U_I$ for some $I$. 

\begin{defn}\label{defn_gcs}\cite[Definition A.25]{Tian_Xu_2021}
A virtual orbifold atlas ${\mc A}$ on ${\mc X}$ is called a {\it good coordinate system} if the following conditions are satisfied.
\begin{enumerate}

\item $\curlyvee$ is an equivalence relation.

\item The virtual neighborhood $|{\mc A}|$ is a Hausdorff space.

\item For all $I\in {\mc I}$, the natural maps $U_I \to |{\mc A}|$ are homeomorphisms onto their images.
\end{enumerate}
\end{defn}

The conditions for good coordinate systems are very useful for later constructions (this is the same as in the Kuranishi approach, see \cite{FOOO_2016}), for example, the construction of suitable multisection perturbations. In these constructions, the above conditions are often implicitly used without explicit reference. Therefore, an important step is to construct good coordinate systems.

\begin{thm}\label{thm618} {\rm (Constructing good coordinate system)}\cite[Theorem A.26]{Tian_Xu_2021}
Let ${\mc A}$ be a virtual orbifold atlas on ${\mc X}$ with the collection of footprints $F_I$ indexed by $I \in {\mc I}$. Let $F_I^\square \sqsubset F_I$ for all $I\in {\mc I}$ be a collection of precompact open subsets such that 
\beqn
{\mc X} = \bigcup_{I \in {\mc I}} F_I^\square.
\eeqn
Then there exists a shrinking ${\mc A}' $ of ${\mc A}$ such that the collection of shrunk footprints $F_I'$ contains $\ov{F_I^\square}$ for all $I \in {\mc I}$ and ${\mc A}'$ is a good coordinate system. Moreover, if ${\mc A}$ is already a good coordinate system, then any shrinking of ${\mc A}$ remains a good coordinate system. 
\end{thm}

\begin{rem}
A similar result was proved earlier in \cite{FOOO_2016}. 
\end{rem}

\subsection{Perturbations}\label{subsection65}

Now we define the notion of perturbations. 

\begin{defn}\cite[Definition A.32]{Tian_Xu_2021}
Let ${\mc A}$ be a good coordinate system on ${\mc X}$. 

\begin{enumerate}

\item A {\it multi-valued perturbation} of ${\mc A}$, simply called a {\it perturbation}, denoted by ${\mf t}$, consists of a collection of multi-valued continuous sections 
\beqn
t_I: U_I \overset{m}{\to} E_I
\eeqn
satisfying (as multisections)
\beqn
t_J \circ \phi_{JI} = \wh \phi_{JI} \circ t_I|_{U_{JI}}.
\eeqn

\item Given a multi-valued perturbation ${\mf t}$, the object
\beqn
\tilde {\mf s} = \Big( \tilde s_I = S_I + t_I: U_I \overset{m}{\to} E_I \Big)
\eeqn
satisfies the same compatibility condition with respect to coordinate changes. The perturbation ${\mf t}$ is called {\it transverse} if every $\tilde s_I$ is a transverse multisection. 


\item The zero locus of a perturbed $\tilde {\mf s}$ gives objects in various different categories. Denote
\beqn
{\mc Z} = \bigsqcup_{I \in {\mc I}}  \tilde s_I^{-1}(0)
\eeqn
and the quotient $|{\mc Z}|:= {\mc Z}/ \curlyvee$ equipped with the quotient topology. Furthermore, there is a natural injection $|{\mc Z}|  \hookrightarrow |{\mc A}|$. Denote by $\| {\mc Z}\|$ the same set as $|{\mc Z}|$ but equipped with the topology as a subspace of $|{\mc A}|$. 
\end{enumerate}
\end{defn}

In order to construct suitable perturbations of a good coordinate system, we need certain tubular neighborhood structures with respect to coordinate changes. 

\begin{defn}\label{defn621}\cite[Definition A.33]{Tian_Xu_2021}
Let ${\mc A}$ be a good coordinate system with charts indexed by elements in a finite partially ordered set $({\mc I}, \preq)$ and coordinate changes indexed by pairs $I \preq J \in {\mc I}$. A {\it normal thickening}\footnote{It was called {\it thickening} in \cite{Tian_Xu_2021}.} of ${\mc A}$ is a collection 
\beqn
\big\{ (N_{JI}, E_{I; J})\ |\ I \preq J \big\}
\eeqn
where $N_{JI} \subset U_J$ is an open neighborhood of $\phi_{JI}(U_{JI})$ and $E_{I; J}$ is a subbundle of $E_J|_{N_{JI}}$. They are required to satisfy the following conditions. 
\begin{enumerate}

\item If $I \preq K$, $J \preq K$ but there is no partial order relation between $I$ and $J$, then 
\beq\label{eqn62}
N_{KI} \cap N_{KJ} = \emptyset.
\eeq

\item For all triples $I \preq J \preq K$, 
\beqn
E_{I;J}|_{\phi_{KJ}^{-1}(N_{KI}) \cap N_{JI} } = \wh\phi_{KJ}^{-1}(E_{I;K})|_{\phi_{KJ}^{-1}(N_{KI}) \cap N_{JI}}.
\eeqn

\item For all triples $I \preq J \preq K$, one has
\beqn
E_{I; K}|_{N_{KI} \cap N_{KJ}} \subset E_{J; K}|_{N_{KI} \cap N_{KJ}}.
\eeqn

\item Each $(N_{JI}, E_{I;J})$ satisfies the {\bf (Tangent Bundle Condition)} of Definition \ref{defn611}.
\end{enumerate}
\end{defn}

\begin{defn}\cite[Definition A.34]{Tian_Xu_2021} Suppose ${\mc A}$ is equipped with a normal thickening ${\mc N}$ as in Definition \ref{defn621}. We say that a perturbation ${\mf t}$ is ${\mc N}$-normal if
\beq\label{eqn63}
t_J(N_{JI}) \subset E_{I; J}|_{N_{JI}},\ \forall I \preq J
\eeq
\end{defn}

\begin{rem}\label{rem623}
The above setting is slightly more general than what we need in this paper. In this paper we will see the following situation in the concrete cases.
\begin{enumerate}

\item The index set ${\mc I}$ consists of certain nonempty subsets of a finite set $\{1, \ldots, m \}$, which has a natural partial order given by inclusions. 

\item For each $i \in I$, $\Gammait_i$ is a finite group and $\Gammait_I =  \mathit{\Pi}_{i\in I} \Gammait_i$. $U_I = \tilde U_I/ \Gammait_I$ where $\tilde U_I$ is a topological manifold acted by $\Gammait_I$. Moreover, ${\bm E}_1, \ldots, {\bm E}_m$ are vector spaces acted by $\Gammait_i$ and the orbifold bundle $E_I \to U_I$ is the quotient 
\beqn
E_I:= (\tilde U_I \times {\bm E}_I ) / \Gammait_I,\ {\rm where}\ {\bm E}_I:= \bigoplus_{i \in I} {\bm E}_i.
\eeqn

\item For $I \preq J$, $U_{JI} = \tilde U_{JI} / \Gammait_I$ where $\tilde U_{JI} \subset \tilde U_I$ is a $\Gammait_I$-invariant open subset and the coordinate change is induced from the following diagram
\beqn
\vcenter{ \xymatrix{ \tilde V_{JI} \ar[r] \ar[d] & \tilde U_J \\
                     \tilde U_{JI}              & }   }
\eeqn
Here $\tilde V_{JI} \to \tilde U_{JI}$ is a covering space with group of deck transformations identical to $\Gammait_{J-I} = \mathit{\Pi}_{j \in J - I} \Gammait_j$; then $\Gammait_J$ acts on $\tilde V_{JI}$ and $\tilde V_{JI} \to \tilde U_J$ is a $\Gammait_J$-equivariant embedding of manifolds, which induces an orbifold embedding $U_{JI} \to U_J$ and an orbibundle embedding $E_I|_{U_{JI}} \to E_J$. 

\end{enumerate}

In this situation, one naturally has subbundles $E_{I; J} \subset E_J$ for all pairs $I \preq J$. Hence a normal thickening of such a good coordinate system is essentially only a collection of neighborhoods $N_{JI}$ of $\phi_{JI}(U_{JI})$ which satisfy \eqref{eqn62} and 
\beqn
S_J^{-1}(E_{I;J}) \cap N_{JI} = \phi_{JI}(U_{JI}).
\eeqn
\end{rem}

\begin{thm}\label{thm624}\cite[Theorem A.35]{Tian_Xu_2021}
Let ${\mc X}$ be a compact Hausdorff space equipped with a good coordinate system 
\beqn
{\mc A} = \Big( \big\{ C_I = (U_I, E_I, S_I, \psi_I, F_I) \ |\ I \in {\mc I} \big\},\ \big\{ T_{JI}  \ |\ I \preq J \big\} \Big).
\eeqn
Let ${\mc A}' \sqsubset {\mc A}$ be any precompact shrinking. Let ${\mc N} = \{(N_{JI}, E_{I;J})\}$ be a normal thickening of ${\mc A}$ and $N_{JI}'\subset U_J'$ be a collection of open neighborhoods of $\phi_{JI}(U_{JI}')$ such that $\ov{N_{JI}'} \subset N_{JI}$. Then they induce a normal thickening ${\mc N}'$ of ${\mc A}'$ by restriction. Let $d_I: U_I \times U_I \to [0, +\infty)$ be a distance function on $U_I$ which induces the same topology as $U_I$. Let $\epsilon>0$ be a constant. Then there exists a perturbation ${\mf t}$ of ${\mc A}$ consisting of a collection of multisections $t_I: U_I \overset{m}{\to} E_I$ satisfying the following conditions.
\begin{enumerate}

\item For each $I\in {\mc I}$, $\tilde s_I:= S_I + t_I$ is transverse.

\item For each $I \in {\mc I}$, 
\beq\label{eqn64}
d_I \big( \tilde s_I^{-1}(0) \cap \ov{U_I'}, S_I^{-1}(0) \cap \ov{U_I'} \big) \leq \epsilon.
\eeq

\item For each pair $I \preq J$, we have 
\beqn
\wh \phi_{JI} \circ t_I|_{U_{JI}'} = t_J \circ \phi_{JI}|_{U_{JI}'}.
\eeqn
Hence the collection of restrictions $t_I':= t_I|_{U_I'}$ defines a perturbation ${\mf t}'$ of ${\mc A}'$. 

\item ${\mf t}'$ is ${\mc N}'$-normal. 

\end{enumerate}
\end{thm}

In the case when the virtual dimension is zero, the weighted count of the zero locus of a transverse perturbation gives a well-defined virtual count of the good coordinate system, denoted by $\# {\mc A}$ (see \cite[Appendix A.7, A.8]{Tian_Xu_2021}).

\subsection{Strongly continuous maps}\label{subsection66}

Let ${\mc X}$ be a compact Hausdorff space equipped with a $d$-dimensional topological virtual orbifold atlas without boundary 
\beqn
{\mc A} = \Big( (C_I)_{I \in {\mc I}}, (T_{JI})_{I \preq J} \Big).
\eeqn

\begin{defn}\label{defn_strong_continuous}(cf. \cite[Definition 3.40]{FOOO_Kuranishi})
A {\it strongly continuous map} from $({\mc X}, {\mc A})$ to a smooth orbifold $Y$, denoted by $\mf{e}: {\mc A} \to Y$, consists of a family of continuous maps $e_I: U_I \to Y$ such that for all pairs $I \preq J$, $e_J \circ \phi_{JI} = e_I|_{U_{JI}}$. The strongly continuous map $\mf{e}$ is called {\it weakly submersive} if for each $I$, $e_I: U_I \to Y$ is transverse to every point (which has a canonical normal microbundle in any orbifold local chart).
\end{defn}

Suppose we are in the situation of Definition \ref{defn_strong_continuous} where $\mf{e}$ is a strongly continuous map which is weakly submersive. If $Z$ is a smooth orbifold and $\iota_Z: Z \to Y$ is a smooth map, then for each chart $C_I$, the map $e_I: U_I \to Y$ is transverse to $Z$ (meaning that when locally lifting the map to orbifold charts it is transverse to $Z$, no matter which normal microbundle is chosen). Hence we can define the fibre product
\beqn
{\mc A}_Z:= {\mc A} \times_{Y} Z,
\eeqn
which is a topological virtual orbifold atlas on the (not necessarily compact) space
\beqn
{\mc X}_Z:= {\mc X} \times_Y Z
\eeqn
whose charts are
\beqn
C_{I, Z}:= C_I|_{U_{I, Z}},\ U_{I, Z}:= U_I\times_Y Z
\eeqn
and whose coordinate changes are naturally induced from $T_{JI}$. It is easy to see the following
\begin{enumerate}

    \item The strongly continuous map $\mf{e}$ induces a strongly continuous map $\mf{e}_Z: {\mc A}_Z \to Z$.

    \item If ${\mc A}$ is a good coordinate system, so is ${\mc A}_Z$.

    \item Any normal thickening of ${\mc A}$ restricts to a normal thickening of ${\mc A}_Z$.

\end{enumerate}

\subsection{The virtual fundamental class}\label{subsection67}

We expect that the zero locus of a transverse multi-valued perturbation carries a fundamental class in rational homology. This is indeed the case in the smooth category as discussed in \cite{Fukaya_Ono} and later in \cite{MW_2}. Indeed, one can find a triangulation of the zero locus. However, it is not so clear if such a triangulation still exists in the topological category. Instead, we only define the virtual fundamental class in the target space by reducing the problem to the case where the virtual dimension is at most $1$. This case was treated in the appendix of our previous paper \cite[Appendix]{Tian_Xu_2021}. 

We first need certain technical discussions about homology theory of smooth orbifolds. Let $Y$ be an $m$-dimensional compact oriented smooth orbifold. Notice that $Y$ satisfies the following topological properties:
\begin{enumerate}
    \item Poincar\'e duality holds over rational coefficients.

    \item Rational homology classes can be represented by piecewise smooth cubical cycles and the difference of two homologous cycles is the boundary of a piecewise smooth cubical chain.
\end{enumerate}
Here a piecewise smooth cubical cycle is a linear combination of smooth cubes, while a smooth cube (of dimension $k$) is a continuous map 
\beqn
f: \Box^k = [0,1]^k \to Y
\eeqn
which can be lifted to a smooth map into some orbifold chart of $Y$. For technical reasons we need to restrict the kind of smooth cubes we allow. We say that a cube $f: \Box^k \to Y$ is a {\it collared} cube if near each facet it is independent of the normal variables. Because every smooth cube is smoothly homotopic to a collared cube, rational homology classes of $Y$ can represented by piecewise smooth collared cubical cycles and the difference of two homologous cycles is the boundary of a piecewise smooth collared cubical chain. 

Now let ${\mc A}$ be an oriented good coordinate system of virtual dimension $d$ on a compact Hausdorff space ${\mc X}$ and let $\mf{e}: {\mc A} \to Y$ be a strongly continuous and weakly submersive map. We would like to define a homology class
\beqn
[{\mc X}]^{\rm vir} \in H_d(Y; {\mb Q})
\eeqn
although it surely depends on ${\mc A}$ and $\mf{e}$. By Poincar\'e duality, this class is uniquely determined by the intersection numbers 
\beqn
[{\mc X}]^{\rm vir} \cap \beta  \in {\mb Q}
\eeqn
for all $\beta \in H_{m-d} (Y; {\mb Q})$. We only need to define these intersection numbers.

In general $\beta$ can be represented by a piecewise smooth collared cubical cycle which is a linear combination of smooth collared cubes $\iota: Z = \Box^{m-d} \to Y$. For each facet $\partial_\alpha Z \subset Z$, 
one obtains a good coordinate system ${\mc A}_{\partial_\alpha Z}$ on the space ${\mc X}_{\partial_\alpha Z}$. The collared condition allows us to inductively construct transverse perturbations on ${\mc A}_{\partial_\alpha Z}$ for all facets of all cubes in this cycle. For dimensional reasons, the perturbation only vanishes at finitely many points on the interior of these simplices. The weighted count of these zeroes gives an intersection number. By the same argument as \cite[Appendix A.8]{Tian_Xu_2021} one can prove that this intersection number is independent of either the perturbations or the choice of the piecewise smooth cycles representing $\beta$. Hence the number $[{\mc X}]^{\rm vir} \cap \beta$ is well-defined.

From the definition of the virtual fundamental class one can easily see the following.

\begin{prop}\label{prop626}
Let ${\mc A}$ be an oriented good coordinate system on a compact Hausdorff space ${\mc X}$ equipped with a strongly continuous and weakly submersive map $\mf{e}: {\mc A} \to Y$ into a smooth oriented compact orbifold $Y$. Let $\iota_S: S \hookrightarrow Y$ be a smooth closed oriented suborbifold which induces a good coordinate system ${\mc A}_S$ on ${\mc X}_S= {\mc X}\times_Y S$. Then
\beq\label{eqn66}
(\iota_S)_* [ {\mc X}_S ]^{\rm vir} = [{\mc X}]^{\rm vir} \cap [S] \in H_*(Y; {\mb Q}).
\eeq
\end{prop}

\section{The Main Theorem on the Virtual Cycle}\label{section7}

In this section we restate the main theorem of this paper in full generality (cf. Theorem \ref{thm11}). Let ${\sf \Gamma}$ be a stable decorated dual graph  (not necessarily connected) with $n$ tails. Then we have defined the moduli space of gauge equivalence classes of solutions to the gauged Witten equation with combinatorial type $\sf\Gamma$, and its compactification, denoted by 
\beqn
\ov{\mc M}_{\sf \Gamma}:= \ov{\mc M}_{\sf \Gamma}(V, G, W, \mu). 
\eeqn
There are the natural {\it evaluation map}
\beqn
\ev: \ov{\mc M}_{\sf\Gamma} \to X^n
\eeqn
and the {\it forgetful map} or the {\it stabilization map}, denoted by 
\beqn
{\rm st}: \ov{\mc M}_{\sf \Gamma} \to \ov{\mc M}_{\Gamma} \subset \ov{\mc M}_{g, n}. 
\eeqn
Here $\Gamma$ is the underlying dual graph of $\sf\Gamma$, where unstable rational components are contracted, and ${\rm st}$ is defined by forgetting the $r$-spin structure, forgetting the gauged maps, and stabilization. Denote
\beq\label{eqn71}
\eev:= ( {\rm ev}, {\rm st}): \ov{\mc M}_{{\sf\Gamma}} \to X^{n} \times \ov{\mc M}_{\Gamma}.
\eeq
Notice that both $X^n$ and $\ov{\mc M}_{\Gamma}$ are compact complex orbifolds, hence satisfy Poincar\'e duality and have homological intersection pairings in rational coefficients. 

Now we state the main result about the virtual cycles. Recall the knowledge in Subsection \ref{subsection66} which says that if a compact Hausdorff space ${\mc X}$ has an oriented topological virtual orbifold atlas ${\mc A}$ equipped with a strongly continuous map to a manifold $Y$, then there is a well-defined pushforward virtual fundamental cycle in $H_*(Y; {\mb Q})$. 

\begin{thm}\label{thm71}
For each stable decorated dual graph ${\sf \Gamma}$ with $n$ tails (not necessarily connected), there exists an oriented topological virtual orbifold atlas ${\mc A}_{\sf \Gamma}$ on the compactified moduli space $\ov{\mc M}_{\sf \Gamma}$, and a strongly continuous and weakly submersive map 
\beqn
\mf{ev}: {\mc A}_{\sf \Gamma} \to X^n \times \ov{\mc M}{}_{\Gamma}.
\eeqn 
which extends the map \eqref{eqn71}. They further satisfy the following properties. All homology and cohomology groups are of rational coefficients. 
\begin{enumerate}

\item \label{thm71a} {\bf (Dimension)} If $\sf\Gamma$ is connected, then the virtual dimension of the atlas is 
\beqn
(2-2g) {\rm dim}_{\mb C} X + {\rm dim}_{\mb R} \ov{\mc M}_\Gamma + 2 \langle c_1^{\ubar K} (TV), \ubar B \rangle.
\eeqn
Here $\ubar B \in H_2^{\ubar K}(V;{\mb Z})$ is the equivariant curve class associated to $\sf\Gamma$.

\item \label{thm71b} {\bf (Disconnected Graph)} Let ${\sf \Gamma}_1, \ldots, {\sf \Gamma}_k$ be connected decorated dual graphs. Let ${\sf \Gamma} = {\sf \Gamma}_1 \sqcup \cdots \sqcup {\sf\Gamma}_k$ be the disjoint union. Then one has 
\beqn
\left[ \ov{\mc M}_{\sf \Gamma} \right]^{\rm vir} = \bigotimes_{\alpha=1}^k \left[ \ov{\mc M}_{ {\sf \Gamma}_\alpha} \right]^{\rm vir}. 
\eeqn
Here the last equality makes sense with respect to the K\"unneth decomposition 
\beqn
H_* ( X^n \times \ov{\mc M}_{\Gamma} ) \cong \bigotimes_{\alpha=1}^k H_* ( X^{n_\alpha} \times \ov{\mc M}_{ {\Gamma}{}_\alpha} ).
\eeqn

\item \label{thm71c} {\bf (Cutting Edge)} Let $\sf \Gamma$ be the decorated dual graph and let $\sf \Pi$ be the decorated dual graph obtained from $\sf \Gamma$ by shrinking a loops. Then one has 
\beqn
(\iota_\Pi)_* [ \ov{\mc M}_{\sf \Pi} ]^{\rm vir} = [ \ov{\mc M}_{\sf \Gamma} ]^{\rm vir} \cap [ \ov{\mc M}_{\Pi} ];
\eeqn
Notice that $X^{\sf \Gamma} = X^{\sf \Pi}$, and the map 
\beqn
(\iota_{\Pi} )_*: H_* ( X^n ) \otimes H_* ( \ov{\mc M}_{\Pi} ) \to H_* ( X^n \big) \otimes H_* ( \ov{\mc M}_{\Gamma} )
\eeqn
is induced from the inclusion $\iota_{\Pi}: \ov{\mc M}_{\Pi} \hookrightarrow \ov{\mc M}_{\Gamma}$. 

\item \label{thm71d} {\bf (Composition)} Let $\sf \Pi$ be a decorated dual graph with one edge. Let $\tilde {\sf \Pi}$ be the decorated dual graph obtained from $\sf \Pi$ by normalization. Let $\Delta_X \subset X \times X$ be the diagonal. Notice that there is a natural isomorphism 
\beqn
\ov{\mc M}_{\Pi} \cong \ov{\mc M}_{{\tilde \Pi}}.
\eeqn
Then one has
\beqn
[ \ov{\mc M}_{\sf \Pi}  ]^{\rm vir} =  [ \ov{\mc M}_{\tilde {\sf \Pi}}  ]^{\rm vir} \setminus {\rm PD} [ \Delta_X].
\eeqn
Here we use the slant product 
\beqn
\setminus: H_* ( X^n \times X^2 )  \otimes H^* ( X^2 )  \to H_*  ( X^n ).
\eeqn

\end{enumerate}
\end{thm}

\section{Constructing the Virtual Cycle. I. Fredholm Theory}\label{section8}

In this section we construct the virtual atlas in the infinitesimal level. Namely, we prove that the deformation complex of the moduli space of gauged Witten equation is Fredholm, and calculate the Fredholm index. 

Recall that moduli spaces can be labelled by the represented curve classes. For $\ubar B \in H_2^{\ubar K}(V; {\mb Z})$, consider the moduli space ${\mc M}_{g, n}^r(V, G, W, \mu; \ubar B)$. 
We would like to realize it (locally) as the zero set of certain Fredholm section, and calculate the Fredholm index. In particular, we prove the following formal statement. 

\begin{prop}\label{prop81}
The expected dimension of ${\mc M}_{g, n}^r (V, G, W, \mu; \ubar B)$ is equal to 
\beqn
(2-2g) {\rm dim}_{\mb C} X + {\rm dim}_{\mb R} \ov{\mc M}_{g, n}  + 2 \langle c_1^{\ubar K}(TV), \ubar B \rangle.
\eeqn
\end{prop}

\begin{rem}
We remark that the expected dimension is the same as the expected dimension of the moduli space of holomorphic curves into the classical vacuum.
\end{rem}

\subsection{The Banach manifold and the Banach bundle}\label{subsection81}

We first set up the Banach manifold ${\mc B}$ and the Banach vector bundle ${\mc E} \to {\mc B}$. We fix the following data and notations.

\begin{enumerate}

\item Fix a smooth $r$-spin curve ${\mc C}$. The only case for an unstable ${\mc C}$ is when it is an infinite cylinder. Let $\Sigma^*$ be the underlying punctured Riemann surface. 

\item A cylindrical metric on $\Sigma^*$ is fixed. Let $P_{\mc C} \to \Sigma^*$ be the $U(1)$-bundle coming from the Hermitian metric on the orbifold line bundle $L_{\mc C} \to {\mc C}$. 

\item Fix a $K$-bundle $P \to \Sigma^*$. Let $\ubar P \to \Sigma^*$ be the $\ubar K$-bundle defined by $P_{\mc C}$ and $P$. 

\item Fix trivializations of $\ubar P$ over certain cylindrical ends, and fix cylindrical coordinates on these ends. 

\end{enumerate}
When these data are understood from the context, we only use  $({\mc C}, P)$ to represent. 

To define various norms, we also need a Riemannian metric on $V$ whose exponential map has good properties. The original K\"ahler metric $g_V$ doesn't satisfy our need, for example, the level set of the moment map $\mu^{-1}(0)$ is not totally geodesic. We choose a metric $h_V$ on $V$ that satisfies the following conditions.
\begin{itemize}
\item $h_V$ is invariant under the $\ubar K$-action.

\item $h_V = g_V$ near the unstable locus of the $G$-action.

\item ${\rm Crit} W$ is totally geodesic near $\mu^{-1}(0)$. 

\item ${\rm Crit} W \cap \mu^{-1}(0)$ is totally geodesic. 

\item $J: TV \to TV$ is isometric.
\end{itemize}
We leave to the reader to verify the existence of such Riemannian metrics. 

Let $\exp$ be the exponential map associated to $h_V$. The metric $h_V$ induces a fibrewise metric on the associated bundle $\ubar P(V)$ and we also use $\exp$ to denote the exponential map in the vertical directions.

Now we can define the Banach manifolds. First, the cylindrical metric on $\Sigma^*$ allows us to define the weighted Sobolev spaces 
\beqn
W^{k, p, w}(\Sigma^*),\ \ \forall k \geq 0,\ p > 2,\ w \in {\mb R}. 
\eeqn
Now define 
\beqn
{\mc A}_{{\mc C}, P}^{k, p, w}
\eeqn
to be the space of connections $A$ on $P$ of regularity $W^{k,p}_{loc}$ satisfying the following conditions. Over each cylindrical end $U_a$, with respect to the fixed trivializations of $P|_{U_a}$,
\beqn
A = d + \lambda_a d t + \alpha,\ {\rm where}\ \alpha \in W^{k, p, w}(U_a, T^* U_a \otimes {\mf k})
\eeqn
where $\lambda_a \in {\mf k}$ such that 
\beqn
e^{ 2\pi {\bf i} \lambda_a}  = \iota_W \left( e^{- \frac{2\pi m_a {\bf i}}{r} } \right) \in K
\eeqn
where $\iota_W: U(1) \to K$ is the group homomorphism assumed in Hypothesis \ref{hyp33}. On the other hand, define 
\beqn
{\mc S}_{{\mc C},P}^{k,p,w}
\eeqn
to be the space of sections $u:\Sigma^* \to \ubar P(V)$ of regularity $W^{k, p}_{loc}$, satisfying the following condition. Over each cylindrical end $U_a$, under the fixed trivializations of $P|_{U_a}$, we can identify $u$ as a map $u: U_a \to V$. Then we require that there exists a point $x_a \in \mu^{-1}(0) \cap {\rm Crit} W$ such that 
\beqn
u(s, t) = \exp_{x_a} \xi_a (s, t),\ {\rm where}\ \xi_a \in W^{k, p, w}(U_a, x_a^* TV).
\eeqn
(Here to differentiate $\xi_a$, we use the Levi--Civita connection of certain Riemannian metric on $V$. The resulting Banach manifold is independent of choosing the metric.) We also define a Banach manifold of gauge transformations. Let
\beqn
{\mc G}_{{\mc C}, P}^{k, p, w}
\eeqn
be the space of gauge transformations $g: P \to K$ of regularity $W^{k + 1,p}_{loc}$, such that over the cylindrical end $U_a$, with respect to the trivialization of $P|_{U_a}$, $g$ is identified with a map 
\beqn
g =g_a e^{h},\ {\rm where}\ g_a \in K,\ h \in W^{k + 1,p,w }(U_a, {\mf k}).
\eeqn
Then we define 
\beqn
{\mc B}_{{\mc C}, P}^{k, p, w} = {\mc A}_{{\mc C}, P}^{k, p, w} \times {\mc S}_{{\mc C}, P}^{k, p, w}.
\eeqn

In most cases, we will take $k = 1$ and in this case we will drop the index $k$ from the notations. We also drop ${\mc C}$ and $P$ from the notation temporarily. There is a smooth action by the group ${\mc G}^{p, w}$ on ${\mc B}^{p, w}$. Further, notice that for any $(A, u) \in {\mc B}^{p,w}$, $u$ is continuous. Hence an element represents an equivariant curve class in $H_2^{\ubar K}(V; {\mb Z})$. Then we can decompose the Banach manifold by the equivariant curve classes, i.e.,
\beqn
{\mc B}^{p, w} = \bigsqcup_{\ubar B \in H_2^{\ubar K} (V; {\mb Z})} {\mc B}^{p, w}(\ubar B).
\eeqn
To simplify the notations, we do not label the Banach manifolds by the curve classes. Next we define a Banach vector bundle 
\beqn
{\mc E}^{p, w} \to {\mc B}^{p, w}.
\eeqn
Over each point ${\bm v} = (A, u)$ of the Banach manifold, we have the vector bundle 
\beqn
u^* \ubar P(TV) \to \Sigma
\eeqn	
whose transition functions are of class $W^{1, p}_{loc}$. It also has an induced Hermitian metric. Fix a metric connection on $u^* \ubar P(TV)$ which is trivialized over cylindrical ends, we define
\beqn
{\mc E}^{p, w}|_{{\bm v}} = L^{p, w}(\Sigma^*, \Lambda^{0,1} \otimes u^* \ubar P(TV)) \oplus L^{p, w}(\Sigma^*, {\rm ad} P).
\eeqn
Then the union over all ${\bm v} \in {\mc B}^{p, w}$ gives a smooth Banach space bundle. An action by the group ${\mc G}^{p, w}$ makes it an equivariant bundle. It is straightforward to verify that the gauged Witten equation \eqref{eqn34} defines a smooth and ${\mc G}^{p,w}$-equivariant section 
\beqn
{\mc F}: {\mc B}^{p, w} \to {\mc E}^{p, w}.
\eeqn
Hence we define
\beqn
\tilde {\mc M}_{{\mc C}, P}^{p, w}:= {\mc F}^{-1}(0).
\eeqn
It has the topology induced as a subset of the Banach manifold ${\mc B}^{p, w}$. Then define the quotient which is equipped with the quotient topology:
\beqn
{\mc M}_{{\mc C}, P}^{p, w}:= \tilde {\mc M}_{{\mc C}, P}^{p, w}/ {\mc G}^{p, w}.
\eeqn

\begin{prop} Let ${\mc C}, P, p, w$ be as above.
\begin{enumerate}

\item ${\mc B}^{p, w}$ is a Banach manifold and ${\mc E}^{p, w}$ is a Banach vector bundle over it. 

\item Suppose ${\bm v} = (A, u) \in {\mc B}^{p, w }$ and it is asymptotic to a point $x_a \in \mu^{-1}(0) \cap {\rm Crit} W$ at the $a$-th puncture, then the tangent space of ${\mc B}^{p, w}$ at ${\bm v} = (A, u)$ can be identified with the Banach space
\beqn
W^{1, p, w}(\Sigma^*, \Lambda^1 \otimes {\mf k} \oplus u^*  \ubar P(TV) ) \oplus \bigoplus_{a =1}^n T_{x_a} ({\rm Crit} W \cap \mu^{-1}(0))).
\eeqn

\item For any smooth solution ${\bm v} = (A, u)$ to the gauged Witten equation over ${\mc C}$, there exists a smooth gauge transformation $g$ on $P$ such that $g^* (A, u) \in {\mc B}^{p, w }$.

\item The quotient topology ${\mc M}_{{\mc C}, P}^{p, w}$ is homeomorphic to the topology induced from the notion of convergence. 
\end{enumerate}
\end{prop}


\subsection{Gauge fixing}

In this subsection we take a digress onto the discussion of gauge fixing. This discussion is not only relevant to the index calculation, but is also useful in the gluing construction. Let ${\mc C}$, $P$ be as above and we have defined the Banach manifold ${\mc B}^{p, w}$, the Banach space bundle ${\mc E}^{p, w}$ and the section ${\mc F}$. 

First we would like to define a local distance function on ${\mc B}^{p, w}$. Fix a reference point ${\bm v} = (A, u) \in {\mc B}^{p, w}$, for another $(A', u') \in {\mc B}^{p,w}$, suppose $u$ and $u'$ are $C^0$ close and
\beqn
u' = \exp_{u} \xi,\ {\rm where}\ \xi \in W^{1, p}_{loc}(\Sigma^*, u^* \ubar P(TV) ).
\eeqn
Then we define 
\begin{multline}
{\rm dist} ({\bm v}, {\bm v}') = \| A' - A \|_{L^{p, w}} + \| D_A (A' - A)\|_{L^{p, w}} \\
 + \| \xi \|_{L^\infty} + \| D_A \xi \|_{L^{p, w}} + \| d\mu(u) \xi \|_{L^{p, w}} + \| d\mu(u) J \xi \|_{L^{p, w}}.
\end{multline}
As before, $D_A \xi$ is actually the covariant derivative with respect to $\ubar A$. This is gauge invariant in the sense that if $g\in {\mc G}^{p, w}$, then ${\rm dist}(g^* {\bm v}, g^* {\bm v}') = {\rm dist} ({\bm v}, {\bm v}')$. This function gives a way to specify a natural open neighborhood of ${\bm v}$. We do not need to know whether ${\rm dist} ({\bm v}, {\bm v}')$ is equal to ${\rm dist} ({\bm v}', {\bm v})$ or not. For $\epsilon>0$ sufficiently small, denote temporarily 
\beqn
{\mc B}_{{\bm v}}^\epsilon = \Big\{ {\bm v}' \in {\mc B}^{p, w}\ |\ {\rm dist} ({\bm v}, {\bm v}') < \epsilon \Big\}.
\eeqn

\begin{defn}
Given $p > 2$, $w > 0$ and $\epsilon>0$. Suppose ${\bm v}, {\bm v}' \in {\mc B}^{p,w}$ such that ${\rm dist} ({\bm v}, {\bm v}') < \epsilon$. We say that ${\bm v}'$ is in {\it Coulomb gauge relative to ${\bm v}$}, if
\beq\label{eqn82}
d_A^{\, *} ( A' - A ) + d\mu ( u^{\rm mid} ) J \xi^{\rm mid} = 0.
\eeq 
Here $u^{\rm mid}$ is the unique map such that
\beqn
\exp_{u^{\rm mid}}^{-1} u + \exp_{u^{\rm mid}}^{-1} u' = 0
\eeqn
and $\xi^{\rm mid}$ is defined via
\beqn
\exp_{u^{\rm mid}} \xi^{\rm mid} = u'.
\eeqn
\end{defn}
If $P$ is trivialized locally so that $A$ and $A'$ are written as 
\begin{align*}
&\ A = d + \phi ds + \psi dt,\ &\ A' = d + \phi' ds + \psi' dt
\end{align*}
and $u$, $u'$ are identified with genuine maps into $X$, then \eqref{eqn82} reads
\begin{multline*}
\partial_s (\phi' - \phi) + [ \phi, \phi' - \phi] + \partial_t (\psi' - \psi) + [\psi, \psi'- \psi] + d\mu(u_\xi^{\rm mid}) J \xi \\
= \partial_s(\phi' - \phi) + [\phi, \phi'] + \partial_t (\psi' - \psi) + [\psi, \psi'] + d\mu(u_\xi^{\rm mid}) J \xi = 0.
\end{multline*}
If we switch $(A, u)$ and $(A', u')$, then the left hand side changes sign. Hence we obtain the following fact. 

\begin{lemma}
${\bm v}'$ is on the Coulomb slice through ${\bm v}$ if and only if ${\bm v}$ is on the Coulomb slice through ${\bm v}'$. 
\end{lemma}

Using the implicit function theorem it is routine to prove the following results (cf. \cite{Cieliebak_Gaio_Mundet_Salamon_2002}) for the case where the domain curve is compact). 

\begin{lemma}\label{lemma86} {\rm (Local slice theorem)}
Given ${\bm v} = (A, u) \in {\mc B}^{p,w}$. Then there exist $\epsilon > 0$, $\delta > 0$ and $C>0$ satisfying the following condition. For any ${\bm v}' = (A', u') \in {\mc B}_{\bm v}^{\epsilon}$, there exists a unique gauge transformation $g = e^h \in {\mc G}^{p,  w}$ satisfying  
\begin{align*}
&\ \| h \|_{W^{2,p,w}} \leq \delta,\ &\  d_A^{\, *} (A_h' - A) + d\mu(u_{\xi_h'}^{\rm mid}) \cdot J \xi_h' = 0,
\end{align*}
where $A_h':= (e^h)^* A'$ and $\xi_h'$ is defined by $\exp_u \xi_h' = u_h':= (e^h)^* u'$. Moreover, 
\beqn
\| h \|_{W^{2, p, w}} \leq C {\rm dist} ({\bm v}, {\bm v}').
\eeqn
\end{lemma}

\begin{proof}
Define
\beqn
{\mc R}_{A'}(h) = d_A^{\, *} (A_h' - A) + d\mu(u_{\xi_h'}^{\rm mid} )\cdot J \xi_h' \in L^{p, w}(\Sigma_{\mc C}, {\rm ad} P).
\eeqn
We look for zeroes of this map via the implicit function theorem. A crucial step is to show the invertibility of the linearization of ${\mc R}$. When $A = A'$ we have
\beqn
{\mc L}_A:= D_0 {\mc R}_A (h) = d_A^{\, *} d_A h + d\mu(u)\cdot J {\mc X}_h. 
\eeqn
This defines a Fredholm operator ${\mc L}_A: W^{2, p, w} (\Sigma, {\rm ad}P) \to L^{p,w}( \Sigma, {\rm ad} P)$ and it is easy to show the positivity of this operator. Hence ${\mc L}_A$ is invertible and bounded from below by a constant $C >0$ (all the constants in this proof only depends on the gauge equivalence class of ${\bm v}$). On the other hand, there exists $\epsilon >0$ such that if ${\rm dist}({\bm v}, {\bm v}') \leq \epsilon$, then 
\beqn
\| {\mc L}_{A'} - {\mc L}_A \| \leq \frac{C}{2}. 
\eeqn
Hence ${\mc L}_{A'}$ is invertible and its inverse is uniformly bounded. Then this lemma follows from the implicit function theorem and the smoothness of the map ${\mc R}_{A'}$.
\end{proof}

Lastly we consider the gauge fixing problem for solitons which has an additional subtlety. Let ${\bm v}$ be a soliton. Notice that by reprarametrization, one obtains a continuous family of solitons, locally parametrized by ${\bf w} \in {\mb C}$. For ${\bf w} = S + {\bf i} T$ which is small, let ${\bm v}_{\bf w} = (u_{\bf w}, \phi_{\bf w}, \psi_{\bf w})$ be the reparametrized soliton, namely
\beqn
u_{\bf w}(s, t) = u(s + S, t + T),\ \phi_{\bf w}(s, t) = \phi(s + S, t + T),\ \psi_{\bf w}= \psi(s + S, t + T).
\eeqn
For another ${\bm v}'$ which is close to ${\bm v}$, the gauge fixing conditions relative to any ${\bm v}_{\bf w}$ belonging to this family are different.

In particular, when ${\bf w}$ is small, by Lemma \ref{lemma86}, one can gauge transform ${\bm v}_{\bf w}$ to the Coulomb slice through ${\bf v}$. Such gauge transformation is of the form $e^{h_{\bf w}}$ with $h_{\bf w} \in W^{2,p,w}(\Thetait, {\mf k})$ small and unique. The following lemma will be useful.

\begin{lemma}\label{lemma87}
Suppose ${\bm v} = (u, \phi, \psi)$, then 
\begin{align*}
&\ \left. \frac{\partial h_{\bf w}}{\partial s} \right|_{{\bf w} = 0} = \phi,\ &\ \left. \frac{\partial h_{\bf w}}{\partial t} \right|_{{\bf w} = 0} = \psi.
\end{align*}
\end{lemma}

\begin{proof}
The implicit function theorem follows from the method of Newton iteration. For ${\bf w}$ small enough, $h_{\bf w}$ is close to the unique solution to 
\beqn
\Delta_A^* h + d\mu(u) \cdot J {\mc X}_h = - d_A^* \big( (\phi_{\bf w} - \phi) ds + (\psi_{\bf w} - \psi) dt \big) + d\mu(u_{\bf w}^{\rm mid}) J \xi_{\bf w},
\eeqn
where $\xi_{\bf w}$ and $u_{\bf w}^{\rm mid}$ are  defined by 
\begin{align*}
&\ \exp_u \xi_{\bf w} = u_{\bf w},\ &\ u_{\bf w}^{\rm mid} = \exp_u \left( \frac{\xi_{\bf w}}{2} \right).
\end{align*}
The difference from $h_{\bf w}$ and the actual solution $h$ is of higher order in ${\bf w}$. Hence differentiating in ${\bf w}$ in the $s$-direction, one has
\beqn
\Delta_A^* \partial_s h_{\bf w} + d\mu(u) J {\mc X}_{\partial_s h_{\bf w}} = d_A^* \big( \partial_s \phi ds + \partial_s \psi dt \big) - d\mu(u) \cdot J \partial_s u.
\eeqn
We only verify that $\partial_s h_{\bf w} = \phi$ solves the above equation; the case for the $t$-derivative is similar. Indeed, the difference between the left hand side and the right hand side reads
\beqn
\begin{split}
&\ - (\partial_s + {\rm ad} \phi)( \partial_s \phi) - (\partial_t + {\rm ad} \psi) (\partial_t \phi + [\psi, \phi]) + d\mu(u) \cdot J {\mc X}_\phi\\
&\ + (\partial_s + {\rm ad} \phi) (\partial_s \phi) + (\partial_t + {\rm ad}\psi) (\partial_s \psi) + d\mu(u) \cdot J \partial_s u\\
= &\ ( \partial_t + {\rm ad} \psi) ( \partial_s \psi - \partial_t \phi + [\phi, \psi]) + d\mu(u) \cdot J (\partial_s u + {\mc X}_\phi) \\
= &\ (\partial_t + {\rm ad} \psi) \big( \partial_s \psi - \partial_t \phi + [\phi, \psi]  + \mu(u) \big).
\end{split}
\eeqn
The last line vanishes by the vortex equation. 
\end{proof}

\subsection{The index formula}

Now we form the deformation complex of our problem. Fix a smooth $r$-spin curve ${\mc C}$ of genus $g$ with $n$ marked points. Let ${\bm v} = (P, A, u)$ be a solution to the gauged Witten equation over ${\mc C}$ representing the curve class $\ubar B$. 

Consider the following {\it deformation complex}
\beqn
{\mf C}_{\bm v}: \xymatrix{{\mc E}_{\bm v}^0 \ar[r]^{\partial_{\bm v}} & {\mc E}_{\bm v}^1 \ar[r]^{{\mc D}_{\bm v}} & {\mc E}_{\bm v}^2}.
\eeqn
Here 
\beqn
{\mc E}_{\bm v}^0 = W^{2, p, w}(\Sigma^*, {\rm ad} P) \oplus \bigoplus_{a=1}^n {\mf k}
\eeqn
parametrizes infinitesimal gauge transformations $h: \Sigma^* \to {\rm ad} P$ which are asymptotic to constants at cylindrical ends; ${\mc E}_{\bm v}^1$ is the tangent space of ${\mc B}^{p, w}$ at ${\bm v}$; ${\mc E}_{{\bm v}}^2$ is the fibre of ${\mc E}^{p, w}$ at ${\bm v}$. The map ${\mc D}_{\bm v}$ is the linearization of the gauged Witten equation at ${\bm v}$, while $\partial_{\bm v}$ is the infinitesimal gauge transformation given by 
\beqn
\partial_{\bm v}(h) = ( - d_A h, {\mc X}_h(u)).
\eeqn
Since the gauged Witten equation is gauge invariant, ${\mc D}_{\bm v} \circ \partial_{\bm v} = 0$. Formally, the kernel of ${\mc D}_{\bm v}$ is the tangent space of the solution space (without taking the quotient by gauge transformations) and the complement of the image of $\partial_{\bm v}$ within this kernel is the tangent space of the moduli space (over a fixed curve).

\begin{thm}\label{thm88}
The deformation complex ${\mf C}_{\bm v}$ is Fredholm. Moreover, its Euler characteristic, which coincides with the expected dimension of the moduli space, is
\beq\label{index}
\chi({\mf C}_{\bm v}) = 
(2-2g) {\rm dim}_{\mb C} X + 2 \langle c_1^{\ubar K} (TV), \ubar B \rangle.
\eeq
\end{thm}

\begin{proof}
As $u$ is holomorphic and contained in $\ov{({\rm Crit} W)^{\rm ss}}$, it passes through the unstable locus at isolated points. It is not hard to see that it suffices to consider the case when $u$ is contained in the unstable locus $({\rm Crit} W)^{\rm ss}$ which is a smooth submanifold. In this case ${\bm v}$ is also a solution to the vortex equation with target space $({\rm Crit} W)^{\rm ss}$ (with a compact perturbation, see \eqref{induced_vortex}). 

Now the space ${\mc E}_{\bm v}^1$ can be decomposed as 
\beqn
{\mc E}_{\bm v}^1 = ({\mc E}_{\bm v}^1)' \oplus ({\mc E}_{\bm v}^1)''
\eeqn
where $({\mc E}_{\bm v}^1)'$ parametrizes deformations in the tangent direction of ${\rm Crit} W$ (together with deformations of the connection $A$) and $({\mc E}_{\bm v}^1)''$ parametrizes deformations in the normal direction of ${\rm Crit} W$. We can similarly decompose 
\beqn
{\mc E}_{\bm v}^2 = ({\mc E}_{\bm v}^2)' \oplus ({\mc E}_{\bm v}^2)''.
\eeqn
Then the subcomplex
\beqn
{\mf C}_{\bm v}': \xymatrix{ {\mc E}_{\bm v}' \ar[r]^{\partial_{\bm v}} & ({\mc E}_{\bm v}^1)' \ar[r]^{{\mc D}_{\bm v}} & ({\mc E}_{\bm v}^2)'' }
\eeqn
represents deformations of ${\bm v}$ as a vortex in $({\rm Crit} W)^{\rm ss}$, while the operator
\beq\label{normal_operator}
{\mc D}_{\bm v}: ({\mc E}_{\bm v}^1)'' \to ({\mc E}_{\bm v}^2)''
\eeq
represents deformations of ${\bm v}$ in the normal direction. The index formula for symplectic vortices over compact surfaces is given in \cite[Proposition 4.6]{Cieliebak_Gaio_Mundet_Salamon_2002} and the case for surfaces with cylindrical ends can be derived from it by taking into account of the asymptotic behaviors at infinities (see for example \cite[Section 4]{Xu_VHF} for the case with vortex Floer homology). In the current situation, this index formula gives the Euler characteristic 
\beq\label{tangent_index}
\chi({\mf C}_{\bm v}') = (2-2g) {\rm dim}_{\mb C} X + 2 {\rm deg} ( u^* \ubar P( T {\rm Crit} W) ).
\eeq

It remains to consider the index of \eqref{normal_operator}. This operator is the Cauchy--Riemann operator on the pullback bundle $u^* \ubar P( N{\rm Crit} W)$ plus the zero-th order term coming from the Hessian of the superpotential. Hence by the Riemann--Roch formula, its index should be twice of the degree of $u^* \ubar P(N{\rm Crit}W)$ plus contributions from the asymptotics of the Hessian of $W$. As $W$ is a holomorphic Morse function, its Hessian at the critical submanifold, as a real operator, has the same number of positive eigenvalues and negative eigenvalues. Therefore the asymptotic contribution is zero. So
\beq\label{normal_index}
{\rm Index} ( {\mc D}_{\bm v}: ({\mc E}_{\bm v}^1)'' \to ({\mc E}_{\bm v}^2)'') = 2 {\rm deg} (u^* \ubar P(N{\rm Crit} W)).
\eeq
Then \eqref{index} follows from \eqref{tangent_index} and \eqref{normal_index}. 
\end{proof}

Proposition \ref{prop81} follows immediately since there are additional parameters coming from the Deligne--Mumford space.

\section{Constructing the Virtual Cycle. II. Gluing}\label{section9}

In this section we construct a local virtual orbifold chart for every point of the compactified moduli space $\ov{\mc M}{}_{g, n}^r(V, G, W, \mu)$. The main result of this section is stated as Corollary \ref{cor917}, while a crucial step is Proposition \ref{prop916}, in which via the gluing construction we obtain local charts for the {\it thickened moduli space}.

\subsection{Stabilizers of solutions}

An automorphism of a solution {\it a priori} consists of two parts, one coming from the domain symmetry, i.e., the underlying $r$-spin curve has an automorphism, the other coming from the gauge symmetry. To construct a local model, we first show that any stable solution has finite automorphism group.

\begin{defn}\label{defn91} {\rm (Isomorphisms of stable solutions)} 
Let ${\mc C}_1$ resp. ${\mc C}_2$ be a smooth or nodal $r$-spin curve and ${\bm v}_1$ resp. ${\bm v}_2$ be a solution to the gauged Witten equation over ${\mc C}_1$ resp. ${\mc C}_2$.  An isomorphism from $({\mc C}_1, {\bm v}_1)$ to $({\mc C}_2, {\bm v}_2)$ is a pair $(\rho, g)$, where $\rho: {\mc C}_1 \to {\mc C}_2$ is an isomorphism of $r$-spin curves, and $g: P_1 \to P_2$ is a smooth isomorphism of principal $K$-bundles that covers $\rho$. We require that $(\rho, g)$ satisfies the following condition. 

\begin{itemize}

\item For each irreducible component $\alpha_1$ of ${\mc C}_1$, the isomorphism $\rho$ of $r$-spin curves provides a corresponding irreducible component $\alpha_2$ of ${\mc C}_2$ and a biholomorphic map $\gamma_{\alpha_1}: \Sigma_{\alpha_1}^* \to \Sigma_{\alpha_2}^*$ that preserves the punctures. Moreover, there is a bundle isomorphism $g_{\alpha_1}: P|_{\Sigma_{\alpha_1}^*} \to P |_{\Sigma_{\alpha_2}^*}$ that covers $\gamma_{\alpha_1}$. We require that 
\begin{align*}
&\ A_1 |_{\Sigma_{\alpha_1}^*} = g_{\alpha_1}^* A_2 |_{\Sigma_{\alpha_2}^*},\ &\ u_1|_{P_1|_{\Sigma_{\alpha_1}^*}} = u_2 |_{P_2|_{\Sigma_{\alpha_2}^*}} \circ g_{\alpha_1}.
\end{align*}
\end{itemize}
\end{defn}

In particular, a gauge transformation on $P$ induces an isomorphism from $({\mc C}, {\bm v})$ to $({\mc C}, g^* {\bm v})$. Isomorphisms can be composed. For any pair $({\mc C}, {\bm v})$, denote by ${\rm Aut}({\mc C}, {\bm v})$ the automorphism (self-isomorphism) group. 

Then we show that automorphism groups of stable solutions are finite. A key point is that when $n\geq 1$, a solution converges at punctures to points in $\mu^{-1}(0)$ where $K$-stabilizers are trivial. When the curve has no punctures, we need to require that the area is sufficiently large (compare to the energy) so that solutions are also close to $\mu^{-1}(0)$. 

\begin{lemma}\label{lemma92} {\rm (Finiteness of automorphism groups)}
\begin{enumerate}

\item Suppose $n \geq 1$. Then for every $\ubar B \in H_2^{\ubar K}( V; {\mb Z})$, every stable solution ${\bm v}$ to the gauged Witten equation over a smooth or nodal curve ${\mc C}$ representing a point in $\ov{\mc M}{}_{g, n}^r(V, G, W, \mu; \ubar B)$, the automorphism group $\Gammait_{\bm v}$ is finite.

\item Suppose $n = 0$. Then for every $\ubar B \in H_2^{\ubar K} (V; {\mb Z})$ there exists $\lambda(\ubar B) >0$ such that if the areas of all genus $g$ $r$-spin curve (with no punctures) are no less than $\lambda(\ubar B)$, then every stable solution ${\bm v}$ to the gauged Witten equation representing a point in ${\mc M}{}_{g, 0}^r(V, G, W, \mu; \ubar B)$, the automorphism group $\Gammait_{\bm v}$ is finite. 
\end{enumerate}

\end{lemma}

\begin{proof}
Consider $({\mc C}, {\bm v})$ and its automorphism group ${\rm Aut}({\mc C}, {\bm v})$. Let ${\rm Aut}_0({\mc C}, {\bm v}) \subset {\rm Aut}({\mc C}, {\bm v})$ be the subgroup of automorphisms $(\rho, g)$ which fix every irreducible component of ${\mc C}$. Since there are at most finitely many permutations among components, it suffices to prove that ${\rm Aut}_0({\mc C}, {\bm v})$ is finite. Hence we only need to prove the finiteness for a smooth $r$-spin curve ${\mc C}$. 

First consider the case that ${\mc C}$ is stable. Then $\rho$ is an automorphism of the stable $r$-spin curve, which is finite. Then it suffices to consider automorphisms $(\rho,  g)$ such that $\rho = {\rm Id}_{\mc C}$. Since $g$ fixes the connection, $g$ is a covariantly constant section of $P \times_{\rm Ad} K$. Hence such $g$ corresponds to a subgroup of $K$. On the other hand, if ${\mc C}$ has a marking, at which the section $u$ approaches to $\mu^{-1}(0)$. Since $g$ has to fix the section, and stabilizers of $\mu^{-1}(0)$ are trivial. Hence there are no nontrivial gauge transformations that fix ${\bm v}$. 

If ${\mc C}$ is stable but has no marking, then $\Sigma_{\mc C}$ has finite volume. By the energy inequality (see \cite[Theorem 4.4]{Tian_Xu_geometric}), for some constant $C>0$ one has 
\beqn
C + \langle \omega_V^{\ubar K}, \ubar B \rangle  = E({\bm v}) \geq \int_{\Sigma_{\mc C}} |\mu(u)|^2.
\eeqn
Hence when the area of ${\mc C}$ is sufficiently large, there is a nonempty subset of $\Sigma_{\mc C}$ whose images under $u$ is in the region of $V$ where the $K$-action is free. Hence for the same reason as the above case the automorphism group is finite.

Lastly we consider the case when ${\mc C}$ is unstable, namely ${\bm v}$ is a soliton. In this case ${\mc C} = {\mb R} \times S^1$. Let $w = s + {\bf i} t$ be the cylindrical coordinate. We can view the $r$-spin structure $(L_{\mc C}, \phi_{\mc C})$ as a trivial one. Namely, the log-canonical bundle $\omega_{\mc C}$ is trivialized by $dw$, $L_{\mc C}$ is the trivial bundle and $\phi_{\mc C}: L_{\mc C}^{\otimes 5} \to \omega_{\mc C}$ is the identity of the trivial bundle. Then for an automorphism $(\rho, g)$ of ${\bm v}$, the curve isomorphism $\rho$ is a translation of the cylinder, namely, $(s, t) \mapsto (s + S, t + T)$ for certain constants $S, T$. Then $\rho^* L_{\mc C}$ and $L_{\mc C}$ are both canonically identified with the trivial bundle. Therefore the bundle isomorphism contained in $\rho$, 
\beqn
L_{\mc C} \to \rho^* L_{\mc C} \underset{canonical}{\cong} L_{\mc C}
\eeqn
is essentially an element of ${\mb Z}_r$, and there is a group homomorphism ${\rm Aut} {\bm v} \to {\mb Z}_r$. Hence to prove that ${\rm Aut} {\bm v}$ is finite, it suffices to consider the subgroup of automorphism whose images in ${\mb Z}_r$ are trivial. Furthermore, by the finite energy condition, it is easy to see that $S$ must be zero, namely $\rho$ is a rotation of the infinite cylinder. 

\begin{claim}
$T$ is a rational multiple of $2 \pi$.
\end{claim}

\begin{proof}[Proof of the claim]
We can transform ${\bm v}$ into temporal gauge, namely  
\beqn
A = d + \psi dt. 
\eeqn
Hence we can abbreviate ${\bm v}$ by a pair $(u, \psi)$. Since $g$ preserves $A$ upto a rotation, $g$ must only depend on $t$. Furthermore, we can also fix the gauge at $-\infty$ so that 
\beqn
\lim_{s \to -\infty} ( \psi(s, t), u(s, t)) = (\eta_-, x_-(t))
\eeqn
where $\eta_- \in {\mf k}$ and ${\bm x}_-(t) = (x_-(t), \eta_-)$ is a critical loop, i.e.
\beqn
x_-'(t) + {\mc X}_{\eta_-}(x_-(t)) = 0. 
\eeqn
Then
\begin{align*}
&\ e^{2\pi \eta_-} x_-(0) = x_-(2\pi) = x_-(0)=: x_-,\ &\ x_-(t) = e^{- \eta_- t} x_-. 
\end{align*}
Since $x_- \in \mu^{-1}(0)$ which has finite stabilizer, $e^{2\pi \eta_-}$ is of finite order in $K$, say order $k$. Then we can lift ${\bm v}$ to $k$-fold covering of the infinite cylinder, denoted by $(\tilde u, \tilde \psi)$, which is defined by
\begin{align*}
&\ \tilde u(\tilde s, \tilde t) = u(\tilde s, \tilde t),\ &\ \tilde \psi (\tilde s, \tilde t) = \psi( \tilde s, \tilde t).
\end{align*}
The automorphism is also lifted to the $k$-fold cover by $\tilde \rho (\tilde s, \tilde t) = (\tilde s, \tilde t + T)$, $\tilde g(\tilde s, \tilde t) = g(\tilde s, \tilde t) = g(\tilde t)$. Then using the gauge transformation 
\beqn
\tilde h ( \tilde s, \tilde t) = e^{- \eta_- \tilde t}
\eeqn
the pair $\tilde {\bm v}$ is transformed to $\tilde {\bm v}' = (\tilde u', \tilde \psi')$ such that
\beqn
\lim_{s \to -\infty} \tilde u'(\tilde s, \tilde t) = x_-,\ \lim_{ s \to -\infty} \tilde \psi'(\tilde s, \tilde t) = 0.
\eeqn
The automorphism is transformed to the $T$-rotation and 
\beqn
\tilde g' = \tilde h^{-1} \tilde g \tilde h.
\eeqn
Therefore, $\tilde g'$ fixes the constant loop $x_-$. So $\tilde g'(\tilde t)$ is a constant element $\gamma_-$ of the finite stabilizer of $x_-$. Then it follows that 
\beq\label{eqn91}
g(t) = e^{- \eta_- t} \gamma_- e^{\eta_- t}. 
\eeq
Since the $T$-rotation and $g(t)$ altogether fixes the loop $x_-(t) = e^{- \eta_- t} x_-$, we see
\beqn
e^{- \eta_- t} \gamma_-^{-1} e^{\eta_- t} \cdot e^{- \eta_- (t + T) } x_- = e^{- \eta_- t} x_- \Longrightarrow \gamma_-^{-1} e^{- \eta_- T} x_- = x_-. 
\eeqn
Therefore $e^{\eta_- T}$ is a stabilizer of $x_-$. Since there is no continuous stabilizer of $x_-$, either $T$ is a rational multiple of $2\pi$ or $\eta_- = 0$. 

We need to prove that in the case $\eta_- = 0$, $T$ still has to be a rational multiple of $2\pi$. Indeed, if $\eta_- = 0$, then \eqref{eqn91} implies that $g(t)\equiv \gamma_-$. Since $\gamma_-^k = 1$, we see that the $kT$-rotation is an automorphism of $({\mc C}, {\bm v})$. If $T$ is an irrational multiple of $2\pi$, it follows that ${\bm v}$ is independent of $t$. Then the gauged Witten equation is reduced to the ODE for pairs $(u, \psi): {\mb R} \to X \times {\mf k}$:
\begin{align*}
&\ u'(s) + J {\mc X}_{\psi(s)}(u(s)) = 0,\ &\ \psi'(s)  + \mu (u(s)) = 0.
\end{align*}
Indeed this is the negative gradient flow equation of the Morse--Bott function 
\beqn
(x, \eta) \mapsto \langle \mu(x), \eta \rangle.
\eeqn
whose critical submanifold is $\mu^{-1}(0) \times \{0\}$ and which has only one critical value $0$. It follows that ${\bm v}$ must be trivial, which contradicts the stability condition. Therefore $T$ is a rational multiple of $2\pi$ whether $\eta_- = 0$ or not. \end{proof}

Therefore, we can only consider automorphisms of $({\mc C}, {\bm v})$ whose underlying automorphism of the cylinder is the identity. They are given by gauge transformations of the form \eqref{eqn91}, hence only depends on a stabilizer $\gamma_-$ of $x_-$. Therefore there are only finitely many such automorphisms. \end{proof}

\subsection{Thickening data}

\subsubsection{Thickening data}

Recall the notion of generalized $r$-spin curves in Subsection \ref{subsection23} and resolution data in Subsection \ref{subsection22}.

\begin{defn}\label{defn93} (cf. \cite[Definition 9.2.1]{Pardon_virtual})
A {\it thickening datum} $\alpha$ is a tuple 
\beqn
({\mc C}_\alpha, {\bm v}_\alpha, {\bf y}_\alpha^*, {\bm r}_\alpha, {\bm E}_\alpha, \iota_\alpha, H_\alpha )
\eeqn
where
\begin{enumerate}

\item ${\mc C}_\alpha$ is a smooth or nodal $r$-spin curve with $n$ punctures.

\item ${\bm v}_\alpha$ is a stable smooth solution over ${\mc C}_\alpha$.

\item ${\bf y}_\alpha^*$ is an unordered list of $l_\alpha$ points on ${\mc C}_\alpha$ that makes $({\mc C}_\alpha, {\bf y}_\alpha^*)$ a stable generalized $r$-spin curve of type $(n, 0, l_\alpha)$. Moreover, we require the following condition. Let $\Gammait_\alpha$ be the automorphism group of $({\mc C}_\alpha, {\bm v}_\alpha)$, then the image of the natural map $\Gammait_\alpha \to {\rm Aut}({\mc C}_\alpha)$ lies in the finite subgroup ${\rm Aut}({\mc C}_\alpha, {\bf y}_\alpha^*)$. 

\item ${\bm r}_\alpha$ is a resolution data for $({\mc C}_\alpha, {\bf y}_\alpha^* )$ (see Subsection \ref{subsection22}), which induces an explicit universal unfolding of $({\mc C}_\alpha, {\bf y}_\alpha^* )$, denoted by 
\beqn
\pi_\alpha: {\mc U}_\alpha \to {\mc V}_\alpha = {\mc V}_{\alpha, {\rm def}} \times {\mc V}_{\alpha, {\rm res}}.
\eeqn
The resolution data induces a smooth (orbifold) fibre bundle ${\mc Y}_\alpha \to {\mc V}_\alpha$. Notice that the automorphism $\Gammait_\alpha$ acts on ${\mc V}_\alpha$ via the map $\Gammait_\alpha \to {\rm Aut}({\mc C}_\alpha, {\bf y}_\alpha^*)$, and also acts on the total space ${\mc Y}_\alpha$. 

\item ${\bm E}_\alpha$ is a finite dimensional representation of $\Gammait_\alpha$. 

\item $\iota_\alpha$ is a $\Gammait_\alpha$-equivariant linear map
\beqn
\iota_\alpha: {\bm E}_\alpha \to C^\infty \big( {\mc Y}_\alpha, \Omega^{0,1}_{{\mc Y}_\alpha/ {\mc V}_\alpha} \otimes_{\mb C} T^{\rm vert} {\mc Y}_\alpha \big).
\eeqn
We also require that for all $e_\alpha \in {\bm E}_\alpha$, $\iota_\alpha(e_\alpha)$ is supported away from the nodal neighborhoods.\footnote{Notice that we do not require the obstruction space ${\bm E}_\alpha$ to be gauge invariant.} Moreover, the restriction of the image $\iota_\alpha({\bm E}_\alpha)$ to the central fibre of ${\mc U}_\alpha \to {\mc V}_\alpha$ is transverse to the image of $D_\alpha: T_{{\bm v}_\alpha}^\circ {\mc B}_\alpha \to {\mc E}_\alpha$.

\item $H_\alpha\subset V$ is a compact smooth embedded codimension two submanifold with boundary (not necessarily $K$-invariant) that satisfies the following condition. The resolution data ${\bm r}_\alpha$ induces a thin-thick decomposition of ${\mc C}_\alpha$ and a trivialization of the $K$-bundle over the thin part ${\mc C}_\alpha^{\rm thin}$, hence the matter field $u_\alpha$ restricted to the thin part is identified with an ordinary map $u_\alpha^{\rm thin}: {\mc C}_\alpha^{\rm thin} \to V$. We require the following conditions.
\begin{enumerate}

\item $u_\alpha^{\rm thin}$ intersects with $H_\alpha$ transversely in the interior exactly at ${\bf y}_\alpha^*$.

\item There is no intersection on either $\partial {\mc C}_\alpha^{\rm thin}$ or $\partial H_\alpha$. 

\item $H_\alpha$ is transverse to the complex line spanned by $\partial_s u_\alpha^{\rm thin} + {\mc X}_{\phi_\alpha^{\rm thin}}$ and $\partial_t u_\alpha^{\rm thin} + {\mc X}_{\psi_\alpha^{\rm thin}}$ at each point of ${\bf y}_\alpha^*$.
\end{enumerate}
\end{enumerate}
Here ends Definition \ref{defn93}.
\end{defn}

The pregluing construction (i.e. the definition of approximate solutions) depends on choosing a gauge. The last condition on gauge implies the following convenient fact: the family of approximate solutions have inherited symmetry. 

\subsubsection{Thickened moduli space within the same stratum}

Fix the thickening datum (see Definition \ref{defn93})
\beqn
\alpha = ( {\mc C}_\alpha, {\bm v}_\alpha,  {\bf y}_\alpha^*, {\bm r}_\alpha, {\bm E}_\alpha, \iota_\alpha, H_\alpha).
\eeqn
We would like to construct a family of gauged maps over the fibres of ${\mc U}_\alpha \to {\mc V}_\alpha$ which have the same combinatorial type as ${\mc C}_\alpha$. First we need to setup a Banach manifold, a Banach space bundle and a Fredholm section. Recall that in the resolution data, ${\mc V}_a = {\mc V}_{\alpha, {\rm def}} \times {\mc V}_{\alpha, {\rm res}}$, and the resolution data ${\bm r}_\alpha$ contains a smooth trivialization 
\beqn
{\mc U}_\alpha|_{{\mc V}_{\alpha, {\rm def}} \times \{0\}} \cong {\mc V}_{\alpha, {\rm def}} \times {\mc C}_\alpha.
\eeqn
Let the coordinates of ${\mc V}_{\alpha, {\rm def}}$ and ${\mc V}_{\alpha, {\rm res}}$ be $\eta$ and $\zeta$ respectively. This trivialization induces a smooth trivialization of the universal log-canonical bundle (viewed as the cotangent bundle of punctured surfaces), and hence a smooth trivialization the $U(1)$-bundle $P_{{\mc C}_{\alpha, \eta}}$ over each fibre. On the other hand, ${\bm r}_\alpha$ contains trivializations of ${\mc P}_\alpha \to {\mc V}_{\alpha, {\rm def}}$. Define 
\beqn
{\mc B}_\alpha = {\mc V}_{\alpha, {\rm def}} \times {\mc B}_{{\mc C}_\alpha, P_\alpha}^{p, w}.
\eeqn
Here ${\mc B}_{{\mc C}_\alpha, P_\alpha}^{p, w}$ is the Banach manifold defined in Subsection \ref{subsection81}. Moreover, including vectors in the obstruction space, define 
\beqn
\hat{\mc B}_\alpha = {\bm E}_\alpha \times {\mc V}_{\alpha, {\rm def}} \times {\mc B}_{{\mc C}_\alpha, P_\alpha}^{p, w}.
\eeqn
Let the Banach space bundle $\hat{\mc E}_\alpha \to \hat{\mc B}_\alpha$ to be the pull-back of ${\mc E}_{{\mc C}_\alpha, P_\alpha}^{p, w} \to {\mc B}_{{\mc C}_\alpha, P_\alpha}^{p, w}$. Then, using the linear map $\iota_\alpha$, one can define operator associated to the ${\bm E}_\alpha$-perturbed gauged Witten equation
\beq\label{eqn92}
\hat{\mc F}_\alpha ( e_\alpha, \eta, {\bm v}) = {\mc F}_{\alpha, \eta }({\bm v}) + \iota_\alpha(e_\alpha) ( {\mc C}_{\alpha, \eta}, {\bm v}).
\eeq
Here ${\mc F}_{\alpha, \eta}$ includes both the gauged Witten equation over the fibre ${\mc C}_{\alpha, \eta}$ and the gauge fixing condition relative to the gauged map ${\bm v}_\alpha$. 

Let the linearization of the operator $\hat{\mc F}_\alpha$ be 
\beqn
\hat D_\alpha: {\bm E}_\alpha \oplus {\mc V}_{\alpha, {\rm def}} \oplus T_{{\bm v}_\alpha} {\mc B}_\alpha \to {\mc E}_\alpha.
\eeqn
When the deformation parameter is turned off, one obtains the linear map
\beq\label{eqn93}
D_\alpha: {\bm E}_\alpha \oplus T_{{\bm v}_\alpha} {\mc B}_\alpha \to {\mc E}_\alpha.
\eeq
By the transversality assumption it is surjective. Define
\beqn
{\mc V}_{\alpha, {\rm map}} = {\rm ker} D_\alpha \subset {\bm E}_\alpha \oplus T_{{\bm v}_\alpha} {\mc B}_\alpha. 
\eeqn
Denote a general element of ${\mc V}_{\alpha, {\rm map}}$ by $\xi$ and add the subscript $\alpha$ whenever necessary. Then we are going to construct a family 
\beqn
\hat{\bm v}_{\alpha, \xi, \eta} = \Big( {\bm v}_{\alpha, \xi, \eta}, e_{\alpha, \xi, \eta}\Big),\ \xi \in {\mc V}_{\alpha, {\rm map}},\ \eta \in {\mc V}_{\alpha, {\rm def}}
\eeqn
where ${\bm v}_{\alpha, \xi, \eta}$ is a gauged map over the bundle $P_{\alpha, \eta} \to {\mc C}_{\alpha, \eta}$. 

The following construction needs another choice. We choose a bounded right inverse to the operator $D_\alpha^E$ in \eqref{eqn93}
\beqn
Q_\alpha: {\mc E}_\alpha \to {\bm E}_\alpha \oplus T_{{\bm v}_\alpha} {\mc B}_\alpha.
\eeqn
It induces a right inverse to $\hat D_\alpha$, denoted by 
\beq\label{eqn94}
\hat Q_\alpha: {\mc E}_\alpha \to {\mc V}_{\alpha, {\rm def}} \oplus {\bm E}_\alpha \oplus T_{{\bm v}_\alpha} {\mc B}_\alpha.
\eeq
Using the exponential map of the Banach manifold ${\mc B}_\alpha$, we define the {\it approximate solution}
\beqn
\hat{\bm v}_{\alpha, \xi, \eta}^{\rm app} =\Big( {\bm v}_{\alpha, \xi, \eta}^{\rm app}, e_{\alpha, \xi, \eta}^{\rm app} \Big) =  \Big( \exp_{{\bm v}_\alpha} \xi, e_\xi \Big),\ \xi \in {\mc V}_{\alpha, {\rm map}},\ \eta\in {\mc V}_{\alpha, {\rm def}},
\eeqn
where $e_\xi \in {\bm E}_\alpha$ is the ${\bm E}_\alpha$-component of $\xi$.

\begin{lemma}
There exists a constant $C>0$ such that when $\| \xi \|$ and $\| \eta\|$ are sufficiently small, one has
\beqn
\left\| \hat {\mc F}_\alpha \Big( e_{\alpha, \xi, \eta}^{\rm app}, \eta, {\bm v}_{\alpha, \xi, \eta}^{\rm app} \Big) \right\| \leq C \Big( \| \xi \| + \| \eta\| \Big).
\eeqn
\end{lemma}

\begin{proof}
Left to the reader.
\end{proof}

Then the implicit function theorem implies that one can correct the approximate solutions by adding a vector in the image of the right inverse $\hat Q_\alpha$ and the correction is unique. Therefore, we obtain a family 
\beq\label{eqn95}
\hat{\bm v}_{\alpha, \xi, \eta} = \Big( {\bm v}_{\alpha, \xi, \eta}, e_{\alpha, \xi, \eta} \Big)
\eeq
which solves the ${\bm E}_\alpha$-perturbed gauged Witten equation over the $r$-spin curve ${\mc C}_{\alpha, \eta}$.

\subsubsection{Pregluing}

Now we construct approximate solutions by allowing nonzero gluing parameters. The construction relies on choosing certain cut-off functions. We fix a pair cut-off functions $\rho_\pm: {\mb R} \to [0, 1]$ satisfying
\beq\label{eqn96}
{\rm supp} \rho_- = (-\infty, 0],\ \rho|_{(-\infty, -1]} = 1,\  {\rm supp} d\rho_- = [-1,0],\ \rho_+(t) = \rho_-(-t).
\eeq
For all $T >>1$, also denote 
\begin{align*}
&\ \rho_-^T(t) = \rho_-( \frac{ t}{T}),\ &\ \rho_+^T(t) = \rho_+( \frac{ t}{T}).
\end{align*}

Consider a typical node $w\in {\rm Irre} {\mc C}_\alpha$. Let the monodromies of the $r$-spin structure at the two sides of $w$ be $\gamma_-, \gamma_+ \in {\mb Z}_r$. 

Using the resolution data and the functions $z\mapsto \pm \log z$ we can identify the two sides of the nodes as semi-infinite cylinders 
\begin{align*}
&\ C_- = [a, +\infty) \times S^1,\ &\ C_+ = (-\infty, b] \times S^1. 
\end{align*}

The bundle $L_R$ is then trivialized up to ${\mb Z}_r$ action. Moreover, the resolution data trivializes $P_K$ over $C_- \sqcup C_+$. Then we can identify the connection as forms 
\begin{align*}
&\ A|_{C_-} = d + \alpha_-,\ &\ A|_{C_+} = d + \alpha_+.
\end{align*}
The restrictions of the section $u$ onto $C_-$ and $C_+$ are identified with maps 
\begin{align*}
&\ u_-: C_- \to ({\rm Crit} W)^{\rm ss},\ &\ u_+: C_+ \to ({\rm Crit} W)^{\rm ss}.
\end{align*}
Then using the cut-off function $\rho$ and the exponential map of $({\rm Crit} W)^{\rm ss}$ we can define the approximate solution as in the case of Gromov--Witten theory. For any gluing parameter $\zeta = (\zeta_w)_w$, introduce 
\beqn
- \log \zeta_w = 4T_w + {\bf i} \theta_w.
\eeqn
For $k \in {\mb R}_+$, let $\Sigma_{kT} \subset \Sigma_{\mc C}$ be the closed subset obtained by removing the radius $e^{-kT_w}$ disk around each nodal points (but not punctures). The long cylinder is identified with 
\beqn
N_{w, kT} = [ (k-2)T, (2-k) T] \times S^1.
\eeqn
Then we have
\beqn
\Sigma_\eta = \Sigma_T \cup \bigoplus_w N_{w, T}
\eeqn
where the intersections of $\Sigma_T$ and $N_{w, T}$ are only their boundaries.

Moreover, suppose $(A_\pm, u_\pm)$ converges to the loop $(x_\pm, \eta_\pm): S^1 \to X_W\cap \mu^{-1}(0) \times {\mf k}$ as $z$ approaches $\hat w_\pm$. The gauge we choose guarantees that the two loops are identical, denoted by $(x_w, \eta_w)$. Further, using the exponential map of $({\rm Crit} W)^{\rm ss}$, we can write 
\beqn
u_\pm (s, t) = \exp_{x_w(t)} \xi_\pm(s, t),\ {\rm where}\ \xi_\pm \in W^{1, p, \delta}(C_\pm, x_w^* T{\rm Crit} W).	
\eeqn

\begin{defn}\label{defn95} {\rm (Central approximate solution)}
Given a thickening datum $\alpha$ as denoted in Definition \ref{defn93}, for any gluing parameter $\zeta \in \hat{\mc V}_{\rm res}$, we obtain an $r$-spin curve ${\mc C}_{\alpha, \eta, \zeta}$ with extra unordered marked points, which represents a point in $\ov{\mc M}{}_{g, n, l_\alpha}^r$. The central approximate solution is the object ${\bm v}_{\alpha, \eta} = (A_{\alpha,\eta}, u_{\alpha,\eta})$ where 

\begin{itemize}
\item $(A_{\alpha,\eta}, u_{\alpha,\eta} )|_{\Sigma_T} = (A_\alpha, u_\alpha)|_{\Sigma_T}$.

\item For each node $w$ with $\eta_w \neq 0$, we have 
\begin{align}\label{eqn97}
&\ A_{\alpha,\eta}|_{N_{w, T}} = d + a_\eta = d + \rho_-^T a_- + \rho_+^T a_+,\ &\ u_{\alpha, \eta} |_{N_{w, T}} = \exp_{x_w} \left( \rho_-^T \xi_- + \rho_+^T \xi_+ \right). 
\end{align}
\end{itemize}
\end{defn}

We also consider an auxiliary object ${\bm v}_\eta'$ that is defined over the unresolved curve ${\mc C}$. Indeed, we can include $\Sigma_{2T} \hookrightarrow \Sigma_{\mc C}$ and define ${\bm v}_\eta'|_{\Sigma_{2T}} = {\bm v}_\eta|_{\Sigma_{2T}}$. Further, \eqref{eqn96} and \eqref{eqn97} implies that ${\bm v}_\eta$ is equal to the loop $(x_w, \eta_w)$ near the boundary of $\Sigma_{2T}$. Then over $\Sigma_{\mc C} \setminus \Sigma_{2T}$, ${\bm v}_\eta'$ is defined to be the loop $(x_w, \eta_w)$.

\subsubsection{Thickened moduli spaces}

We postponed the definition of $I$-thickened moduli spaces here because its definition relies on the description of approximate solutions.

\begin{defn}\label{defn96} {\rm (Thickened solutions)} (cf. \cite[Definition 9.2.3]{Pardon_virtual})
For a finite set $I$ of thickening data, an {\it $I$-thickened solution} is a quadruple
\beqn
({\mc C}, {\bm v}, ({\bf y}_\alpha)_{\alpha \in I}, (\wh\phi_\alpha)_{\alpha \in I}, (e_\alpha)_{\alpha \in I})
\eeqn
where
\begin{enumerate}
\item ${\mc C} =  (\Sigma_{\mc C}, \vec{\bf z}_{\mc C}, L_{\mc C}, \varphi_{\mc C})$ is a smooth or nodal $r$-spin curve of type $(g, n)$ and for each $\alpha \in I$, ${\bf y}_\alpha$ is an unordered list of marked points ${\bf y}_\alpha$ labelled by $\alpha \in I$. We require that for every $\alpha \in I$, $({\mc C}, {\bf y}_\alpha)$ is a stable generalized $r$-spin curve. 

\item ${\bm v} = (P, A, u)$ is a smooth gauged map over ${\mc C}$. 

\item For each $\alpha \in I$, $\wh\phi_\alpha: P \to {\mc P}_\alpha$ is an inclusion of $K$-bundles which covers an isomorphism $\phi_\alpha: {\mc C} \cong {\mc C}_{\phi_\alpha} = {\mc C}_{\alpha, \eta, \zeta}$ for some fibre ${\mc C}_{\phi_\alpha} \subset {\mc U}_\alpha$ such that $\phi_\alpha({\bf y}_\alpha) = {\bf y}_{\alpha, \eta, \zeta} = {\bf y}_{\phi_\alpha}$. 
Remember there is an approximate solution ${\bm v}_{\alpha, \eta, \zeta}^{\rm app}$ on ${\mc C}_{\alpha, \eta, \zeta}$. There is also a thick-thin decomposition
\beqn
{\mc C}_{\phi_\alpha} = {\mc C}_{\phi_\alpha}^{\rm thick} \cup {\mc C}_{\phi_\alpha}^{\rm thin}.
\eeqn

\item ${\bm v}$ and $\wh\phi_\alpha$ need to satisfy the following condition. For each $\alpha$ and each irreducible component $v$, $(\wh\phi_\alpha^{-1})^* {\bm v}_v$ is in Coulomb gauge with respect to ${\bm v}_{\alpha, \eta, \zeta}^{\rm app}|_{{\mc C}_{\phi_\alpha, v}}$. Here ${\mc C}_{\phi_\alpha, v} \subset {\mc C}_{\phi_\alpha}$ is the component identified with ${\mc C}_v$. Moreover, over the intersection ${\mc C}_{\phi_\alpha, v} \cap {\mc C}_{\phi_\alpha}^{\rm thin}$ where the $K$-bundle is trivialized, $(\wh\phi_\alpha^{-1})^* u_v$ is identified with a smooth map $u_v^{\rm thin}: {\mc C}_{\phi_\alpha, v}^{\rm thin} \to X$. We require that ${\bf y}_\alpha \cap \Sigma_v = u_v^{\rm thin} \pitchfork H_\alpha$ and the intersections are transverse and are away from the boundaries.

\item For every $\alpha \in I$, $e_\alpha \in {\bm E}_\alpha$.

\item The following $I$-thickened gauged Witten equation is satisfied:
\beqn
{\mc F}_{\mc C} ({\bm v}) + \sum_{\alpha \in I} \iota_\alpha( e_\alpha) ( {\mc C}, {\bm v}, {\bf y}_\alpha, \wh\phi_\alpha) = 0. 
\eeqn
\end{enumerate}

Two $I$-thickened solutions $({\mc C}, {\bm v}, ({\bf y}_\alpha), (\wh\phi_\alpha), (e_\alpha))$ and $({\mc C}', {\bm v}', ({\bf y}_\alpha'), (\wh\phi_\alpha'), (e_\alpha'))$ are {\it isomorphic} if for every $\alpha \in I$, $e_\alpha = e_\alpha'$, and if there is a commutative diagram
\beqn
 \xymatrix{  P \ar[r]^{\wh\phi_\alpha} \ar[d] \ar@/^2pc/[rr]^{\wh\rho} & {\mc P}_\alpha \ar[d] &  P' \ar[l]_{\wh\phi_\alpha'} \ar[d] \\
                      {\mc C} \ar[r]^{\phi_\alpha} \ar@/_2pc/[rr]^{\rho} & {\mc U}_\alpha        &  {\mc C}' \ar[l]_{\phi_\alpha'} }
\eeqn
where $\rho$ is an isomorphism of $r$-spin curves and $\wh\rho$ is an isomorphism of bundles such that 
\beqn
{\bm v} = \wh\rho^* {\bm v}'.
\eeqn
Here ends Definition \ref{defn96}.
\end{defn}

The $I$-thickened solutions give local charts of the moduli spaces (in a formal sense). Indeed, it is straightforward to define a topology on the space of isomorphism classes of $I$-thickened solutions, denoted by ${\mc M}_I$. The finite group 
\beqn
\Gammait_I = \prod_{\alpha \in I} \Gammait_\alpha
\eeqn
acts on ${\mc M}_I$ continuously as follows. For any $\gamma = (\gamma_\alpha) \in \Gammait_I$, define 
\beqn
\gamma \cdot \Big( {\mc C}, {\bm v}, ({\bf y}_\alpha), (\wh\phi_\alpha), (e_\alpha) \Big) = \Big( {\mc C}, {\bm v}, ({\bf y}_\alpha), ( \gamma_\alpha^{{\mc P}_\alpha} \circ \wh\phi_\alpha), (\gamma_\alpha \cdot e_\alpha) \Big).
\eeqn
Since all $\iota_\alpha$ are $\Gammait_\alpha$-equivariant, this is indeed an action on ${\mc M}_I$. Furthermore, there is hence a continuous, $\Gammait_\alpha$-equivariant map 
\beqn
\tilde S_I: {\mc M}_I \to {\bm E}_I:= \bigoplus_\alpha {\bm E}_\alpha
\eeqn
sending an $I$-thickened solution to $\{ e_\alpha\} \in {\bm E}_I$. Denote 
\beqn
U_I = {\mc M}_I/ \Gammait_I,\ E_I  = ({\mc M}_I \times {\bm E}_I)/ \Gammait_I.
\eeqn
and the induced section $S_I: U_I \to E_I$. There is the canonical map 
\beqn
\psi_I: S_I^{-1}(0) \to {\mc M}
\eeqn
by only remembering the data ${\mc C}$ and ${\bm v}$, and sending $({\mc C}, {\bm v})$ to its isomorphism class. The image of $\psi_I$ is denoted by $F_I$ and denote 
\beq\label{eqn98}
C_I = (U_I, E_I, S_I, \psi_I, F_I).
\eeq
By the definition of local charts (Definition \ref{defn610}), $C_I$ is a local chart if we can verify the following facts:
\begin{itemize}
\item ${\mc M}_I$ is a topological manifold.

\item $F_I$ is an open subset of ${\mc M}$.

\item $\psi_I$ is a homeomorphism.
\end{itemize}
They will be treated in the next section for general $I$. In the rest of this section, we will prove the corresponding results for the case that $I$ contains a single thickening datum $\alpha$.

\subsubsection{Constructing thickening data} 

Now we prove the following lemma. 

\begin{lemma}\label{lemma97}
Let $\sf\Gamma$ be a stable decorated dual graph. For any $p \in \ov{\mc M}_{\sf \Gamma}$, there exists a transverse thickening datum $\alpha$ centered at $p$. 
\end{lemma}

First, choose a representative ${\bm v}_\alpha$ as a stable solution to the gauged Witten equation over a smooth or nodal $r$-spin curve ${\mc C}_\alpha$. Let $\Gammait_\alpha$ be the automorphism group of $({\mc C}_\alpha, {\bm v}_\alpha)$. Then $\Gammait_\alpha$ acts on the set of irreducible components ${\rm Irre} {\mc C}$ and the subset ${\rm Unst}{\mc C} \subset {\rm Irre} {\mc C}$ of unstable components is $\Gammait_\alpha$-invariant. For each $\Gammait_\alpha$-orbit $O$ of unstable components, for each representative $v \in {\rm Unst} {\mc C}$ of this orbit, by the stability condition, the soliton ${\bm v}_v = (A_v, u_v)$ is nonconstant. Notice that we regard the $K$-bundle over this rational component is trivialized and $u_v$ is regarded as a genuine map into $V$. Hence the subset of $\Sigma_v$ of points where the covariant derivative $D_{A_v} u_v$ is nonzero is an open and dense subset. Hence we can choose a codimension two submanifold $H_O$ with boundary which intersects transversely with $u_v$ at some point of $\Sigma_v$. Choose $H_O$ for every $\Gammait_\alpha$-orbit $O$ of unstable components, and define
\beqn
H_\alpha = \bigcup_{O \in {\rm Unst} {\mc C}/ \Gammait_\alpha} H_O. \footnote{Different $H_O$ might intersect. So we actually need maps from the disjoint union of $H_O$ into $X$.}
\eeqn
Then let ${\bf y}_\alpha$ be the (unordered) set of intersection points between $H_\alpha$ and the images of all unstable components. 

Now the smooth or nodal $r$-spin curve ${\mc C}_\alpha = (\Sigma_{{\mc C}_\alpha}, \vec{\bf z}_{{\mc C}_\alpha})$ together with the collection ${\bf y}_\alpha$ (which is unordered) form a stable generalized $r$-spin curve. Then one can choose a resolution data ${\bm r}_\alpha$ of $({\mc C}_\alpha, {\bf y}_\alpha)$. It contains the following objects
\begin{enumerate}
\item A universal unfolding $\pi_\alpha: {\mc U}_\alpha \to {\mc V}_\alpha$, where ${\mc V}_\alpha = {\mc V}_{\alpha, {\rm def}} \times {\mc V}_{\alpha, {\rm res}}$.

\item A smooth principal $K$-bundle ${\mc P}_{K,\alpha} \to {\mc U}_\alpha$ whose restriction to the central fibre is identified with $P_\alpha \to {\mc C}_\alpha$. 

\item The universal cylindrical metric induces a fibrewise cylindrical metric, and hence a $U(1)$-bundle ${\mc P}_{R, \alpha} \to {\mc U}_\alpha$. Together with the bundle ${\mc P}_{K, \alpha}$ one obtains a principal $\ubar K$-bundle $\ubar {\mc P}_\alpha \to {\mc U}_\alpha$. 

\item A thick-thin decomposition ${\mc U}_\alpha = {\mc U}_\alpha^{\rm thick} \cup {\mc U}_\alpha^{\rm thin}$.

\item A trivialization of the thick part of the universal unfolding
\beqn
t_C^{\rm thick}: {\mc U}_\alpha^{\rm thick} \cong {\mc V}_\alpha \times {\mc C}_\alpha^{\rm thick}.
\eeqn

\item Trivializations of the $K$-bundles
\begin{align*}
&\ t_P^{\rm thin}: {\mc P}_{K,\alpha} |_{{\mc U}_\alpha^{\rm thin}} \cong K \times {\mc U}_\alpha^{\rm thin},\ &\ t_P^{\rm thick}: {\mc P}_{K, \alpha} |_{{\mc U}_\alpha^{\rm thick}} \cong {\mc V}_\alpha \times P_{\alpha} |_{{\mc C}_\alpha^{\rm thick}}.
\end{align*}
Moreover, $t_P^{\rm thin}$ extends the trivialization of $P_{\alpha}$ over ${\mc C}_\alpha^{\rm thin}$. 

\item $\ubar {\mc P}_\alpha$ induces an associated fibre bundle ${\mc Y}_\alpha:= \ubar {\mc P}_\alpha(V) \to {\mc U}_\alpha$ whose fibres are $V$. The above trivializations induce trivializations of $\ubar {\mc P}_\alpha(V)$.

\begin{align}\label{eqn99}
&\ t_Y^{\rm thin}: {\mc Y}_\alpha |_{{\mc U}_\alpha^{\rm thin}} \cong V \times {\mc U}_\alpha^{\rm thin},\ &\ t_Y^{\rm thick}: {\mc Y}_\alpha |_{{\mc U}_\alpha^{\rm thick}} \cong {\mc V}_\alpha \times Y_\alpha|_{{\mc C}_\alpha^{\rm thick}}.
\end{align}
\end{enumerate}

Now we construct the obstruction space ${\bm E}_\alpha$ and the inclusion $\iota_\alpha$. For each irreducible component $\Sigma_v \subset \Sigma_{{\mc C}_\alpha}$, the augmented linearization of the gauged Witten equation over $\Sigma_v$ at ${\bm v}_v$ is a Fredholm operator
\beqn
D_v: T_{{\bm v}_v} {\mc B}_v \to {\mc E}_v.
\eeqn
Here $T_{{\bm v}_v} {\mc B}_v$ consists of infinitesimal deformations of the gauged map ${\bm v}_v$. Inside there is a finite-codimensional subspace 
\beqn
T_{{\bm v}_v}^\circ {\mc B}_v \to T_{{\bm v}_v} {\mc B}_v
\eeqn 
consisting of infinitesimal deformations whose values at all punctures and nodes on this component vanish. Then as in the case of Gromov--Witten theory, one can find a finite dimensional space of sections ${\bm E}_\alpha \subset L^{p, \delta}( \Sigma_v, \Lambda^{0,1} \otimes u_v^* \ubar P_\alpha (TV))$ satisfying the following: 
\begin{enumerate}

\item Elements of ${\bm E}_v$ are smooth sections and are supported in a compact subset $O_v \subset \Sigma_v$ that is disjoint from the special points.

\item There is a $K$-invariant open neighborhood $O$ of $\mu^{-1}(0)$ in $({\rm Crit} W)^{\rm ss}$ such that $u_v(O_v)$ is contained in $P_v \times_K O$. 

\item ${\bm E}_v$ is transverse to the image $D_v( T_{{\bm v}_v}^\circ {\mc B}_v)$. 
\end{enumerate}

We may regard each $e_v \in {\bm E}_v$ as a smooth section of $\pi^* \Lambda^{0,1} \otimes \ubar P_\alpha (TV)$ restricted to the graph of $u_v$, where temporarily denote by $\pi$ the projection $\ubar P_\alpha (V) \to \Sigma_v$. Then we extend them to global smooth sections of $\pi^* \Lambda^{0,1} \otimes \ubar P_\alpha (TV)$ over $\ubar P_\alpha (V)$, or equivalently, construct a linear map 
\beqn
\iota_v: {\bm E}_v \to C^\infty( \ubar P_\alpha(V), \pi^* \Lambda^{0,1} \otimes \ubar P_\alpha (TV) ).
\eeqn
Then define 
\beqn
{\bm E}_\alpha = \bigoplus_{v \in {\rm Irre} {\mc C}_\alpha} {\bm E}_v.
\eeqn
We have a linear map 
\beq\label{eqn910}
\iota_\alpha: {\bm E}_\alpha \to C^\infty( \ubar P_\alpha(V), \pi^* \Lambda^{0,1} \otimes \ubar P_\alpha (TV) ).
\eeq

We have not imposed the $\Gammait_\alpha$ equivariance condition. Recall that $\Gammait_\alpha$ acts on $\ubar P_\alpha(V)$ and the bundle $\pi^* \Lambda^{0,1} \otimes \ubar P_\alpha(TV)$. Hence by enlarging ${\bm E}_\alpha$ so that it becomes $\Gammait_\alpha$-invariant while remaining finite-dimensional. Hence the above inclusion is $\Gammait_\alpha$-equivariant. 

So far we have chosen obstruction spaces over the central fibre. Lastly we need to construct a $\Gammait_\alpha$-equivariant linear map 
\beqn
\iota_\alpha: {\bm E}_\alpha \to C^\infty ( {\mc Y}_\alpha, \Omega^{0,1}_{{\mc Y}_\alpha / {\mc V}_\alpha} \otimes \ubar{\mc P}_\alpha(TV) )
\eeqn
which extends \eqref{eqn910}. Notice that we have a $\Gammait_\alpha$-equivariant decomposition 
\beqn
{\bm E}_\alpha = {\bm E}_\alpha^{\rm thick} \oplus {\bm E}_\alpha^{\rm thin}
\eeqn
and such that $\iota_\alpha = \iota_\alpha^{\rm thick} \oplus \iota_\alpha^{\rm thin}$ where 
\beq\label{eqn911}
\iota_\alpha^{\rm thick}: {\bm E}_\alpha^{\rm thick} \to C^\infty( \ubar P_\alpha(V)^{\rm thick}, \pi^* \Lambda^{0,1} \oplus \ubar P_\alpha(TV)),
\eeq
\beq\label{eqn912}
\iota_\alpha^{\rm thin}: {\bm E}_\alpha^{\rm thin} \to C^\infty( \ubar P_\alpha(V)^{\rm thin}, \pi^* \Lambda^{0,1} \otimes \ubar P_\alpha(TV)).
\eeq
Here $\ubar P_\alpha(V) = \ubar P_\alpha(V)^{\rm thick}\cup \ubar P_\alpha(V)^{\rm thin}$ is the thick-thin decomposition coming from the resolution data. For the part ${\bm E}_\alpha^{\rm thick}$ which comes from obstructions over stable components, using the trivialization $t_Y^{\rm thick}$ in \eqref{eqn99}, the map \eqref{eqn911} can be extended to 
\beqn
\iota_\alpha^{\rm thick}: {\bm E}_\alpha^{\rm thick} \to C^\infty( {\mc Y}_\alpha^{\rm thick}, \Omega^{0,1}_{{\mc Y}_\alpha/ {\mc V}_\alpha} \otimes \ubar {\mc P}_\alpha(TV)).
\eeqn
On the other hand, for the part ${\bm E}_\alpha^{\rm thin}$ which comes from the unstable components, since over the thin part the resolution data provides canonical cylindrical coordinates as well as a trivialization of ${\mc Y}_\alpha|_{{\mc U}_\alpha^{\rm thin}}$, one has the identification 
\beqn
C^\infty( {\mc Y}_\alpha^{\rm thin}, \Omega_{{\mc Y}_\alpha/{\mc V}_\alpha}^{0,1} \otimes \ubar {\mc P}_\alpha(TV)  ) \cong C^\infty( {\mc U}_\alpha^{\rm thin} \times V, TV )
\eeqn
Then using the trivialization $t_Y^{\rm thin}$ in \eqref{eqn99}, one can extend \eqref{eqn912} to a linear map 
\beqn
\iota_\alpha^{\rm thin}: {\bm E}_\alpha^{\rm thin} \to C^\infty( {\mc Y}_\alpha^{\rm thin}, \Omega^{0,1}_{{\mc Y}_\alpha/ {\mc V}_\alpha} \otimes \ubar {\mc P}_\alpha(TV) ).
\eeqn
Both $\iota_\alpha^{\rm thin}$ and $\iota_\alpha^{\rm thick}$ remain $\Gammait_\alpha$-equivariant. Hence their direct sum $\iota_\alpha:= \iota_\alpha^{\rm thick} \oplus \iota_\alpha^{\rm thin}$ provides the last piece of the thickening datum $\alpha$. 

This finishes the proof of Lemma \ref{lemma97}.

\subsection{Nearby solutions}

Let $\alpha$ be a thickening datum, which contains an $r$-spin curve ${\mc C}_\alpha$ and stable solution ${\bm v}_\alpha$ to the gauged Witten equation. Then for any gluing parameter $\zeta$, we have constructed an approximate solution ${\bm v}_{\alpha, \zeta}^{\rm app}$ over every curve of the form ${\mc C}_{\alpha, \eta, \zeta}$. Notice that the role of the deformation parameter $\eta$ is only to vary the complex structure (including the position of the markings) on the same curve. 

We would like to have a quantitative way to measure the distance between an $\alpha$-thickened solution to the central one. When the $\alpha$-thickened solution is defined over a resolved domain, we consider the distance from a corresponding approximate solution. Hence we need to define certain weighted Sobolev norm over the approximate solution. 

\begin{defn}\label{defn98} {\rm (Weight function over the resolved domain)} Let ${\bm r}_\alpha$ be the resolution datum contained in the thickening datum $\alpha$. Let ${\mc U}_\alpha^* \subset {\mc U}_\alpha$ be the complement of the nodes and punctures. Define a function 
\beqn
\omega_\alpha: {\mc U}_\alpha^* \to {\mb R}_+
\eeqn
as follows. For each component of the thin part ${\mc U}_{\alpha, i}^{\rm thin} \subset {\mc U}_{\alpha}^{\rm thin}$, if it corresponds to a marking with fibrewise cylindrical coordinate $s + {\bf i} t$, then define 
\beqn
\omega_\alpha|_{{\mc U}_{\alpha, i}^{\rm thin}} = e^s. 
\eeqn
If the component ${\mc U}_{\alpha, i}^{\rm thin}$ corresponds to a node, then for each corresponding gluing parameter $\zeta_i$ with the long cylinder $w_i^+ w_i^- = \zeta_i$ with $|w_i^\pm | \leq r_i$, then define 
\beqn
\omega_\alpha|_{{\mc U}_{\alpha, i}^{\rm thin}} = \left\{\begin{array}{cc} |w_i^+|^{-1}, &\   \sqrt{ |\zeta_i|} \leq |w_i^+| \leq r_i,\\
                                                                      |w_i^-|^{-1}, &\  \sqrt{|\zeta_i|} \leq |w_i^-| \leq r_i.    \end{array}  \right.
\eeqn
Then extend $\omega_\alpha|_{{\mc U}_\alpha^{\rm thin}}$ to a smooth function over the thick part that has positive values with minimal value at least 1. For any deformation parameter $\eta \in {\mc V}_{\alpha, {\rm def}}$ and gluing parameter $\zeta \in {\mc V}_{\alpha, {\rm res}}$ with $(\eta, \zeta) \in {\mc V}_\alpha$, define 
\beqn
\omega_{\alpha, \eta, \zeta}:= \omega_\alpha|_{{\mc C}_{\alpha, \eta, \zeta}}.
\eeqn
\end{defn}

\begin{defn}\label{defn99}{\rm (Weighted Sobolev norms over the resolved domain)} 
Let $\alpha$ be a thickening data and $w_\alpha$ be the weight function defined in Definition \ref{defn98}. Then for any $(\eta, \zeta) \in {\mc V}_\alpha$, define the weighted Sobolev norm 
\beqn
\| f \|_{W^{k,p,w}_{\alpha, \eta, \zeta}} = \sum_{l=0}^k  \Big[ \int_{\Sigma_{{\mc C}_{\alpha, \eta, \zeta}}} | \nabla^l f|^p \omega_{\alpha, \eta, \zeta}^{pw} \nu_{\alpha, \eta, \zeta}. \Big]^{\frac{1}{p}}
\eeqn
Here $\nu_{\alpha, \eta, \zeta} \in \Omega^2(\Sigma_{{\mc C}_{\alpha, \eta, \zeta}})$ is the area form of the family of cylindrical metrics specified in \cite[Section 2.3]{Tian_Xu_geometric}, $\nabla$ is the covariant derivative associated to the same cylindrical metric, and $|\cdot|$ is the norm on tensors associated to the same cylindrical metric. 
\end{defn}

Now for any $(\eta, \zeta)$, the object $({\mc C}_{\alpha, \eta, \zeta}, {\bf y}_{\alpha, \eta, \zeta})$ is a stable generalized $r$-spin curve. There is also the principal $K$-bundle $P_{\alpha, \eta, \zeta} \to {\mc C}_{\alpha, \eta, \zeta}$ which is contained in the resolution data. Forgetting ${\bf y}_{\alpha, \eta, \zeta}$, one can define the Banach manifold 
\beqn
{\mc B}_{\alpha, \eta, \zeta}^{p,w} = {\mc B}_{{\mc C}_{\alpha, \eta, \zeta}, P_{\alpha, \eta, \zeta}}^{p, w}
\eeqn
of gauged maps from ${\mc C}_{\alpha, \eta, \zeta}$ to $X$ which satisfy the prescribed asymptotic constrain. There is also the Banach space bundle 
\beqn
{\mc E}_{\alpha, \eta, \zeta}^{p,w} = {\mc E}_{{\mc C}_{\alpha, \eta, \zeta}, P_{\alpha, \eta, \zeta}}^{p,w} \to {\mc B}_{\alpha, \eta, \zeta}^{p,w}.
\eeqn
On the fibres of ${\mc E}_{\alpha, \eta, \zeta}^{p,w}$ and the tangent spaces of ${\mc B}_{\alpha, \eta, \zeta}^{p,w}$, we define the weighted Sobolev norms associated to the weight function $w_{\alpha, \eta, \zeta}$. Notice that for any gauged map ${\bm v} = (A, u) \in {\mc B}_{\alpha, \eta, \zeta}^{p,w}$, we use the covariant derivative associated to $A$ in the definition of higher Sobolev norms.

\begin{defn}\label{defn910} {\rm ($\epsilon$-closedness)} Let $({\mc C}, {\bf y})$ be a stable generalized $r$-spin curve and ${\bm v} = (P, A, u)$ be a smooth gauged map over ${\mc C}$. Let $\epsilon>0$ be a constant. We say that $({\mc C}, {\bm v}, {\bf y})$ is {\it $\epsilon$-close} to $\alpha$, if there exists an inclusion of bundles
\beqn
\vcenter{\xymatrix{ P \ar[r]^-{\wh\phi_\alpha} \ar[d] & P_{\phi_\alpha} = P_{\alpha, \eta, \zeta} \ar[d]\\
                    {\mc C}             \ar[r]^-{\phi_\alpha} & {\mc C}_{\phi_\alpha} = {\mc C}_{\alpha, \eta, \zeta} }}
\eeqn
satisfying the following conditions. 
\begin{enumerate}
\item $\eta \in {\mc V}_{\alpha, {\rm def}}^\epsilon$ and $\zeta \in {\mc V}_{\alpha, {\rm res}}^\epsilon$.

\item If we view $(\wh\phi_\alpha^{-1})^* {\bm v}$ as a gauged map over ${\mc C}_{\alpha, \eta, \zeta}$ which is smoothly identified with ${\mc C}_{\alpha, \zeta}$, then it is in the $\epsilon$-neighborhood of ${\bm v}_{\alpha, \zeta}^{\rm app}$ in the Banach manifold ${\mc B}_{\alpha, \zeta}^{p, w}$. 
\end{enumerate}

We say that $({\mc C}, {\bm v})$ is $\epsilon$-close to $\alpha$ if there exists a stabilizing list ${\bf y}$ such that $({\mc C}, {\bm v}, {\bf y})$ is $\epsilon$-close to $\alpha$. Here ends Definition \ref{defn910}.
\end{defn}

We need to prove that the notion of $\epsilon$-closedness defines open neighborhoods of ${\bm p}_\alpha$ in ${\mc M}_\alpha$. This is necessarily to construct the manifold charts of ${\mc M}_\alpha$ near this point.

\begin{lemma}\label{lemma911}
For any $\epsilon>0$ sufficiently small, there exists an open neighborhood ${\bm W}_\alpha^\epsilon$ of ${\bm p}_\alpha$ in ${\mc M}_\alpha$ such that for any point ${\bm p} \in W_\alpha^\epsilon$, any representative $({\mc C}, {\bm v}, {\bf y}, \phi_\alpha, g_\alpha, e_\alpha)$ of ${\bm p}$, the triple $({\mc C}, {\bm v}, {\bf y})$ is $\epsilon$-close to $\alpha$. 
\end{lemma}

\begin{proof}
One can estimate the distance between $(\wh\phi_\alpha^{-1})^* {\bm v}$ and the approximate solution ${\bm v}_{\alpha, \zeta}^{\rm app}$ by utilizing the annulus lemma. The details are left to the reader. 	
\end{proof}

Similarly, suppose $({\mc C}_\alpha, {\bm v}_\alpha)$ represents a point $p_\alpha$ in the moduli space $\ov{\mc M}{}_{g, n}^r(V, G, W, \mu)$. 

\begin{lemma}\label{lemma912}
For any $\epsilon>0$ sufficiently small, there exists an open neighborhood $W_\alpha^\epsilon$ of $p_\alpha$ in $\ov{\mc M}{}_{g, n}^r(V, G, W, \mu)$ such that for any point $p \in W_\alpha^\epsilon$, any representative $({\mc C}, {\bm v})$ of $p$ is $\epsilon$-close to $\alpha$. 
\end{lemma}

\begin{proof}
Left to the reader. 
\end{proof}

\begin{lemma}\label{lemma913}
Given a thickening datum $\alpha$, there exists $\epsilon_\alpha>0$ such that, if $({\mc C}, {\bm v})$ is $\epsilon_\alpha$-close to $\alpha$, then there exist ${\bf y}_\alpha$ and $\wh\phi_\alpha$ satisfying the following conditions.
\begin{enumerate}

\item ${\bf y}_\alpha$ stabilizes ${\mc C}$ as a generalized $r$-spin curve and $({\mc C}, {\bm v}, {\bf y}_\alpha)$ is $\epsilon_\alpha$-close to $\alpha$.

\item $\wh\phi_{\alpha}$ is an isomorphism 
\beqn
\wh\phi_{\alpha}: ({\mc C}, {\bf y}_\alpha, P) \cong ({\mc C}_{\alpha, \eta, \zeta}, {\bf y}_{\alpha, \eta, \zeta}, P_{\alpha, \eta, \zeta})
\eeqn
such that $(\wh\phi_\alpha^{-1})^* {\bm v}$ is in the $\epsilon_\alpha$-neighborhood of ${\bm v}_{\alpha, \zeta}^{\rm app}$. 

\item $(\wh\phi_\alpha^{-1})^* {\bm v}$ is in the Coulomb slice through ${\bm v}_{\alpha, \zeta}^{\rm app}$. 

\item If we write $(\wh\phi_\alpha^{-1})^* {\bm v} = (A, u)$, then 
\beqn
{\bf y}_\alpha = u^{\rm thin} \pitchfork H_\alpha.
\eeqn
\end{enumerate}
In particular, if $({\mc C}, {\bm v})$ is a stable solution to the gauged Witten equation, then 
\beqn
\big({\mc C}, {\bm v}, {\bf y}_\alpha, \wh\phi_\alpha, {\it 0}_\alpha \big)
\eeqn
is an $\alpha$-thickened solution.
\end{lemma}

\begin{proof}
By definition, there exist ${\bf y}$, $\wh\phi_\alpha$ and $\eta, \zeta$ satisfying the first two conditions above. Without loss of generality, we assume that ${\mc C} = {\mc C}_{\alpha, \eta, \zeta}$ and ${\bm v} = (P_{\alpha, \eta, \zeta}, A, u)$ is defined over ${\mc C}_{\alpha, \eta, \zeta}$. One can choose an order among points of ${\bf y}$, say $y_1, \ldots, y_l$. A small deformation of ${\bf y}$ can be written as 
\beqn
y_1 + w_1, \ldots, y_l + w_l,\ w_1, \ldots, w_l \in {\mb C}
\eeqn
since the resolution datum ${\bm r}_\alpha$ provides cylindrical coordinates near $y_1, \ldots, y_l$. Such a small deformation ${\bf y}' = {\bf y} + {\bf w}$ induces a small deformation of $\eta$ and $\zeta$, denoted by $\eta_{\bf w}$ and $\zeta_{\bf w}$. Then there exists a canonical isomorphism 
\beqn
{\mc C}_{\alpha, \eta, \zeta} \cong {\mc C}_{\alpha, \eta_{\bf w}, \zeta_{\bf w}}.
\eeqn
One can then identify the approximate solution ${\bm v}_{\alpha, \zeta_{\bf w}}^{\rm app}$ as defined over ${\mc C}_{\alpha, \eta, \zeta}$. 

Now we consider a nonlinear equation on variables ${\bf w}$ and $h_\alpha \in W^{2, p, w}_{\alpha, \eta, \zeta} (\Sigma_{\mc C}, {\rm ad} P)$. For each small ${\bf w}$, there is a unique small gauge transformation $k_{\bf w} = e^{h_{\bf w}}$ making $k_{\bf w}^* {\bm v}$ in the Coulomb slice of ${\bm v}_{\alpha, \zeta_{\bf w}}^{\rm app}$. Regard $H_\alpha$ as defined by the local vanishing locus of $f_1, \ldots, f_l: X \to {\mb C}$. Then define 
\beq\label{eqn913}
F_{\bm v} ({\bf w}, h) = \left[\begin{array}{c} f_1( k_{\bf w}^* u_{\bf w}^{\rm thin} (y_1)) \\
\cdots \\ f_l(k_{\bf w}^* u_{\bf w}^{\rm thin} (y_l)) \\
h - k_{\bf w}
\end{array}  \right].
\eeq
We would like to use the implicit function theorem to prove the existence of a zero. First consider the case that ${\bm v} = {\bm v}_{\alpha, \zeta}^{\rm app}$. At ${\bf w} = 0$, we claim that  
\begin{align*}
&\ \frac{\partial k_{\bf w}^* u_{\bf w}^{\rm thin} }{\partial w_i'}(y_i) \sim \partial_s u_\alpha^{\rm thin}(y_i) + {\mc X}_{\phi_\alpha^{\rm thin}}(y_i),\ &\ \frac{\partial k_{\bf w}^* u_{\bf w}^{\rm thin}}{\partial w_i''}(y_i) \sim \partial_t u_\alpha^{\rm thin} (y_i) + {\mc X}_{\psi_\alpha^{\rm thin}}(y_i).
\end{align*}
The error can be controlled by $|\zeta|$. Here $w_i = w_i' + {\bf i} w_i''$ is the complex coordinates of $w_i$. Indeed, the derivatives of $u_{\bf w}^{\rm thin}$ in $w_i'$ and $w_i''$ at $y_i$ are clearly $\partial_s u_\alpha^{\rm thin}$ and $\partial_t u_\alpha^{\rm thin}$. The the fact that the derivatives of $h_{\bf w}$ are very close to $\phi_\alpha^{\rm thin}$ and $\psi_\alpha^{\rm thin}$ follows from a gluing argument and Lemma \ref{lemma87}. The proof is left to the reader. 

By the last condition on the hypersurface $H_\alpha$ in the definition of thickening data (see Definition \ref{defn93}), the derivative of the first $l$ coordinates of \eqref{eqn913} is surjective. On the other hand, the derivative of $h- h_{\bf w}$ on $h$ is uniformly invertible. Hence, when $\epsilon_\alpha$ is small enough and $(\wh\phi_\alpha^{-1})^* {\bm v}$ is in the $\epsilon_\alpha$-neighborhood of ${\bm v}_{\alpha, \zeta}^{\rm app}$, the implicit function theorem implies the existence and uniqueness of a solution $F_{\bm v}({\bf w}, h) = 0$. 
\end{proof}

We also need to show the uniqueness of $({\bf y}_\alpha, \wh\phi_\alpha)$ up to automorphisms of $({\mc C}_\alpha, {\bm v}_\alpha)$. 

\begin{lemma}\label{lemma914}
Let $\alpha$ be a thickening datum. There exists $\epsilon_\alpha>0$ satisfying the following condition. Let $({\mc C}, {\bm v})$ be a stable solution to the gauged Witten equation which is $\epsilon_\alpha$-close to $\alpha$, and such that 
\begin{align*}
&\ \big( {\mc C}, {\bm v}, {\bf y}_\alpha,  \wh\phi_\alpha, {\it 0}_\alpha \big),\ &\ \big({\mc C}, {\bm v}, {\bf y}_\alpha', \wh\phi_\alpha', {\it 0}_\alpha \big)
\end{align*}
are both $\alpha$-thickened solution with  
\begin{align*}
&\ {\rm dist}( {\bm v}_{\alpha, \zeta}^{\rm app}, (\wh\phi_\alpha^{-1})^* {\bm v}) \leq \epsilon_\alpha,\ &\ {\rm dist}({\bm v}_{\alpha, \zeta'}^{\rm app}, ((\wh\phi_\alpha')^{-1})^* {\bm v}) \leq \epsilon_\alpha.
\end{align*}
Then ${\bf y}_\alpha = {\bf y}_\alpha'$ and there exists $\gamma \in {\rm Aut}({\mc C}_\alpha, {\bm v}_\alpha)$ such that 
\beqn
\gamma^{{\mc P}_\alpha} \circ \wh\phi_\alpha = \wh\phi_\alpha'.
\eeqn
\end{lemma}

\begin{proof}
Let $\epsilon_i$ be a sequence of positive numbers converging to zero and let $({\mc C}_i, {\bm v}_i)$ be a sequence of stable solutions to the gauged Witten equation which are $\epsilon_i$-close to $\alpha$. Suppose we have two sequences 
\begin{align*}
&\ \big( {\mc C}_i, {\bm v}_i, {\bm y}_i, \wh\phi_i, {\it 0}_\alpha \big),\ &\ \big( {\mc C}_i, {\bm v}_i, {\bm y}_i', \wh\phi_i', {\it 0}_\alpha \big)
\end{align*}
which are $\alpha$-thickened solutions satisfying 
\begin{align*}
&\ {\rm dist} \big( {\bm v}_{\alpha, \zeta_i}^{\rm app}, (\wh\phi_i^{-1})^* {\bm v}_i \big) \leq \epsilon_i,\ &\ {\rm dist} \big( {\bm v}_{\alpha, \zeta_i'}^{\rm app}, ((\wh\phi_i')^{-1})^* {\bm v}_i \big) \leq \epsilon_i.
\end{align*}
If we can show that for large $i$, ${\bf y}_i = {\bf y}_i'$ and there exists $\gamma_i \in {\rm Aut}({\mc C}_\alpha, {\bm v}_\alpha)$ such that 
\beqn
\gamma_i^{{\mc P}_\alpha} \circ \wh\phi_i = \wh\phi_i'
\eeqn
then the lemma follows. Suppose this is not the case. Then by taking a subsequence, we may assume that all ${\mc C}_i$ has the same topological type, and that the maps $\wh\phi_i, \wh\phi_i': {\mc C}_i \to {\mc C}_{\alpha, \eta_i, \zeta_i}$ have their combinatorial types independent of $i$. More precisely, let the $r$-spin dual graphs of ${\mc C}_i$ be $\Pi$ and the $r$-spin dual graph for ${\mc C}_\alpha$ be $\Gamma$. Then there are two fixed maps $\rho, \rho': \Gamma \to \Pi$ modelling the maps $\wh\phi_i$ and $\wh\phi_i'$. Let $v \in \V(\Pi)$ be an arbitrary vertex, which corresponds to a subtree $\Gamma_v \subset \Gamma$ and $\Gamma_v'  \subset \Gamma$. Then consider the isomorphisms 
\begin{align*}
&\ \wh\phi_{i, v}: {\mc C}_{i, v} \cong {\mc C}_{\alpha, \eta_i, \zeta_i, v},\ &\ \wh\phi_{i, v}': {\mc C}_{i, v} \cong {\mc C}_{\alpha, \eta_i', \zeta_i', v}.
\end{align*}
We claim that there is a subsequence (still indexed by $i$) for which 
\beqn
\wh\phi_{i, v}' \circ \wh\phi_{i, v}^{-1} 
\eeqn
converges to an isomorphism 
\beqn
\gamma_v: ( {\mc C}_{\alpha, v}, P_{\alpha, v}) \cong ({\mc C}_{\alpha, v}', P_{\alpha, v}').
\eeqn
Here ${\mc C}_{\alpha, v}, {\mc C}_{\alpha, v}' \subset {\mc C}_\alpha$ are the $r$-spin curves corresponding to the subtrees $\Gamma_v$ and $\Gamma_v'$. Moreover, one has 
\beqn
\gamma_v^* {\bm v}_{\alpha, v}' = {\bm v}_{\alpha, v}.
\eeqn
The proof of the claim is left to the reader. 

By the above claim, $\gamma_v$ induces a family of isomorphisms 
\beqn
( \gamma_v^{{\mc U}_\alpha}, \gamma_v^{{\mc P}_\alpha}): ( {\mc C}_{\alpha, \eta_v, \zeta_v}, P_{\alpha, \eta_v, \zeta_v} ) \cong ( {\mc C}_{\alpha, \gamma_v \eta_v, \gamma_v \zeta_v}, P_{\alpha, \gamma_v \eta_v, \gamma_v \zeta_v}).
\eeqn
Here $(\eta_v, \zeta_v)$ denote the deformation and gluing parameters on the subtree $\Gamma_v$, and $(\gamma_v \eta_v, \gamma_v \zeta_v)$ are the transformed deformation and gluing parameters on the subtree $\Gamma_v'$. We can see that for $i$ sufficiently large, one has that 
\beqn
| \gamma_v \eta_i - \eta_i'| \to 0,\ |\log( \gamma_v \zeta_{i, v}) - \log \zeta_{i, v}'| \to 0.
\eeqn
It means that there is a (canonical) isomorphism 
\beqn
\vartheta_{i, v}: {\mc C}_{\alpha, \gamma_v \eta_{i, v}, \gamma_v \zeta_{i, v}} \cong {\mc C}_{\alpha, \eta_{i, v}', \zeta_{i, v}'}
\eeqn
and that the extra markings $\gamma_v {\bf y}_{i, \eta_v, \zeta_v}$ differ from ${\bf y}_{i, \eta_v', \zeta_v'}$ by a small shift. Moreover, there is a bundle isomorphism $\wh\vartheta_{i, v}$ lifting $\vartheta_{i, v}$ such that we can write 
\beqn
e^{h_{i, v}} = k_{i, v} = \wh\phi_{i, v} (\wh\phi_{i, v}')^{-1} \wh\vartheta_{i, v}  \circ \gamma_v^{{\mc P}_\alpha}: P_{\alpha, \eta_{i, v}, \zeta_{i, v}} \to P_{\alpha, \eta_{i, v}, \zeta_{i, v}},
\eeqn
and 
\beqn
\| h_{i, v} \|_{W^{2,p, w}_{\eta_{i, v}, \zeta_{i, v}}} \to 0.
\eeqn

Now using $\wh\phi_{i, v}$ we may regard ${\bm v}_i$ restricted to the component $v$ (denoted by ${\bm v}_{i, v}$) is defined over ${\mc C}_{\alpha, \eta_{i, v}, \zeta_{i, v}}$ with extra marked points equal to ${\bf y}_{\alpha, \eta_{i, v}, \zeta_{i, v}}$ and we regard ${\bm v}_{\alpha, \zeta_{i, v}'}^{\rm app}$ defined on the same curve with possibly different extra markings ${\bf y}_{\alpha, \eta_{i, v}', \zeta_{i, v}'}$. But the two sets of markings are very close to each other (distance measured via the cylindrical metric). Moreover, we have the condition that ${\bm v}_{i, v}$ is in Coulomb gauge relative to ${\bm v}_{\alpha, \zeta_{i, v}}^{\rm app}$ while $k_{i, v}^* {\bm v}_{i, v}$ is in the Coulomb slice through ${\bm v}_{\alpha, \zeta_{i, v}'}^{\rm app}$, and 
\begin{align*}
&\ u_{i, v}^{\rm thin}( {\bf y}_{\alpha, \eta_{i,v}, \zeta_{i,v}}) \subset H_\alpha,\ &\ k_{i, v}^* u^{\rm thin}( {\bf y}_{\alpha, \eta_{i,v}', \zeta_{i,v}'}) \subset H_\alpha. 
\end{align*}
Then by the uniqueness part of the implicit function theorem and Lemma \ref{lemma913}, we have that for $i$ large, $k_{i, v}$ is the identity and ${\bf y}_{\alpha, \eta_v, \zeta_v} = {\bf y}_{\alpha, \eta_v', \zeta_v'}$. Hence for big $i$, one has
\beqn
\gamma_{i, v}^{{\mc P}_\alpha} \circ \wh\phi_{i, v} = \wh\phi_{i, v}'
\eeqn

Putting all components of ${\mc C}_i$ together, one obtains an automorphism $\gamma \in {\rm Aut}({\mc C}_\alpha, {\bm v}_\alpha)$ and obtains that for large $i$, 
\beqn
\gamma^{{\mc P}_\alpha} \wh\phi_i = \wh\phi_i'
\eeqn
and 
\beqn
{\bf y}_i = {\bf y}_i' \Longleftrightarrow \gamma^{{\mc U}_\alpha} {\bf y}_{\alpha, \eta_{i, v}, \zeta_{i, v}} = {\bf y}_{\alpha, \eta_{i, v}', \zeta_{i, v}'}.
\eeqn
Hence we finished the proof.
\end{proof}

\subsection{Thickened moduli space and gluing}

In this subsection we construct a chart of topological manifold for the moduli space of $\alpha$-thickened solutions. This induces a virtual orbifold chart of the moduli space $\ov{\mc M}$. We first show that the $\alpha$-thickened moduli space ${\mc M}_\alpha$ is a topological manifold near ${\bm p}_\alpha$. 

\begin{prop}
There is a $\Gammait_\alpha$-invariant neighborhood of ${\bm p}_\alpha$ which is a manifold. 
\end{prop}

It suffices to construct a chart which we will do by gluing. For each irreducible component $v$ of ${\mc C}_\alpha$, there is the augmented linearization map
\beqn
D_v: T_{{\bm v}_v} {\mc B}_v \to {\mc E}_v
\eeqn
which is a Fredholm operator. The domain of the total linearization is the subspace of direct sum of $T_{{\bm v}_v}{\mc B}_v$ under the constraint given by the matching condition. Let $T_{{\bm v}_v}'{\mc B}_v \subset T_{{\bm v}_v} {\mc B}_v$ be the finite-codimensional subspace consisting of infinitesimal deformations whose values at special points are zero. Then set
\beqn
T_{{\bm v}_\alpha}' {\mc B}_\alpha = \bigoplus_v T_{{\bm v}_v}' {\mc B}_v
\eeqn
The restriction $D_\alpha: T_{{\bm v}_\alpha}' {\mc B}_\alpha \to {\mc E}_\alpha$ is still Fredholm.

The following is the main gluing theorem for the $\alpha$-thickened moduli space. Let ${\mc M}_\alpha^+$ be the moduli space of the same type of objects as ${\mc M}_\alpha$ without imposing the condition that ${\bf y}_\alpha = u^{\rm thin}\pitchfork H_\alpha$. Then ${\mc M}_\alpha^+$ contains ${\mc M}_\alpha$ as an $\Gammait_\alpha$-invariant subset. 

\begin{prop}\label{prop916}
Abbreviate ${\bm \xi} = (\xi, \eta, \zeta)$. Let $\alpha$ be a thickening datum. Then for $\epsilon>0$ sufficiently small, there exist a family of objects 
\beq\label{eqn914}
\hat{\bm v}_{\alpha, {\bm \xi}} = \big( {\bm v}_{\alpha, \bm \xi }, e_{\alpha, {\bm \xi}} \big),\ {\rm where}\ {\bm \xi} = (\xi, \eta, \zeta) \in {\mc V}_{\alpha, \rm map}^\epsilon \times {\mc V}_{\alpha, \rm def}^\epsilon \times {\mc V}_{\alpha, \rm res}^\epsilon
\eeq
where ${\bm v}_{\alpha, {\bm \xi}}$ is a gauged map over $P_{\alpha, \eta, \zeta} \to {\mc C}_{\alpha, \eta, \zeta}$ and $e_{\alpha, {\bm \xi}} \in {\bm E}_\alpha$. They satisfy the following: 
\begin{enumerate}
\item When ${\bm \xi} = (\xi, \eta, 0)$, $\hat{\bm v}_{\alpha, {\bm \xi}}$ coincides with $\hat{\bm v}_{\alpha, \xi, \eta}$ of \eqref{eqn95}. 

\item The family is $\Gammait_\alpha$-equivariant in the sense that for any $\gamma_\alpha \in \Gammait_\alpha$,
\beqn
( {\bm v}_{\alpha, \gamma_\alpha {\bm \xi}}, e_{\alpha, \gamma_\alpha {\bm \xi}} ) = ( {\bm v}_{\alpha, {\bm \xi}} \circ \gamma_\alpha^{-1}, \gamma_\alpha e_{\alpha, {\bm \xi}}).
\eeqn

\item  There holds
\beqn
\hat{\mc F}_{\alpha, \eta, \zeta} \big( \hat{\bm v}_{\alpha, {\bm \xi}} \big) = 0.
\eeqn

\item The natural map ${\mc V}_{\alpha, {\rm map}}^\epsilon \times {\mc V}_{\alpha, {\rm def}}^\epsilon \times {\mc V}_{\alpha, {\rm res}}^\epsilon  \to {\mc M}_\alpha^+$ defined by 
\beq\label{eqn915}
{\bm \xi} \mapsto \big[ {\mc C}_{\alpha, \eta, \zeta}, {\bm v}_{\alpha, {\bm \xi}}, {\bf y}_{\alpha, \eta, \zeta}, \wh\phi_{\alpha, \eta, \zeta}, e_{\alpha, {\bm \xi}} \big]
\eeq
is a $\Gammait_\alpha$-equivariant homeomorphism onto an open neighborhood of ${\bm p}_\alpha$. Here ${\bf y}_{\alpha, \eta, \zeta}$ and $\wh\phi_{\alpha, \eta, \zeta}$ are defined tautologically. 
\end{enumerate}
\end{prop}

As a consequence one obtains a local chart of the original moduli space. Namely, define 
\beqn
\hat U_\alpha^\epsilon = \Big\{ {\bm \xi} \in {\mc V}_{\alpha, {\rm map}}^\epsilon \times {\mc V}_{\alpha, {\rm def}}^\epsilon \times {\mc V}_{\alpha, {\rm res}}^\epsilon \ |\ u_{\alpha, {\bm \xi}}^{\rm thin} \pitchfork H_\alpha = {\bf y}_{\alpha, \eta,\zeta} \Big\}.
\eeqn
In the following notations, we omit the dependence on $\epsilon$. This is a topological manifold acted continuously by $\Gammait_\alpha$. Define 
\begin{align}
&\ U_\alpha = \hat U_\alpha/ \Gammait_\alpha,\ &\ E_\alpha = ( \hat U_\alpha \times {\bm E}_\alpha)/ \Gammait_\alpha.
\end{align}
Then $E_\alpha \to U_\alpha$ is an orbifold vector bundle with a section induced by the natural map $\hat S_\alpha: {\mc M}_\alpha \to {\bm E}_\alpha$, denoted by  
\beq
S_\alpha: U_\alpha \to E_\alpha
\eeq
and there is an induced map 
\begin{align}
&\ \psi_\alpha: S_\alpha^{-1}(0) \to \ov{\mc M}_{\sf \Gamma}\ &\ \psi_\alpha \big[ \hat{\bm v}_{\alpha, {\bm \xi}} \big] = \big[ {\mc C}_{\alpha, \eta, \zeta}, {\bm v}_{\alpha, {\bm \xi}} \big].
\end{align}
Let the image of $\psi_\alpha$ be $F_\alpha$. 

\begin{cor}\label{cor917}
When $\epsilon$ is small enough, the 5-tuple 
\beqn
C_\alpha = (U_\alpha, E_\alpha, S_\alpha, \psi_\alpha, F_\alpha)
\eeqn
is a topological virtual orbifold chart around of $\ov{\mc M}_{\sf \Gamma}$ around $p_\alpha$. 
\end{cor}

\begin{proof}
According to Definition \ref{defn610}, one only needs to prove that $\psi_\alpha$ is a homeomorphism onto $F_\alpha$ and $F_\alpha$ is an open neighborhood of $p_\alpha$. Since $U_\alpha$ is a topological orbifold which is locally compact, while the moduli space $\ov{\mc M}_{\sf\Gamma}$ is Hausdorff, to prove that $\psi_\alpha$ is a homeomorphism, one only needs to verify it is one-to-one. Suppose there are two isomorphism classes of $\alpha$-thickened solutions 
\begin{align*}
&\ {\bm p} = \Big[ {\mc C}, {\bm v}, {\bf y}_\alpha, \wh\phi_\alpha, {\it 0}_\alpha \Big],\ &\ {\bm p}'=  \Big[ {\mc C}', {\bm v}', {\bf y}_\alpha', \wh\phi_\alpha', {\it 0}_\alpha \Big]
\end{align*}
whose $\Gammait_\alpha$-orbits are mapped to the same point in $\ov{\mc M}_{\sf \Gamma}$. Then by definition, $({\mc C}, {\bm v})$ is isomorphic to $({\mc C}', {\bm v}')$, meaning that there is an isomorphism $\rho: {\mc C} \to {\mc C}'$ as $r$-spin curves and an isomorphism $h: P \to P'$ as $K$-bundles that cover $\rho$ such that $h^* {\bm v}' = {\bm v}$. Therefore we may regard the two $\alpha$-thickened solutions are defined on the same $r$-spin curve ${\mc C}$ with the same gauged map ${\bm v}$. Their differences are the stabilizing points ${\bf y}_\alpha$, ${\bf y}_\alpha'$ and the bundle inclusions $\wh\phi_\alpha$, $\wh\phi_\alpha'$. Moreover, by Proposition \ref{prop916}, they are isomorphic as $\alpha$-thickened solutions to two specific exact solutions, and hence in particular $\epsilon_\alpha$-close to $\alpha$. Then by Lemma \ref{lemma914}, ${\bf y}_\alpha = {\bf y}_\alpha'$, and $\wh\phi_\alpha$ and $\wh\phi_\alpha'$ differ by an action of $\Gammait_\alpha$. It means that the two points in ${\mc M}_\alpha$ are on the same $\Gammait_\alpha$-orbit, from which injectivity of $\psi_\alpha$ follows.

It remains to show the local surjectivity, namely $F_\alpha$ contains an open neighborhood of $p_\alpha$. Suppose on the contrary that this is not true. Then there exists a sequence of stable solutions $({\mc C}_i, {\bm v}_i)$ which converge modulo gauge in c.c.t. to $({\mc C}_\alpha, {\bm v}_\alpha )$. Then by Lemma \ref{lemma912}, given any $\epsilon>0$, for sufficiently large $i$, $({\mc C}_i, {\bm v}_i)$ will be $\epsilon$-close to $\alpha$. Then by Lemma \ref{lemma913}, for large $i$ one can upgrade $({\mc C}_i, {\bm v}_i)$ to an $\alpha$-thickened solution, and the isomorphism classes of the sequence of $\alpha$-thickened solutions converge in the topology of ${\mc M}_\alpha$ to ${\bm p}_\alpha$. Hence for large $i$, the isomorphism class $p_i$ of $({\mc C}_i, {\bm v}_i)$ is indeed in the image of $\psi_\alpha$, which is a contradiction. 
\end{proof}

\subsection{Proof of Proposition \ref{prop916}}

The strategy of the proof is standard. However we need to set up the Fredholm theory more carefully. 

First we construct a family of approximate solutions using the pregluing construction as in Definition \ref{defn95}. More precisely, we have obtained a family of $\alpha$-thickened solutions 
\beqn
\Big( {\mc C}_{\alpha, \xi, \eta}, {\bf y}_{\alpha, \xi, \eta}, {\bf v}_{\alpha, \xi, \eta}, \phi_{\alpha, \xi, \eta}, g_{\alpha, \xi, \eta}, e_{\alpha, \xi, \eta} \Big).
\eeqn
Notice that the domains ${\mc C}_{\alpha, \xi, \eta}$ are smoothly identified with the domain ${\mc C}_\alpha$ where the identification is contained in the thickening datum $\alpha$. Hence we only need to repeat the pregluing construction of Definition \ref{defn95} while replacing the central element ${\bm v}_\alpha$ by ${\bm v}_{\alpha, \xi, \eta}$. Then for each small gluing parameter $\zeta \in {\mc V}_{\alpha, {\rm res}}^\epsilon$, denoting ${\bm \xi} = (\xi, \eta, \zeta)$, we obtain a gauged map 
\beqn
{\bm v}_{\alpha, {\bm \xi}}^{\rm app}:= {\bm v}_{\alpha, \xi, \eta, \zeta}^{\rm app}.
\eeqn
Define 
\beqn
e_{\alpha, {\bm \xi}}^{\rm app} = e_{\alpha, \xi, \eta, \zeta}^{\rm app} = e_{\alpha, \xi, \eta}
\eeqn
and 
\beqn
\hat {\bm v}_{\alpha, {\bm \xi}}^{\rm app} = ( {\bm v}_{\alpha, {\bm \xi}}^{\rm app}, e_{\alpha, {\bm \xi}}^{\rm app}).
\eeqn

Our next task is to ``estimate the error'' of the approximate solution. 

\begin{lemma}\label{lemma918}
There exist $C_\alpha>0$ and $\epsilon_\alpha>0$ such that for $|\zeta| \leq \epsilon_\alpha$, we have 
\beqn
\| \hat {\mc F}_{\alpha, \eta, \zeta}  ( \hat {\bm v}_{\alpha, \bm \xi}^{\rm app} )  \|_{W_{\alpha, \eta, \zeta}^{0,p,w}} \leq C_\alpha | \zeta |^{2-\tau - w}.
\eeqn
\end{lemma}

\begin{proof}
For every irreducible component $\Sigma_v \subset \Sigma_{\alpha, \eta, \zeta}$, the gauged Witten equation plus the gauge fixing condition relative to the central approximate solution has three components, denoted by $\hat {\mc F}_1, \hat {\mc F}_2, \hat {\mc F}_3$. For $a = 1, 2, 3$, we first estimate $\hat {\mc F}_a(\hat{\bm v}_{\alpha, {\bm \xi}}^{\rm app})$ over $\Sigma_T$. Indeed, over $\Sigma_T\subset \Sigma_{\alpha, \eta, \zeta}$, we have $\hat {\bm v}_{\alpha, \xi, \eta, \zeta}^{\rm app} = \hat{\bm v}_{\alpha, \xi, \eta}$. Hence we have
\beqn
\hat {\mc F}_{\alpha, \eta, \zeta} ( \hat {\bm v}_{\bm \xi}^{\rm app} ) = {\mc F}_{\alpha, \eta, \zeta} (\tilde {\bm v}_{\eta, \eta'}) + \iota_{\alpha, \eta, \zeta }(e_{\alpha, \xi, \eta}, u_{\alpha, \xi, \eta}) = \iota_{\alpha, \eta, \zeta}(e_{\alpha, \xi, \eta},  u_{\alpha, \xi, \eta}) - \iota_{\alpha, \eta,  0} (e_{\alpha,\xi, \eta}, u_{\alpha, \xi, \eta}).
\eeqn
Since the inclusion $\iota_\alpha$ is smooth in $\zeta$, and the supports of the images of $\iota_{\alpha, \eta, \zeta}$ are contained in a region where the weight function is uniformly bounded, there is $C_\alpha>0$ such that
\beqn
\| \hat {\mc F}_{\alpha, \eta, \zeta} ( \hat {\bm v}_{\alpha, \bm \eta}^{\rm app} ) \|_{W_{\alpha, \eta, \zeta}^{0, p, w}(\Sigma_T)} \leq C_\alpha | \zeta |. 
\eeqn

Then look at the neck regions, where the domain complex structure is fixed and the obstruction vanishes. Hence
\beqn
\hat {\mc F}_a ( \hat {\bm v}_{\alpha, {\bm \xi}}^{\rm app}) = {\mc F}_a ( {\bm v}_{\alpha, \bm \xi}^{\rm app}).
\eeqn
Recall how the approximate solution is defined in Definition \ref{defn95}. Here we see over the interval $C_{\pm T} = \pm [0, T]$ the approximate solution is 
\beqn
u_{\alpha, {\bm \xi}}^{\rm app} = \exp_{x_w (t)} ( \rho_\pm^T \xi_\pm) ,\ \phi^{\rm app} = \rho_\pm^T \beta_{\pm, s},\ \psi^{\rm app} = \rho_\pm^T \beta_{\pm, t} + \eta_w (t). 
\eeqn
Here $(x(t), \eta(t))$ is the limiting critical loop at the node $w$ which satisfies
\beqn
x'(t) + {\mc X}_{\eta(t)} (x(t)) = 0. 
\eeqn
Using the splitting $TX = H_X \oplus G_X$, we can decompose 
\beqn
\xi_\pm = \xi_\pm^H + \xi_\pm^G. 
\eeqn
Then by the exponential decay property (see Theorem \ref{thm46}), one has 
\beq\label{eqn919}
\sup_{C_{\pm T}} \Big[ \| \xi_\pm (s, t) \|  + \| \partial_s \xi_\pm (s, t)  \| + \| \nabla_t \xi_\pm (s, t) + \nabla_{\xi_\pm} {\mc X}_\eta \|  \Big] \leq e^{- \tau T},
\eeq
and for $a = s, t$,
\beq\label{eqn920}
\sup_{C_{\pm T}} \Big[ \| \beta_{\pm, a}  (s, t) \|  + \| \partial_s \beta_{\pm, a} (s, t)  \| + \| \partial_t \beta_{\pm, a} (s, t) + [\eta(t), \beta_{\pm, a} (s, t)] \| \Big] \leq e^{- \tau T}.
\eeq
Moreover, denote the derivative of the exponential map by two maps $E_1, E_2$, namely
\beqn
d \exp_x \xi = E_1(x, \exp_x \xi) dx + E_2(x, \exp_x \xi) \nabla \xi. 
\eeqn
Then because the metric is $\ubar K$-invariant, for all $ \ubar a \in \ubar {\mf k}$, one has (see \cite[Lemma C.1]{Gaio_Salamon_2005})
\beq\label{eqn921}
{\mc X}_{\ubar a} ( \exp_x v) = E_1(x, \exp_x v) {\mc X}_{\ubar a}(x) + E_2( x, \exp_x v) \nabla_v {\mc X}_{\ubar a}(x).
\eeq

Moreover, the bundle is trivialized over the neck region. Hence we can identify the connections with 1-forms and sections with maps. 

\vspace{0.2cm}

\noindent {\bf --Estimate ${\mc F}_1$ over the neck region.} One has
\beqn
{\mc F}_1( {\bm v}_{\alpha, \bm \xi}^{\rm app}) = \partial_s u_{\alpha, \bm \xi}^{\rm app} + {\mc X}_{\phi_{\alpha, {\bm \xi}}^{\rm app}} (u_{\alpha, {\bm \xi}}^{\rm app}) + J \Big( \partial_t u_{\alpha, {\bm \xi}}^{\rm app} + {\mc X}_{\psi_{\alpha, {\bm \xi}}^{\rm app}} (u^{\rm app}) \Big) + \nabla W (u_{\alpha, {\bm \xi}}^{\rm app}).
\eeqn
We estimate each term above as follows. First, by the definition of $E_1$, $E_2$, one has
\beqn
\Big\| \partial_s u_{\alpha, {\bm \xi}}^{\rm app} \Big\| = \Big\| \partial_s \big( \exp_{x_w(t)} \rho_\pm^T \xi_\pm \big) \Big\| = \Big\| E_2 \big( \partial_s \rho_\pm^T \xi_\pm + \rho_\pm^T \partial_s \xi_\pm \big)\Big\|  \leq C \Big( \| \xi_\pm \| + \| \partial_s \xi_\pm \| \Big).
\eeqn
Further by \eqref{eqn921}, one has
\beqn
\Big\| {\mc X}_{\phi_{\alpha, {\bm \xi}}^{\rm app}}(u_{\alpha, {\bm \xi}}^{\rm app}) \Big\| = \Big\| {\mc X}_{\rho_\pm^T \beta_{\pm, s}} \big( 	\exp_{x_w(t)} \rho_\pm^T \xi_\pm  \big)  \Big\| =  \rho_\pm^T \Big\| E_1 \big( {\mc X}_{\beta_{\pm, s}}(x_w(t)) \big) + E_2 \big( \nabla_{\rho_\pm^T \xi_\pm} {\mc X}_{\beta_{\pm, s}} \big)  \Big\| \leq C \big\| \beta_{\pm, s} \big\|.
\eeqn
Also by \eqref{eqn921}, one has
\beqn
\begin{split}
&\ \partial_t u_{\alpha, {\bm \xi}}^{\rm app} + {\mc X}_{\psi_{\alpha, {\bm \xi}}^{\rm app}} (u_{\alpha, {\bm \xi}}^{\rm app})\\
= &\ \partial_t \big( \exp_{x_w(t)} \rho_\pm^T \xi_\pm  \big) + {\mc X}_{\psi_{\alpha, {\bm \xi}}^{\rm app}} \big( \exp_{x_w(t)} \rho_\pm^T \xi_\pm \big)\\
= &\ E_1 x_w'(t) + \rho_\pm^T E_2 ( \nabla_t \xi_\pm) + E_1 {\mc X}_{\psi_{\alpha, {\bm \xi}}^{\rm app}} + E_2 \nabla_{\rho_\pm^T \xi_\pm} {\mc X}_{\psi_{\alpha, {\bm \xi}}^{\rm app}} \\
= &\ \rho_\pm^T E_2( \nabla_t \xi_\pm) + E_1 {\mc X}_{\rho_\pm^T \beta_{\pm, t}} + \rho_\pm^T E_2 \nabla_{\xi_\pm} {\mc X}_{\psi_{\alpha, {\bm \xi}}^{\rm app}} \\
= &\ \rho_\pm^T E_2 \big( \nabla_t \xi_\pm + \nabla_{\xi_\pm} {\mc X}_{\eta_w} \big) + \rho_\pm^T E_1 ( {\mc X}_{\beta_{\pm t}}) +  (\rho_\pm^T)^2 E_2  \nabla_{\xi_\pm} {\mc X}_{\beta_{\pm, t}}
\end{split}
\eeqn
Hence 
\beqn
\Big\| \partial_t u_{\alpha, {\bm \xi}}^{\rm app} + {\mc X}_{\psi_{\alpha, {\bm \xi}}^{\rm app}}(u_{\alpha, {\bm \xi}}^{\rm app} ) \Big\| \leq C \Big( \| \xi_\pm \| + \| \nabla_t \xi_\pm + \nabla_{\xi_\pm} {\mc X}_{\eta_w} \| + \| \beta_{\pm, t} \| \Big). 
\eeqn
Lastly because $\nabla W(x_w(t)) \equiv 0$, one has
\beqn
\Big\| \nabla W ( u_{\alpha, {\bm \xi}}^{\rm app} ) \Big\| = \Big\| \nabla W ( \exp_{x_w(t)} \rho_\pm^T \xi_\pm ) \Big\| \leq C \| \xi_\pm \|.
\eeqn
In all the above estimates, the constant $C$ can be made independent of ${\bm \xi}$. Then using \eqref{eqn919}, \eqref{eqn920}, one has
\beq
\Big\| {\mc F}_1( {\bm v}_{\alpha, {\bm \xi}}^{\rm app} ) \Big\|_{L^{p, w}(C_{\pm, T})} \leq C e^{-\tau T}  \Big[ \int_{C_{\pm, T}} e^{pws} ds dt  \Big]^{\frac{1}{p}} \leq C e^{- (\tau -w) T}.
\eeq

\noindent {\bf --Estimate ${\mc F}_2$ over the neck region.} Over the neck region one has
\beqn
\begin{split}
{\mc F}_2({\bm v}_{\alpha, \bm \xi}^{\rm app}) = &\ \partial_s \psi_{\alpha, {\bm \xi}}^{\rm app} - \partial_t \phi_{\alpha, {\bm \xi}}^{\rm app} + \Big[ \phi_{\alpha, {\bm \xi}}^{\rm app}, \psi_{\alpha, {\bm \xi}}^{\rm app} \Big] + \mu (u_{\alpha, {\bm \xi}}^{\rm app})\\
= &\ \partial_s ( \rho_\pm^T \beta_{\pm, t} ) - \partial_t ( \rho_\pm^T \beta_{\pm, s} ) + \rho_\pm^T \big[ \beta_{\pm, s}, \eta_w (t) + \rho_\pm^T \beta_{\pm, t} \big] + \mu \Big( \exp_{x_w(t)} \rho_\pm^T \xi_\pm \Big).
\end{split}
\eeqn
Using \eqref{eqn920}, one has 
\beqn
\Big\| \partial_s ( \rho_\pm^T \beta_{\pm, t} ) - \partial_t ( \rho_\pm^T \beta_{\pm, s} ) + \rho_\pm^T [ \beta_{\pm, s}, \eta_w (t) + \rho_\pm^T \beta_{\pm, t} ] \Big\|_{L^{p, w}(C_{\pm T})} \leq C e^{- (\tau - w)T}.
\eeqn
On the other hand, we know $\mu(x_w(t)) \equiv 0$. Hence by \eqref{eqn919}, one has
\beqn
\Big\| \mu \Big( \exp_{x_w(t)} \rho_\pm^T \xi_\pm \Big)\Big\|_{L^{p, w}_{\alpha, \zeta} (C_{\pm T})} \leq C \big\| \rho_\pm^T \xi_\pm^G \big\|_{L^{p, w} (C_{\pm T})} \leq C e^{- ( \tau - w) T}. 
\eeqn
Hence we have
\beq
\Big\| {\mc F}_2 ({\bm v}_{\alpha, \bm \xi}^{\rm app}) \Big\|_{L^{p, w}_{\alpha, \zeta} (C_{\pm T})} \leq C e^{- (\tau - w)T}.
\eeq

\noindent {\bf --Estimate ${\mc F}_3$ over the neck region.} The estimate of ${\mc F}_3$ is slightly different because the gauge fixing is relative to the approximate solution. Let the two sides of the original singular solutions be 
\begin{align*}
&\ (u_\alpha^\pm, \phi_\alpha^\pm, \psi_\alpha^\pm),\ &\ (u_{\alpha, \xi, \eta}^\pm, \phi_{\alpha, \xi, \eta}^\pm, \psi_{\alpha, \xi, \eta}^\pm).
\end{align*}
The central approximate solution is 
\beqn
(u_{\alpha, \zeta}^{\rm app}, \phi_{\alpha, \zeta}^{\rm app}, \psi_{\alpha, \zeta}^{\rm app}).\eeqn
Over the semi-infinite cylinders $[T, +\infty) \times S^1$ we write
\beqn
u_{\alpha, \xi, \eta}^\pm = \exp_{u_\alpha^\pm} v_{\alpha, \xi, \eta}^\pm.
\eeqn
Over the long cylinder $[-T, T] \times S^1$ we write
\beqn
u_{\alpha, {\bm \xi}}^{\rm app} = \exp_{u_{\alpha, \zeta}^{\rm app}} v_{\alpha, {\bm \xi}}.
\eeqn
Then 
\begin{multline*}
{\mc F}_3({\bm v}_{\alpha, {\bm \xi}}^{\rm app}) = \partial_s ( \phi_{\alpha, {\bm \xi}}^{\rm app} - \phi_{\alpha, \zeta}^{\rm app}) + \big[ \phi_{\alpha, \zeta}^{\rm app}, \phi_{\alpha, {\bm \xi}}^{\rm app} - \phi_{\alpha, \zeta}^{\rm app} \big] \\
+ \partial_t \big( \psi_{\alpha, {\bm \xi}}^{\rm app} - \psi_{\alpha, \zeta}^{\rm app} \big) + \big[ \psi_{\alpha, \zeta}^{\rm app}, \psi_{\alpha, {\bm \xi}}^{\rm app} - \psi_{\alpha, \zeta}^{\rm app} \big] + d\mu (u_{\alpha, \zeta}^{\rm app}) \cdot J v_{\alpha, {\bm \xi}}
\end{multline*}
We still estimate each term separately. 
\beqn
\Big\| \partial_s \big( \phi_{\alpha, {\bm \xi}}^{\rm app} - \phi_{\alpha, \zeta}^{\rm app} \big) \Big\| = \Big\| \partial_s \big( \rho_\pm^T ( \phi_{\alpha, \xi, \eta}^\pm - \phi_\alpha^\pm ) \big) \Big\| \leq C \Big( \| \phi_{\alpha, \xi, \eta}^\pm - \phi_\alpha^\pm \| + \| \partial_s ( \phi_{\alpha, \xi, \eta}^\pm - \phi_\alpha^\pm) \| \Big).
\eeqn

Lastly we estimate $d\mu \cdot J v_{\alpha, {\bm \xi}}$. We know that 
\beqn
\sup_{[T, +\infty) \times S^1} \Big\| d \mu(u_\alpha^\pm ) \cdot J v_{\alpha, \xi, \eta}^\pm  \Big\| \leq C e^{- \tau T}.
\eeqn
Moreover, by \eqref{eqn919}, over $C_{\pm T}$, the $C^0$ distance between $u_{\alpha, \zeta}^{\rm app}$ and $u_\alpha^\pm$ is controlled by $e^{-\tau T}$, while the $C^0$ distance between $u_{\alpha, {\bm \xi}}^{\rm app}$ and $u_{\alpha, \xi, \eta}^\pm$ is controlled by $e^{-\tau T}$. Hence we have 
\beqn
\sup_{C_{\pm T}} \Big\| d\mu(u_{\alpha, \zeta}^{\rm app}) \cdot J v_{\alpha, {\bm \xi}} \Big\| \leq C e^{-\tau T}.\qedhere
\eeqn
\end{proof}

\subsubsection{Estimate the variation of the linear operator}

Consider the ${\bm E}_\alpha$-perturbed, augmented operator 
\beqn
\hat{\mc F}_{\alpha, \zeta}: {\bm E}_\alpha \times  {\mc B}_{\alpha, \zeta}^{p,w} \to {\mc E}_{\alpha, \zeta}^{p,w}.
\eeqn
We identify a neighborhood of $\hat{\bm v}_{\alpha, \zeta}^{\rm app}$ in ${\bm E}_\alpha  \times {\mc B}_{\alpha, \zeta}^{p,w}$ with a neighborhood $\hat{\mc O}_{\alpha, \zeta}$ of the origin in the tangent space. Then an element $\hat{\bm v}$ near $\hat{\bm v}_{\alpha, \zeta}^{\rm app}$ is identified with a point $\hat{\bm x} \in \hat{\mc O}_{\alpha, \zeta}$. We also trivialize the bundle ${\mc E}_{\alpha, \zeta}^{p,w}$. Then the map $\hat{\mc F}_{\alpha, \eta, \zeta}$ can be viewed as a nonlinear operator between two Banach spaces. To apply the implicit function theorem one needs to have a quadratic estimate.

\begin{lemma}\label{lemma919} {\rm (Quadratic estimate)} 
There exist $\delta_\alpha>0$, $\epsilon_\alpha>0$ and $C_\alpha>0$ such that for all $|\zeta| \leq \epsilon_\alpha$ and $\hat {\bm x}, \hat {\bm x}{}' \in \hat {\mc O}_{\alpha, \zeta}$ corresponding to $\hat{\bm v}$, $\hat{\bm v}{}'$ with 
\beqn
\| \hat {\bm x} \|,\ \|\hat {\bm x}{}' \| \leq \delta_\alpha,
\eeqn
we have
\beqn
\Big\| \hat D_{\hat {\bm v}} - \hat D_{\hat {\bm v}{}'} \Big\| \leq C_\alpha \Big( \| \hat {\bm x} \| + \| \hat {\bm x}{}' \| \Big) \Big\| \hat {\bm x} - \hat {\bm x}{}' \Big\|.
\eeqn
\end{lemma}

\begin{proof}
It can be proved by straightforward calculation. Notice that because the weight functions defining the weighted Sobolev norms are uniformly bounded from below, we always have the Sobolev embedding 
\beqn
W_{\alpha, \eta, \zeta}^{1,p,w} \hookrightarrow C^0
\eeqn
with a uniform constant. The details are left to the reader.
\end{proof}

\subsubsection{Constructing the right inverse}

We construct a right inverse to the linear operator $\hat D_{\alpha, \zeta}$. Choose a constant $\hbar\in (0, \frac{1}{2})$. Define 
\beqn
\chi_\pm^{\hbar T}: {\mb R} \to [0, 1]
\eeqn
such that 
\begin{align*}
&\ {\rm supp} \chi_-^{\hbar T} \subset (-\infty, \hbar T],\ &\ {\rm supp} \chi_+^{\hbar T} \subset [-\hbar T, +\infty);
\end{align*}
\begin{align}\label{eqn924}
&\ \chi_-^{\hbar T}|_{(-\infty, 0]} \equiv \chi_+^{\hbar T}|_{[0, +\infty)} \equiv 1,\ &\ \sup | \nabla \chi_-^{\hbar T} |,\ \sup| \nabla \chi_+^{\hbar T}| \leq \frac{2}{\hbar T}.
\end{align}

Now we compare the approximate solution ${\bm v}_{\alpha, \zeta}^{\rm app}$ with ${\bm v}_\alpha$. For each irreducible component $v\in {\rm Irre} \Sigma_{{\mc C}_\alpha}$, we can identify $\Sigma_{v, 2T}$ with a compact subset of $\Sigma_{\alpha, \zeta}$ and 
\beqn
\Sigma_{\alpha, \zeta} = \bigcup_v \Sigma_{v, 2T}.
\eeqn
The intersections of the pieces on the right hand side are circles. Furthermore, by the construction of the approximate solution, we know that the distance 
\beqn
{\rm dist} \Big( u_{\alpha, \zeta}^{\rm app}|_{\Sigma_{v, 2T}}, u_\alpha|_{\Sigma_{v, 2T}} \Big)
\eeqn
is sufficiently small. Then we can define the parallel transport 
\beqn
{\it Par}_v: u_\alpha^* TV|_{\Sigma_{v, 2T}} \to ( u_{\alpha, \zeta}^{\rm app})^* TV|_{\Sigma_{v, 2T}}.
\eeqn
The union over all components $v$ is denoted by ${\it Par}$. Further, the bundle ${\rm ad} P_{\mc C}$ over $\Sigma_\alpha$ is canonically identified with the adjoint bundle over $\Sigma_{\alpha, \zeta}$ via the thickening data. 

Now we define two maps
\beqn
{\it Cut}: {\mc E}_{{\bm v}_{\alpha, \zeta}^{\rm app}} \to {\mc E}_{{\bm v}_\alpha},
\eeqn
and
\beqn
{\it Glue}: {\bm E}_\alpha \oplus  T_{{\bm v}_\alpha} {\mc B}_\alpha^{p,w} \to {\bm E}_\alpha \oplus T_{{\bm v}_{\alpha, \zeta}^{\rm app}} {\mc B}_{\alpha, \zeta}^{p,w}.
\eeqn

\noindent {\bf --Definition of {\it Cut}}. Given ${\bm \varsigma} \in {\mc E}_{{\bm v}_{\alpha, \zeta}^{\rm app}}$ and each irreducible component $v$ of ${\mc C}_\alpha$, restrict ${\bm \varsigma}$ to $\Sigma_{v, 2T}$, and use the inverse of ${\it Par}_v$ to obtain an element of ${\mc E}_{{\bm v}_\alpha}$ over the component $v$. This is well-defined as an element of ${\mc E}$ is only required to have $L^{p, \delta}$-regularity.

\noindent {\bf --Definition of {\it Glue}}. ${\it Glue}$ maps ${\bm E}_\alpha$ identically to ${\bm E}_\alpha$. For ${\bm \xi} \in T_{{\bm v}_\alpha} {\mc B}_\alpha^{p,w}$ and each irreducible component $v$ of ${\mc C}_\alpha$, restrict ${\bm \xi}_v$ to $\Sigma_{v, (2+\hbar) T}$, use the parallel transport to transfer it to the domain of ${\bm v}_{\alpha, \zeta}^{\rm app}$, and then use the cut-off function $\chi_\pm^{\hbar T}$ supported there whose derivative is at most $2/ \hbar T$ and whose value is $1$ on $\Sigma_{v, 2T}$ to extend to the long cylinders and glue the obtained sections on different components together. 

We estimate the norms of ${\it Cut}$ and ${\it Glue}$. 

\begin{lemma}\label{lemma920}
There is a universal constant $C>0$ such that 
\begin{align*}
& \| {\it Cut} \| \leq C,\ &\ \| {\it Glue}\| \leq C. 
\end{align*}
\end{lemma}

\begin{proof}
By the way we define the norms (see Subsection \ref{subsection81} and Definition \ref{defn98}, Definition \ref{defn99}), we see that $\| {\it Cut}\| \leq 1$. To estimate the norm of ${\it Glue}$, it remains to estimate the norm of the first order derivatives. Given ${\bm \xi}_v \in T_{{\bm v}_\alpha} {\mc B}_\alpha^{p,w}$ supported on the component $v$, we see that 
\beqn
\| {\it Glue} (\xi_v) \|_{W^{1,p, \delta}} = \| {\bm \xi}_v\|_{W^{1, p,\delta}(\Sigma_{v, 2T})} + \| {\it Glue}({\bm \xi}_v)) \|_{W^{1,p,\delta}([2T, (2 + \hbar )T] \times S^1)}.
\eeqn
The last term can be controlled by
\beqn
C \Big( 1 + \frac{2}{\hbar T} \Big) \| {\bm \xi}_v \|_{W^{1,p,\delta}}. 
\eeqn
This gives a bound on the norm of ${\it Glue}$. 
\end{proof}

Now we can define the approximate right inverse. Recall that we have chosen a right inverse to the linear operator $\hat D_\alpha$:
\beqn
\hat Q_\alpha: {\mc E}_\alpha^{p,w}|_{{\bm v}_\alpha} \to {\bm E}_\alpha \oplus T_{{\bm v}_\alpha} {\mc B}_\alpha^{p,w}. 
\eeqn
Then define 
\beqn
\hat Q_{\alpha, \zeta}^{\rm app} = {\it Glue} \circ \hat Q_\alpha \circ {\it Cut}.
\eeqn
By Lemma \ref{lemma920}, there exist $C_\alpha>0$ and $\epsilon_\alpha>0$ such that 
\beq\label{eqn925}
\| \hat Q_{\alpha, \zeta}^{\rm app} \| \leq C_\alpha,\ \ \ \forall \zeta \in (0, \epsilon_\alpha).
\eeq

Next we show that the family of operators $\hat Q_{\alpha, \zeta}^{\rm app}$ are approximate right inverses to $\hat D_{\alpha, \zeta}$. 
\begin{lemma}\label{lemma921}
For $T$ sufficiently large, one has 
\beq\label{eqn926}
\Big\| \hat D_{\alpha, \zeta} \circ \hat Q_{\alpha, \zeta}^{\rm app} - {\rm Id} \Big\|\leq \frac{1}{2}.
\eeq
\end{lemma}

\begin{proof}
We need to estimate 
\beqn
\Big\| \hat D_{\alpha, \zeta} \big( \hat Q_{\alpha, \zeta}^{\rm app} ( {\bm \varsigma} ) \big) - {\bm \varsigma} \Big\|_{L_{\alpha, \zeta}^{p, w}}
\eeqn
for any ${\bm \varsigma} \in {\mc E}_{\alpha, \zeta}|_{{\bm v}_{\alpha, \zeta}^{\rm app}}$. Since the regularity of ${\bm \varsigma}$ is only $L^p$ and $\Sigma_{\alpha, \zeta}$ is the union of $\Sigma_{v, 2T}$ for all irreducible components $v$ of $\Sigma_\alpha$, it suffices to consider the case that ${\bm \varsigma}_{\zeta}$ is supported in $\Sigma_{v, 2T}$ for one component $v$. Without loss of generality, assume that $v$ has only one cylindrical end and denote by the cut-off function by $\chi_v^{\hbar T}$. Consider the parallel transport
\beqn
{\bm \varsigma}_\alpha = {\it Par}^{-1} {\bm \varsigma} \in L^{p, w}( \Sigma_{\alpha, v}, u_\alpha^* T^{\rm vert} Y \oplus {\rm ad} P \oplus {\rm ad} P). 
\eeqn

Then by definition, over $\Sigma_{v, 2T}$ where $\chi_v^{\hbar T} \equiv 1$, one has
\beq\label{eqn927}
\hat D_{\alpha, \zeta} \big( \hat Q_{\alpha, \zeta}^{\rm app} ( {\bm \varsigma}) \big) - {\bm \varsigma} = \hat D_{\alpha, \zeta} \big( {\it Par} \big( \hat Q_\alpha ({\bm \varsigma}_\alpha) \big) \big) - {\it Par}  \big( \hat D_\alpha\big( \hat Q_\alpha ({\bm \varsigma}_\alpha) \big) \big). 
\eeq
Over the complement of $\Sigma_{v, 2T}$ where ${\bm \varsigma} \equiv 0$, one has 
\beqn	
\begin{split}
\hat D_{\alpha, \zeta} \big( \hat Q_{\alpha, \zeta}^{\rm app} ( \varsigma) \big) - \varsigma = &\ \hat D_{\alpha, \zeta} \big( {\it Par} \big( \chi_v^{\hbar T} \big( \hat Q_\alpha ( \varsigma_\alpha) \big) \big) \big) - {\it Par} {\mb D}_\alpha {\mb Q}_\alpha {\bm \xi}_\alpha \\
= &\ \big[ \hat D_{\alpha ,\zeta},  \chi_v^{\hbar T} \big] {\it Par} \big( \hat Q_\alpha (\varsigma_\alpha) \big) + \chi_v^{\hbar T} \big( \hat D_{\alpha, \zeta} {\it Par} - {\it Par} \hat D_\alpha \big) \big( \hat Q_\alpha ( \varsigma_\alpha) \big). 
\end{split}
\eeqn
Hence the last line coincides with \eqref{eqn927}. Since $\hat D_{\alpha, \zeta}$ is a first order operator, by \eqref{eqn924},
\beq\label{eqn928}
\Big\| \big[ \hat D_{\alpha, \zeta}, \chi_v^{\hbar T} \big] {\it Par} \big( \hat Q_\alpha ( \varsigma_\alpha) \big) \Big\|_{L_{\alpha, \zeta}^{p,w}} \leq \sup \big| \nabla \chi_v^{\hbar T} \big| \big\| \hat Q_\alpha ( \varsigma_\alpha) \big\|_{L_{\alpha, \zeta}^{p, w}} \leq \frac{ 2 C_\alpha}{\hbar T} \| \varsigma \|_{L_{\alpha, \zeta}^{p, w}}.
\eeq
On the other hand, to estimate
\beqn
\Big\| \chi_v^{\hbar T} \big( \hat D_{\alpha, \zeta} {\it Par} - {\it Par}  \hat D_\alpha \big) \big( \hat Q_\alpha( \varsigma_\alpha) \big) \Big\| 
\eeqn
we need to estimate the variation of the linear maps. Indeed, using the same method as proving Lemma \ref{lemma919}, one can show that when $\zeta$ is small, one has
\beq\label{eqn929}
\Big\| \chi_v^{\hbar T} \big( \hat D_{\alpha, \zeta} {\it Par} - {\it Par}  \hat D_\alpha \big) \big( \hat Q_\alpha( \varsigma_\alpha) \big) \Big\| \leq C_\alpha \| {\bm \varsigma}\|_{L_{\alpha, \zeta}^{p,w}} \| \zeta \|
\eeq
for some abusively used $C_\alpha>0$. Combining \eqref{eqn928} and \eqref{eqn929}, one obtains \eqref{eqn926}.
\end{proof}

It follows from Lemma \ref{lemma921} that when $|\zeta|$ is sufficiently small, one has the following exact right inverse to $\hat D_{\alpha, \zeta}$:
\beqn
\hat Q_{\alpha, \zeta}:= \hat Q_{\alpha, \zeta}^{\rm app} \circ \big(  \hat D_{\alpha, \zeta} \circ \hat Q_{\alpha, \zeta}^{\rm app} \big)^{-1}. 
\eeqn
\eqref{eqn925} and Lemma \ref{lemma921} further imply that the norm of $\hat Q_{\alpha, \zeta}$ is uniformly bounded for all $\zeta$ by a constant $C_\alpha>0$. 

Now we can apply the implicit function theorem. Let us recall the precise statement (see \cite[Proposition A.3.4]{McDuff_Salamon_2004}).

\begin{lemma} {\rm (Implicit function theorem)} \cite[Proposition A.3.4]{McDuff_Salamon_2004} Let $X$, $Y$ be Banach spaces, $U \subset X$ be an open subset, and $F: U \to Y$ be a $C^1$ map. Let $x^* \in {\mc  U}$ be such that the differential $dF(x^*): X \to Y$ is surjective and has a bounded linear right inverse $Q: Y \to X$. Choose positive $r>0$ and $C>0$ such that $\| Q \|\leq \leq C$, $B_r(x^*, X) \subset U$, and 
\beqn
\| x - x^* \| \leq r \Longrightarrow \| dF (x) - D \| \leq \frac{1}{2C}.
\eeqn
Suppose that $x' \in X$ satisfies 
\beqn
\| F(x') \| < \frac{ r}{4C},\ \| x' - x^* \| < \frac{ r}{8}.
\eeqn
Then there exists a unique $x \in X$ such that 
\beqn
F (x) = 0,\ \ \ x- x' \in {\rm Image}(Q),\ \ \ \| x - x^* \| \leq r.
\eeqn
Moreover, 
\beqn
 \| x - x' \| \leq 2C \| F(x') \|.
\eeqn
\end{lemma}

Indeed, we identify $x^*$ with the central approximate solution $\hat {\bm v}_{\alpha, \zeta}^{\rm app}$, identity
\begin{align*}
&\ X = {\bm E}_\alpha \oplus T_{{\bm v}_{\alpha, \zeta}^{\rm app}} {\mc B}_{\alpha, \zeta}^{p,w},\ &\ Y = {\mc E}_{\alpha, \zeta}^{p,w}|_{{\bm v}_{\alpha, \zeta}^{\rm app}},
\end{align*}
Then one can apply the implicit function theorem. More precisely, for each approximate solution $\hat{\bm v}_{\alpha, {\bm \xi}}^{\rm app}$, we can write it uniquely as
\beqn
\hat{\bm v}_{\alpha, {\bm \xi}}^{\rm app} = \exp_{\hat{\bm v}_{\alpha, \zeta}^{\rm app}} \hat{\bm x}_{\alpha, {\bm \xi}}^{\rm app},\ {\rm where}\ \hat{\bm x}_{\alpha, {\bm \xi}}^{\rm app} \in {\bm E}_\alpha \oplus T_{{\bm v}_{\alpha, \zeta}^{\rm app}} {\mc B}_{\alpha, \zeta}^{p, w}.   
\eeqn
We summarize the result as the following proposition.
\begin{prop}
There exists $\epsilon_\alpha>0$, $\delta_\alpha>0$, $C_\alpha>0$ satisfying the following properties. For each ${\bm \xi} \in {\mc V}_{\alpha, {\rm map}} \times {\mc V}_{\alpha, {\rm def}} \times {\mc V}_{\alpha, {\rm res}}$ with $\| {\bm \xi}\|\leq \epsilon_\alpha$, there exists a unique 
\beqn
\hat{\bm x}_{\alpha, {\bm \xi}}^{\rm cor} \in {\bm E}_\alpha \oplus T_{{\bm v}_{\alpha, \zeta}^{\rm app}} {\mc B}_{\alpha, \zeta}^{p, w}
\eeqn
satisfying
\beqn
\hat{\mc F}_{\alpha, {\bm \xi}} \Big( \exp_{\hat{\bm v}_{\alpha, \zeta}^{\rm app}} \big( \hat{\bm x}_{\alpha, {\bm \xi}}^{\rm app} + \hat{\bm x}_{\alpha, {\bm \xi}}^{\rm cor} \big) \Big) = 0,\ \ \  \hat{\bm x}_{\alpha, {\bm \xi}}^{\rm cor} \in {\rm Image}(\hat Q_{\alpha, \zeta}),\ \ \ \Big\| \hat{\bm x}_{\alpha, {\bm \xi}}^{\rm app} \Big\| \leq \delta_\alpha.
\eeqn
Moreover, there holds
\beqn
\Big\| \hat{\bm x}_{\alpha, {\bm \xi}}^{\rm cor} \Big\| \leq C_\alpha \Big\| \hat{\mc F}_{\alpha, {\bm \xi}} ( \hat{\bm v}_{\alpha, {\bm \xi}}^{\rm app} ) \Big\|.
\eeqn
\end{prop}

Now we denote 
\beqn
\hat{\bm v}_{\alpha, {\bm \xi}}:= \hat{\bm v}_{\alpha, \xi, \eta, \zeta}: = \exp_{\hat{\bm v}_{\alpha, \zeta}^{\rm app}} \big( \hat{\bm x}_{\alpha, {\bm \xi}} \big) :=  \exp_{\hat{\bm v}_{\alpha, \zeta}^{\rm app}} \big( \hat{\bm x}_{\alpha, {\bm \xi}}^{\rm app} + \hat{\bm x}_{\alpha, {\bm \xi}}^{\rm cor} \big)
\eeqn
and call it the {\it exact solution}. Then we are ready to prove Proposition \ref{prop916}. Indeed, Item (a), (b), (c) of Proposition \ref{prop916} all follow from the construction. The $\Gammait_\alpha$-equivariance of Item (d) also follows from the construction. Hence we only need to prove that the map \eqref{eqn915} is a homeomorphism onto its image. Indeed, since the domain is locally compact and the target is Hausdorff, one only needs to prove that it is injective and its image contains an open neighborhood of ${\bm p}_\alpha^+$ in ${\mc M}_\alpha^+$. 

\vspace{0.2cm}

\noindent {\bf --Injectivity} Given two ${\bm \xi}_1 = (\xi_1, \eta_1, \zeta_1), {\bm \xi}_2 = (\xi_2, \eta_2, \zeta_2)$. Suppose the corresponding exact solutions are isomorphic. Then by definition (see Definition \ref{defn96}), $\eta_1 = \eta_2$ and $\zeta_1 = \zeta_2$. Hence $\hat {\bm v}_{\alpha, {\bm \xi}_1}$ and $\hat {\bm v}_{\alpha, {\bm \xi}_2}$ are in the same Banach manifold ${\bm E}_\alpha \times {\mc V}_{\alpha, {\rm def}} \times {\mc B}_{\alpha, \zeta_1}^{p, w}$. Then $\hat{\bm v}_{\alpha, {\bm \xi}_1} = \hat{\bm v}_{\alpha, {\bm \xi}_2}$ follows from the implicit function theorem. 

\vspace{0.2cm}

\noindent {\bf --Surjectivity} We need to prove the following fact: for any sequence of points in ${\mc M}_\alpha^+$ represented by $\alpha$-thickened solutions (without the requirement at the markings ${\bf y}_\alpha$)
\beqn
\Big( {\mc C}_n, {\bm v}_n, {\bf y}_n, \wh\phi_n, e_n \Big)
\eeqn
that converge to ${\bm p}_\alpha^+$, for sufficiently large $n$, $\big( {\mc C}_n, {\bm v}_n, {\bf y}_n, \wh\phi_n, e_n)$ is isomorphic to a member of the family of exact solutions of \eqref{eqn914}. Indeed, by Lemma \ref{lemma911}, given any $\epsilon>0$, for $n$ sufficiently large, $({\mc C}_n, {\bm v}_n, {\bf y}_n, \wh\phi_n, e_n)$ is $\epsilon$-close to $\alpha$. Then we can identify $({\mc C}_n, {\bf y}_n)$ with the fibre $({\mc C}_{\alpha, \phi_n}, {\bf y}_{\phi_n})$ and regard ${\bm v}_n$ as a gauged map over $P_{\alpha, \phi_n} \to {\mc C}_{\alpha, \phi_n}$. Suppose ${\mc C}_{\alpha, \phi_n}$ corresponds to deformation parameter $\eta_n$ and gluing parameter $\zeta_n$. Then the $\epsilon$-closedness implies that
\beqn
{\rm dist}({\bm v}_n, {\bm v}_{\alpha, \zeta_n}^{\rm app}) \leq \epsilon.
\eeqn
Then ${\bm v}_n$ belonging to the family \eqref{eqn914} is a fact that follows from the implicit function theorem. This finishes the proof of Proposition \ref{prop916}.

\subsection{Last modification}

The construction of this section provides for each thickening datum $\alpha$ and a sufficiently small number $\epsilon_\alpha>0$ a virtual orbifold chart $C_\alpha$ of the moduli space $\ov{\mc M}{}_{\sf \Gamma}$ for any stable decorated dual graph. In particular, one can construct a chart of the top stratum $\ov{\mc M}{}_{g, n}^r(V, G, W, \mu; \ubar B)$. Further, one can shrink the set ${\mc V}_{\alpha, {\rm res}}^\epsilon$ of gluing parameters to obtain subcharts. Indeed, ${\mc V}_{\alpha,{\rm res}}$ has canonical coordinates corresponding to the nodes. Choose a vector $\vec\epsilon = (\epsilon_1, \cdots, \epsilon_m)$ with $0< \epsilon_i < \epsilon_\alpha$ corresponding to how much we can resolve the $i$-th node. We require that ${\mc V}_{\alpha, {\rm res}}^{\vec \epsilon}$ is $\Gammait_\alpha$-invariant. From now on, a thickening datum $\alpha$ also contains a small positive number $\epsilon_\alpha>0$ and such a vector $\vec\epsilon$. Each such thickening datum provides a chart $C_\alpha$ by restricting the previous construction.

\section{Constructing the Virtual Cycle. III. The Atlas}\label{section10}

In the previous section we have shown that we can construct for each stable decorated dual graph $\sf\Gamma$ and each point $p \in \ov{\mc M}_{\sf \Gamma}$, one can construct a local chart
\beqn
C_p = (U_p, E_p, S_p, \psi_p, F_p)
\eeqn
(See Corollary \ref{cor917}). The aim of this section is to construct a good coordinate system out of these charts on the moduli space, which allows one to define the virtual fundamental cycle and the correlation functions. As the first step, we prove the following proposition. 

\begin{prop}\label{prop101}
Let $\sf\Gamma$ be a stable decorated dual graph. Then there exist:
\begin{enumerate}

\item A finite collection of topological virtual orbifold charts on $\ov{\mc M}_{\sf\Gamma}$
\beqn
C_{p_i}^\bullet = (U_{p_i}^\bullet , E_{p_i}^\bullet, S_{p_i}^\bullet, \psi_{p_i}^\bullet, F_{p_i}^\bullet ),\ i = 1, \ldots, N.
\eeqn

\item Another collection of topological virtual orbifold charts
\beqn
C_I^\bullet = (U_I^\bullet, E_I^\bullet, S_I^\bullet, \psi_I^\bullet, F_I^\bullet ),\ I \in {\mc I} = 2^{\{1, \ldots, N\}} \setminus \{\emptyset\}.
\eeqn

\item Define the partial order on ${\mc I}$ by inclusion. Then for each pair $I \preq J$, a weak coordinate change 
\beqn
T_{JI}^\bullet: C_I^\bullet \to C_J^\bullet.
\eeqn
\end{enumerate}
They satisfy the following conditions. 
\begin{enumerate}

\item Each $C_{p_i}^\bullet$ is constructed by the gluing construction in the last section. In particular, for each $i$, there is a thickening datum $\alpha_i$ which contains an obstruction space ${\bm E}_{\alpha_i}$, and a $\Gammait_i$-invariant open neighborhood $\tilde U_{p_i}^\bullet$ of the $\alpha_i$-thickened moduli space ${\mc M}_{\alpha_i}$ which is a topological manifold such that 
\begin{align*}
&\ U_{p_i}^\bullet = \tilde U_{p_i}^\bullet/\Gammait_i,\ &\ E_{p_i}^\bullet = \big( \tilde U_{p_i}^\bullet \times {\bm E}_{p_i} \big)/ \Gammait_i.
\end{align*}
Moreover, $p_i$ is contained in $F_{p_i}^\bullet$. 

\item For each $I \in {\mc I}$ which is identified with the set of thickening data $\{ \alpha_i\ |\ i \in I\}$, there is a $\Gammait_I$-invariant open subset $\tilde U_I^\bullet \subset {\mc M}_I$ of the $I$-thickened moduli space which is a topological manifold, such that 
\begin{align*}
&\ U_I^\bullet = \tilde U_I^\bullet / \Gammait_I,\ &\ E_I^\bullet = \big( \tilde U_{p_i}^\bullet \times {\bm E}_I \big)/ \Gammait_I.
\end{align*}
Moreover, $F_I^\bullet = \bigcap_{i \in I} F_{p_i}^\bullet$.

\item For $I = \{i\}$, $C_{\{i\}}^\bullet = C_{p_i}^\bullet$. 

\item The charts $C_I^\bullet$ and the coordinate changes $T_{JI}^\bullet$ satisfy the {\bf (Covering Condition)} and the {\bf (Cocycle Condition)} of Definition \ref{defn_atlas}.

\item All charts and all coordinate changes are oriented. 
\end{enumerate}
\end{prop}

The construction of the objects in Proposition \ref{prop101} is done in an inductive way. From now on we abbreviate $\ov{\mc M}:= \ov{\mc M}_{\sf \Gamma}$. Moreover, we know that 
\beqn
\ov{\mc M} = \bigsqcup_{\sf \Gamma' \preq \sf \Gamma} {\mc M}_{\sf \Gamma'}.
\eeqn
Then we list all strata $\sf \Gamma'$ indexing the above disjoint union as ${\sf\Gamma}_1, \ldots, {\sf \Gamma}_a$ 
such that
\beqn
{\sf\Gamma}_k \preq {\sf\Gamma}_l \Longrightarrow k \leq l.
\eeqn

\subsection{The inductive construction of charts}

We first state our induction hypothesis. 

\bigskip 

\noindent \tt{INDUCTION HYPOTHESIS}. For $k \in \{ 1, \ldots, a-1 \}$ we have constructed the following objects.
\begin{enumerate}

\item  A collection of virtual orbifold charts 
\beqn
C_{p_i}^k = (U_{p_i}^k, E_{p_i}^k, S_{p_i}^k, \psi_{p_i}^k, F_{p_i}^k),\ i = 1, \ldots, n_k.
\eeqn

\item  A map 
\beq\label{eqn101}
\rho_k: \{ p_1, \ldots, p_{n_k}\} \to \{ {\sf\Gamma}_1, \ldots, {\sf\Gamma}_k \}.
\eeq

\item Define ${\mc I}_k = 2^{\{1, \ldots, n_k\}} \setminus \{\emptyset\}$. For any $I \in {\mc I}_k$, a chart 
\beqn
C_I^k = (U_I^k, E_I^k, S_I^k, \psi_I^k, F_I^k).
\eeqn

\end{enumerate}
They satisfy the following condition. 

\begin{enumerate}

\item These charts satisfy Item (a), (b) and (c) of Proposition \ref{prop101} if we replace $C_{p_i}^\bullet$ by $C_{p_i}^k$ and $C_I^\bullet$ by $C_I^k$. In particular
\beq\label{eqn102}
F_I^k = \bigcap_{i \in I} F_{p_i}^k. 
\eeq

\item For each $l \leq k$, one has 
\beq\label{eqn103}
\ov{\mc M}_l:= \bigcup_{s \leq l} {\mc M}_{{\sf\Gamma}_s} \subset \bigcup_{s \leq l} \bigcup_{\rho_k(p_i) = {\sf \Gamma}_s} F_{p_i}^k.
\eeq

\item For all $l \leq k$ and $I \subset \{1, \ldots, n_l\}$, we have
\beq\label{eqn104}
F_I^k \neq \emptyset \Longleftrightarrow F_I^k \cap \ov{\mc M}_l \neq \emptyset
\eeq

\end{enumerate}
\hfill \tt{END OF THE INDUCTION HYPOTHESIS.} 
\bigskip

Suppose the \tt{INDUCTION HYPOTHESIS} holds for $k$. We aim at extending the objects stated in the induction hypothesis to $k+1$ and modify the charts and coordinate changes which have been already constructed. In the following argument, the base case $k=1$ can also be deduced. Now we start the induction. Notice that $\ov{\mc M}_k$ is compact. Denote
\beqn
Y_{k+1}:= \ov{\mc M}_{k+1} \setminus \bigcup_{1 \leq i \leq n_k} F_{p_i}^k
\eeqn
which is also compact. Then for each $p \in Y_{k+1}$, we choose a thickening datum $\alpha_p$ at $p$ which provides a chart $C_p^\bbox = (U_p^\bbox, E_p^\bbox, S_p^\bbox, \psi_p^\bbox, F_p^\bbox )$ of $\ov{\mc M}$ around $p$. Then because $Y_{k+1}$ is compact, we can then choose a finite collection of such charts around points $p_{n_k+1}, \ldots, p_{n_{k+1}}$ with footprints $F_{p_{n_k+1}}^\bbox, \ldots, F_{p_{n_{k+1}}}^\bbox$ such that 
\beq\label{eqn105}
Y_{k+1} \subset \bigcup_{i=n_k + 1}^{n_{k+1}} F_{p_i}^\bbox.
\eeq
Define 
\beqn
{\mc I}_{k+1} = 2^{\{1, \ldots, n_{k+1}\}} \setminus \{ \emptyset\}.
\eeqn
Extend the map $\rho_k$ of \eqref{eqn101} to a map 
\begin{align}\label{eqn106}
&\ \rho_{k+1}: \{1, \ldots, n_{k+1}\} \to \{ {\sf \Gamma}_1, \ldots, {\sf\Gamma}_{k+1}\},\ &\ \rho_{k+1}(i) = \left\{ \begin{array}{c} \rho_k(i),\ i \leq n_k; \\
                                          {\sf\Gamma}_{k+1},\ n_k < i \leq n_{k+1} \end{array} \right.
\end{align}

By \eqref{eqn104} of \tt{INDUCTION HYPOTHESIS}, by shrinking the range of the gluing parameters, one can choose shrinkings $F_{p_i}^k \sqsubset F_{p_i}^\bbox$ for all $i \in \{n_k+1, \ldots, n_{k+1}\}$ such that \eqref{eqn105} still holds with $F_{p_i}^\bbox$ replaced by $F_{p_i}^k$ and such that for all $I \in {\mc I}_{k+1}$, we have
\beqn
F_I^k:= \bigcap_{i \in I} F_{p_i}^k \neq \emptyset\Longleftrightarrow F_I^k \cap \ov{\mc M}_{k+1} \neq \emptyset. 
\eeqn
 
Now we are going to construct the charts
\beqn
C_I^{k+1} = (U_I^{k+1}, E_I^{k+1}, S_I^{k+1}, \psi_I^{k+1}, F_I^{k+1}),\ \forall I \in {\mc I}_{k+1}.
\eeqn
If $I \in {\mc I}_k$, then $C_I^{k+1}$ will be obtained later by shrinking $C_I^k$. We assume that $I \in {\mc I}_{k+1} \setminus {\mc I}_k$. If $F_I^k = \emptyset$, define $C_I^{k+1}$ to be the empty chart. Then assume $F_I^k \neq \emptyset$. Fix such an $I$. 

By abusing notations, we regard $I$ as a set of thickening data. Recall that one has define the moduli space of $I$-thickened solutions, denoted by ${\mc M}_I$. Every point of ${\mc M}_I$ is represented by a tuple 
\beqn
\big( {\mc C}, {\bm v}, \{ {\bf y}_\alpha\}, \{ \wh\phi_\alpha\}, \{ e_\alpha \} \big).
\eeqn
It has a Hausdorff and second countable topology, and has a continuous $\Gammait_I$-action. Moreover, there is a $\Gammait_I$-equivariant map $\tilde S_I: {\mc M}_I \to {\bm E}_I$, and a natural map 
\beqn
\tilde \psi_I: \tilde  S_I^{-1}(0) \to \ov{\mc M}. 
\eeqn

\begin{prop}
For $I \in {\mc I}_{k+1} \setminus {\mc I}_k$ with $F_I^k \neq \emptyset$, there is a $\Gammait_I$-invariant open neighborhood $\tilde U_I^\bbox \subset {\mc M}_I$ of $\tilde \psi_I^{-1}(\ov{F_I^k} \cap \ov{\mc M}_{k+1} ) \subset {\mc M}_I $ which is a topological manifold.
\end{prop}

\begin{proof}
By our assumption, there exists $i \in \{ n_k + 1, \ldots, n_{k+1}\} \cap I$. Hence $\ov{F_I^k} \cap \ov{\mc M}_{k+1} \subset {\mc M}_{{\sf \Gamma}_{k+1} }$. Choose any point $p \in \ov{F_I^k} \cap \ov{\mc M}_{k+1}$ represented by an $I$-thickened solution 
\beqn
\hat{\bm v}_I:= \big( {\mc C}, {\bm v}, \{ {\bf y}_\alpha\}, \{ \wh\phi_\alpha \}, \{ {\it 0}_\alpha \} \big).
\eeqn
By forgetting the data for $\alpha \neq \alpha_i$, we obtain an $\alpha_i$-thickened solution 
\beqn
\hat{\bm v}_{\alpha_i}:= \big( {\mc C}_{\alpha_i}, {\bm v}_{\alpha_i}, {\bf y}_{\alpha_i}, \wh\phi_{\alpha_i}, {\it 0}_{\alpha_i} \big).
\eeqn
Then we can identify the curve ${\mc C}_{\alpha_i}$ and the bundle $P_{\alpha_i} \to {\mc C}_{\alpha_i}$ as a fibre over the unfolding ${\mc U}_{\alpha_i} \to {\mc V}_{\alpha_i}$ contained in the thickening datum $\alpha_i$. Then $\hat{\bm v}_{\alpha_i}$ belongs to the family $\hat{\bm v}_{\alpha_i, \xi_i, \eta_i}$ \eqref{eqn95} for $\alpha = \alpha_i$, where $\xi_i \in {\mc V}_{\alpha_i, {\rm map}}$ and $\eta_i \in {\mc V}_{\alpha_i, {\rm def}}$. 

Remember that we used the right inverse to the linearization of the $\alpha_i$-thickened gauged Witten equation
\beqn
\hat Q_{\alpha_i}: {\mc E}_{\alpha_i} \to {\mc V}_{\alpha_i, {\rm def}} \oplus  {\bm E}_{\alpha_i} \oplus T_{{\bm v}_{\alpha_i}} {\mc B}_{\alpha_i}
\eeqn
chosen in \eqref{eqn94} for $\alpha = \alpha_i$. By the inclusion ${\bm E}_{\alpha_i} \subset {\bm E}_I$, we obtained a right inverse 
\beqn
\hat Q_{I, \alpha_i}: {\mc E}_{\alpha_i} \to {\mc V}_{\alpha_i, {\rm def}} \oplus {\bm E}_I \oplus T_{{\bm v}_{\alpha_i}} {\mc B}_{\alpha_i}.
\eeqn

Notice that the stratum ${\sf\Gamma}_{k+1}$ also specifies a stratum of the $I$-thickened moduli space with a point ${\bm p}_I$ represented by the $I$-thickened solution $\hat{\bm v}_I$. This combinatorial type provides a Banach manifold of gauged maps ${\mc B}_{{\sf\Gamma}_{k+1}}$ that contains $\hat{\bm v}_I$. Then using the right inverse $\hat Q_{I, \alpha_i}$ and the implicit function theorem, we can construct a family of $I$-thickened solutions of combinatorial type ${\sf\Gamma}_{k+1}$ which are close to $\hat{\bm v}_I$ in the Banach manifold ${\mc B}_{{\sf\Gamma}_{k+1}}$, parametrized by a topological manifold acted by $\Gammait_I$. 

Then by turning on the gluing parameters, a family of approximate solutions can be constructed as $I$-thickened objects. The same gluing construction provides a collection of $I$-thickened solutions parametrized by a topological manifold of the expected dimension. The gluing construction is the same as that of Section \ref{section9} because, here we just enlarged the space of obstructions from ${\bm E}_{\alpha_i}$ to ${\bm E}_I$. 
\end{proof}

Then take such a $\Gammait_I$-invariant neighborhood $\tilde U_I^\bbox \subset {\mc M}_I$ of $\tilde \psi_I^{-1}( \ov{F_I^k} \cap \ov{\mc M}_{k+1})$, we obtain a chart
\beqn
C_I^\bbox = ( U_I^\bbox, E_I^\bbox, S_I^\bbox, \psi_I^\bbox, F_I^\bbox )
\eeqn
Notice that we have the inclusion
\beqn
\ov{F_I^k} \cap \ov{\mc M}_{k+1} \subset F_I^\bbox.
\eeqn

Let us summarize the charts we have obtained. For all $i \in \{1, \ldots, n_{k+1}\}$, we have charts $C_{p_i}^k$ with footprints $F_{p_i}^k$. For all $I \in {\mc I}_k$, we have charts $C_I^k$ with footprints $F_I^k$. For all $I \in {\mc I}_{k+1} \setminus {\mc I}_k$, we have charts $C_I^\bbox$ with footprints $F_I^\bbox$. We would like to shrink these charts so that their footprints satisfy \eqref{eqn102} of \tt{INDUCTION HYPOTHESIS} for $k+1$. 

Indeed, we can shrink all $C_{p_i}^k$ to a chart $C_{p_i}^{k+1}$ with footprints $F_{p_i}^{k+1}$ which satisfy the following conditions. Define
\beqn
F_I^{k+1} = \bigcap_{i \in I} F_{p_i}^{k+1}. 
\eeqn 
\begin{enumerate}

\item \eqref{eqn103} still holds. More precisely, for all $l \leq k+1$, we have
\beqn
\ov{\mc M}_l \subset \bigcup_{s \leq l} \bigcup_{\rho_{k+1} (p_i) = {\sf\Gamma}_s}  F_{p_i}^{k+1}
\eeqn

\item \eqref{eqn104} still holds. More precisely, for all $l \leq k+1$ and $I \in {\mc I}_{k+1}$, 
\beqn
F_I^{k+1} \neq \emptyset \Longleftrightarrow F_I^k \cap \ov{\mc M}_l \neq \emptyset.
\eeqn

\item For all $I \in {\mc I}_{k+1} \setminus {\mc I}_k$, we have
\beqn
F_I^{k+1} \subset F_I^\bbox.
\eeqn
\end{enumerate}

Then shrink all $C_I^k$ for $I \in {\mc I}_k$ (resp. $C_I^\bbox$ for $I \in {\mc I}_{k+1} \setminus {\mc I}_k$) to $C_I^{k+1}$ so that the shrunk footprint is  $F_I^{k+1}$. Together with the map $\rho_{k+1}$ in \eqref{eqn106} these charts satisfy conditions for charts listed in \tt{INDUCTION HYPOTHESIS} for $k+1$. Therefore the induction can be carried on. After the last step of the induction, denote the charts we constructed by 
\beqn
C_I^\bt = (U_I^\bt, E_I^\bt, S_I^\bt, \psi_I^\bt, F_I^\bt).
\eeqn
They satisfy Item (a), (b) and (c) of Proposition \ref{prop101} if we replace $C_I^\bullet$ by $C_I^\bt$. Moreover, the footprints of $C_I^\bt$ cover $\ov{\mc M}$. 

\subsection{Coordinate changes}

Now we start to define the coordinate changes. Since the moduli space $\ov{\mc M}$ is compact and Hausdorff, one can find precompact open subsets $F_{p_i}^\bullet \sqsubset F_{p_i}^\bt$ such that the union of $F_{p_i}^\bullet$ still cover $\ov{\mc M}$. Define 
\beq\label{eqn107}
F_I^\bullet:= \bigcap_{\alpha_i \in I} F_{p_i}^\bullet.
\eeq
Then 
\beqn
\ov{F_I^\bullet } \subset \bigcap_{\alpha_i} \ov{F_{p_i}^\bullet } \subset F_I^\bt.
\eeqn

Consider $I \preq J \in {\mc I}$ which corresponds to two sets of thickening data, which are still denoted by $I$ and $J$. Denote
\begin{align*}
&\ {\bm E}_{JI} = \bigoplus_{\alpha \in J \setminus I} {\bm E}_\alpha,\ & \Gammait_{JI} = \prod_{\alpha \in J \setminus I} \Gammait_\alpha,
\end{align*}
where ${\bm E}_\alpha$ are the vector spaces of obstructions which are acted by $\Gammait_\alpha$. Define 
\beqn
\tilde S_{JI}: {\mc M}_J \to {\bm E}_{JI}
\eeqn
to be the natural map. Then there is a natural map 
\beqn
\tilde \psi_{JI}: \tilde S_{JI}^{-1}(0) \to {\mc M}_I
\eeqn
by forgetting ${\bf y}_\alpha$, $\wh\phi_\alpha$, $e_\alpha$ for all $\alpha \in J \setminus I$. This is clearly equivariant with respect to the homomorphism $\Gammait_J \to \Gammait_I$ which annihilates $\Gammait_{JI}$, hence descends to a map 
\beqn
\psi_{JI}: \tilde S_{JI}^{-1}(0)/ \Gammait_{JI} \to {\mc M}_I.
\eeqn

\begin{lemma}\label{lemma103}
There is a $\Gammait_J$-invariant open neighborhood $\tilde N_{JI}^\bt \subset {\mc M}_J$ of $\tilde  S_J^{-1}(0)$ such that the map 
\beq\label{eqn108}
\psi_{JI}: ( \tilde N_{JI}^\bt \cap \tilde S_{JI}^{-1}(0))/ \Gammait_{JI} \to {\mc M}_I
\eeq
is a homeomorphism onto a $\Gammait_I$-invariant open neighborhood $\tilde U_{JI}^\bt$ of $ \tilde \psi_I^{-1}( \ov{F_J^\bullet} ) \subset {\mc M}_I$. 
\end{lemma}

\begin{proof}
First we show the surjectivity. For any ${\bm p}_I \in \tilde \psi_I^{-1}( \ov{F_J^\bullet } ) \subset \tilde \psi_I^{-1}(\ov{F_I^\bullet })$, it is represented by an $I$-thickened solution 
\beqn
\big( {\mc C}, {\bm v}, \{ {\bf y}_\alpha\}_{\alpha \in I}, \{ \wh\phi_\alpha \}_{\alpha \in I}, \{ {\it 0}_\alpha\}_{\alpha \in I} \big).
\eeqn
By our construction, for any $\beta \in J \setminus I$, $({\mc C}, {\bm v})$ is $\epsilon_\beta$-close to $\beta$ (see Definition \ref{defn910}). Hence by Lemma \ref{lemma913}, there exist ${\bf y}_\beta$ and $\wh\phi_\beta$ such that 
\beqn
\big( {\mc C}, {\bm v}, {\bf y}_\beta, \wh\phi_\beta, {\it 0}_\beta  \big)
\eeqn
is a $\beta$-thickened solution. Moreover, by essentially the same method as in the proof of Lemma \ref{lemma913}, for any point ${\bm p}_I' \in {\mc M}_I$ that is sufficiently close to ${\bm p}_I$, for any representative that contains a gauged map $({\mc C}', {\bm v}')$, it is $\epsilon_\beta$-close to $\beta$ and can be completed to a $\beta$-thickened solution. Apply this for all $\beta \in J \setminus I$, it means that any ${\bm p}_I' \in {\mc M}_I$ sufficiently close to ${\bm p}_I$ is in the image of the map \eqref{eqn108}. Since $\tilde \psi_I^{-1}( \ov{F_J^\bullet })$ is compact, one can choose an open neighborhood $\tilde U_{JI}^\bt \subset {\mc M}_I$ of $\pi_I^{-1}( \ov{F_J})$ which is contained in the image of \eqref{eqn108}. Then we could find a $\Gammait_J$-invariant open subset $\tilde N_{JI}^\bt \subset {\mc M}_J$ such that 
\beq\label{eqn109}
\tilde N_{JI}^\bt \cap \tilde S_{JI}^{-1}(0) = \tilde V_{JI}^\bt:= \tilde \psi_{JI}^{-1}( \tilde U_{JI}^\bt ).
\eeq

On the other hand, Lemma \ref{lemma914} says that the map 
\beqn
\pi_{JI}:  \tilde V_{JI}^\bt / \Gammait_{JI} \to \tilde U_{JI}^\bt
\eeqn 
is bijective. Its continuity is obvious. Hence using the fact that a continuous bijection from a compact space to a Hausdorff space is a homeomorphism, and using the fact that the thickened moduli spaces are locally compact and Hausdorff, after a precompact shrinking of $\tilde U_{JI}^\bt$ which still contains $\tilde \psi_I^{-1}( \ov{F_J^\bullet})$, one proves this lemma. 
\end{proof}

Define 
\beqn
U_{JI}^\bt = \tilde U_{JI}^\bt / \Gammait_I
\eeqn
which is an open suborbifold of $U_I^\bt$. Lemma \ref{lemma103} implies that the inclusion $\tilde V_{JI}^\bt \hookrightarrow {\mc M}_J$ induces a map between orbifolds 
\beqn
\phi_{JI}^\bt: U_{JI}^\bt \to U_J^\bt. 
\eeqn
The natural inclusion ${\bm E}_I \to {\bm E}_J$ induces a bundle map 
\beqn
\wh\phi_{JI}^\bt: E_I^\bt |_{U_{JI}^\bt } \to E_J^\bt
\eeqn
which covers $\phi_{JI}^\bt$. 

\begin{prop}
For every $I$ there exists a shrinking $C_I^\bullet$ of $C_I^\bt$ and an open subset
\beqn
U_{JI}^\bullet \subset U_{JI}^\bt
\eeqn
satisfying the following conditions. 
\begin{enumerate}
\item The footprint of $C_I^\bullet$ is $F_I^\bullet$ (defined previously in \eqref{eqn107}).

\item $T_{JI}^\bullet:= (U_{JI}^\bullet, \phi_{JI}^\bullet, \wh\phi_{JI}^\bullet ):= (U_{JI}^\bullet, \phi_{JI}^\bt|_{U_{JI}^\bullet}, \wh\phi_{JI}^\bt|_{U_{JI}^\bullet })$ is a coordinate change from $C_I^\bullet$ to $C_J^\bullet $. 
\end{enumerate}
\end{prop}

\begin{proof}
According to the definition of coordinate changes (Definition \ref{defn612}), we first show that the map 
\beqn
\phi_{JI}^\bt: U_{JI}^\bt \to U_J^\bt
\eeqn
is a topological embedding of orbifolds. Indeed, one only needs to show that the map $\tilde  S_{JI}: {\mc M}_J \to {\bm E}_{JI}$ is transverse over $\tilde V_{JI}^\bt$. Indeed, choose $i \in I$. For every ${\bm p}_J \in \tilde V_{JI}^\bt$ which descends to an isomorphism of $I$-thickened solutions ${\bm p}_I = \pi_{JI}({\bm p}_J)$, for any representative 
\beqn
\big( {\mc C}, {\bm v}, \{{\bf y}_\alpha\}, \{ \wh\phi_\alpha\}, \{e_\alpha\} \big)
\eeqn
of ${\bm p}_I$, since ${\bm p}_I$ is close to $\tilde \psi_I^{-1}(\ov{F_J^\bullet}) \subset \tilde \psi_I^{-1}(\ov{F_I^\bullet})$, the linearization of the $I$-thickened equation 
\beqn
\hat D_I: T_{{\bm v}} {\mc B} \oplus {\bm E}_I \to {\mc E}_{\bm v}
\eeqn
is surjective. Here ${\mc B}$ is the corresponding Banach manifold and ${\mc E} \to {\mc B}$ is the Banach vector bundle. Therefore we can construct a tubular neighborhood as follows. For any ${\bm p}_J$ near $\tilde V_{JI}^\bt$ which is represented by a $J$-thickened solution 
\beqn
\big( {\mc C}, {\bm v}, \{{\bf y}_\alpha\}, \{ \wh\phi_\alpha\}, \{e_\alpha\} \big),
\eeqn
the $J$-thickened solution implies that 
\beqn
\hat {\mc F}_I({\bm v}, {\bm e}_I) = {\it error}.
\eeqn
Here ${\bm e} = \{ e_\alpha\}_{\alpha \in J}$ is decomposed as $({\bm e}_I, {\bm e}_{JI})$ where ${\bm e}_I \in {\bm E}_I$ and ${\bm e}_{JI} \in {\bm E}_{JI}$ and the ${\it error}$ term can be controlled by the size of ${\bm e}_{JI}$. When the neighborhood of $\tilde V_{JI}^\bt$ is sufficiently small, the error term is sufficiently small. Then by the implicit function theorem, there is a unique pair $({\bm v}', {\bm e}_I')$ lying in the same Banach manifold such that 
\begin{align*}
&\ \hat{\mc F}_I({\bm v}', {\bm e}_I') = 0,\ &\ ({\bm v}', {\bm e}_I') - ({\bm v}, {\bm e}_I) \in {\rm Image}(\hat Q_I). 
\end{align*}
This constructs a $\Gammait_J$-invariant neighborhood $\tilde N_{JI}^\bt$ of $\tilde V_{JI}^\bt$, a $\Gammait_J$-equivariant projection $\tilde \nu_{JI}: \tilde N_{JI}^\bt \to \tilde V_{JI}^\bt$. The implicit function theorem also implies that $\tilde S_{JI}$ induces an equivalence of microbundles. Hence $\tilde S_{JI}$ is transverse along $\tilde V_{JI}^\bt$ and $\phi_{JI}^\bt$ is an embedding. 

It then follows that $\wh\phi_{JI}^\bt$ is a bundle embedding covering $\phi_{JI}^\bt$. It is obvious that the pair $(\phi_{JI}^\bt, \wh\phi_{JI}^\bt)$ satisfies the {\bf (Tangent Bundle Condition)} of Definition \ref{defn611}. It remains to shrink the charts and coordinate changes so that Item (a) and Item (b) of Definition \ref{defn612} are satisfied. We only show it for a pair $I \preq J$, from which one can obtain a way of shrinking all the charts and coordinate changes so that all $T_{JI}^\bullet$ becomes a coordinate changes. First, one can find open subsets $U_I^\bullet \sqsubset U_I^\bt$ such that 
\begin{align*}
&\ \psi_I^\bt \big( \ov{U_I^\bullet } \cap (S_I^\bt)^{-1}(0) \big) = \ov{F_I^\bullet},\ &\ \psi_I^\bt \big( U_I^\bullet \cap (S_I^\bt)^{-1}(0) \big) = F_I^\bullet.
\end{align*}
Moreover, we may take $U_I^\bullet$ such that for all $J \in {\mc I}$, 
\beq\label{eqn1010}
\ov{U_J^\bullet} \subset  \big( \bigcap_{I \preq J} \tilde N_{JI}^\bt \big)/ \Gammait_J
\eeq
where $\tilde N_{JI}^\bt$ is the one we chose by Lemma \ref{lemma103}. Then take 
\beqn
U_{JI}^\bullet = U_I^\bullet \cap ( \phi_{JI}^\bt )^{-1}(U_J^\bullet ) \subset U_{JI}^\bt.
\eeqn
That is, $U_{JI}^\bullet$ is the ``induced'' domain from the shrinkings $C_I^\bullet$ and $C_J^\bullet$. Define $T_{JI}^\bullet = (U_{JI}^\bullet, \phi_{JI}^\bullet, \wh\phi_{JI}^\bullet)$ where $(\phi_{JI}^\bullet, \wh\phi_{JI}^\bullet)$ is the restriction of $(\phi_{JI}^\bt, \wh\phi_{JI}^\bt)$ onto $U_{JI}^\bullet$. 

$T_{JI}^\bullet$ satisfies Item (a) of Definition \ref{defn612} automatically, because the footprint of $U_I^\bullet$, $U_J^\bullet$ and $U_{JI}^\bullet$ are precisely $F_I^\bullet$, $F_J^\bullet$ and $F_J^\bullet$. To verify Item (b), suppose $x_n$ is a sequence of points in $U_{JI}^\bullet$ which converges to $x_\infty \in U_I^\bullet$ and $y_n = \phi_{JI}^\bullet (x_n)$ converges to $y_\infty \in U_J^\bullet$. Notice that each $\tilde \psi_J^{-1}(y_n)$ is a sequence of $\Gammait_J$-orbits of $J$-thickened solutions with $\tilde S_{JI} (\pi_J^{-1}(y_n)) = 0$. Then $\tilde S_{JI} ( \pi_J^{-1}(y_\infty)) = 0$ and by \eqref{eqn109} and \eqref{eqn1010}, we have 
\beqn
\tilde \psi_J^{-1}(y_\infty) \subset \tilde N_{JI}^\bt \cap \tilde S_{JI}^{-1}(0) =  \tilde V_{JI}^\bt \Longrightarrow y_\infty \in \phi_{JI}^\bt ( U_{JI}^\bt).
\eeqn
Since $\phi_{JI}^\bt$ is injective, continuous, and $x_n \to x_\infty$, we must have $x^\infty \in U_{JI}^\bullet$ and $\phi_{JI}^\bullet (x^\infty) = \phi_{JI}^\bt ( x^\infty) = y^\infty$. This establishes Item (b) and finishes the proof.
\end{proof}

Therefore, we have constructed all objects needed for establishing Proposition \ref{prop101}. Two more items need to be verified. First, the {\bf (Cocycle Condition)} for all the coordinate changes $T_{JI}^\bullet$ follows immediately from the construction. Second, all the charts and coordinate changes have canonical orientations, as the linearized operator is of the type of a real Cauchy--Riemann operator over a Riemann surface with cylindrical ends without boundary, and the asymptotic constraints at infinities of the cylindrical ends are given by complex submanifolds or orbifolds. Therefore Proposition \ref{prop101} is proved.

\subsection{Constructing a good coordinate system}

To obtain a good coordinate system from the objects constructed by Proposition \ref{prop101}, one first needs to make them satisfy the {\bf (Overlapping Condition)} of \ref{defn_atlas}). This can be done by shrinking the charts. Our method is modified from the one used in the proof of \cite[Lemma 5.3.1]{MW_3}.

Indeed, for all $i \in {\mc I}$, choose precompact shrinkings $F_{p_i}^\circ \sqsubset F_{p_i}^\bullet$ such that there holds
\beqn
\ov{\mc M} = \bigcup_{i = 1}^N F_{p_i}^\circ.
\eeqn
We also choose an ordering of ${\mc I}$ as 
\beqn
I_1, \ldots, I_M
\eeqn
such that 
\beqn
I_k \preq I_l \Longrightarrow k \leq l.
\eeqn
Choose intermediate precompact shrinkings between $F_{p_i}^\circ$ and $F_{p_i}^\bullet$ as 
\beq\label{eqn1011}
F_{p_i}^\circ =: G_{p_i}^1 \sqsubset F_{p_i}^1 \sqsubset \cdots \sqsubset G_{p_i}^M \sqsubset F_{p_i}^M:= F_{p_i}^\bullet. 
\eeq
Then for all $I_k \in {\mc I}$, define 
\beqn
F_{I_k}^\square:= \Big( \bigcap_{i \in I_k} F_{p_i}^k \Big) \setminus \Big( \bigcup_{i \notin I_k} \ov{G_{p_i}^k} \Big).
\eeqn
These are open subsets of the moduli space $\ov{\mc M}$. 

\begin{lemma}\label{lemma105}
The footprints $F_I^\square$ satisfy the {\bf (Overlapping Condition)}, namely, 
\beqn
\ov{F_I^\square} \cap \ov{F_J^\square} \neq \emptyset \Longrightarrow I \preq J\ {\rm or}\ J \preq I.
\eeqn
\end{lemma}

\begin{proof}
Choose any pair $I_k, I_l \in {\mc I}$ with 
\beqn
\ov{F_{I_k}^\square} \cap \ov{F_{I_l}^\square} \neq \emptyset.
\eeqn
Without loss of generality, assume that $k < l$. We claim that $I_k \preq I_l$. Suppose it is not the case, then there exists $i \in I_k \setminus I_l$. Take $x$ in the intersection. Then by the definition of $F_I^\square$ and \eqref{eqn1011},
\beqn
x \in \ov{F_{I_k}^\square} \subset \ov{F_{p_i}^k} \subset G_{p_i}^l.
\eeqn
On the other hand, 
\beqn
x \in \ov{F_{I_l}^\square} \subset \ov{ \ov{\mc M} \setminus \ov{G_{p_i}^l} } \subset \ov{\mc M} \setminus G_{p_i}^l.
\eeqn
This is a contradiction. Hence $I_k \preq I_l$. 
\end{proof}

Therefore, we can shrink the charts $C_I^\bullet$ provided by Proposition \ref{prop101} to subcharts $C_I^\square$ whose footprints are $F_I^\square$. The shrinkings induce shrunk coordinate changes $T_{JI}^\square$ from $T_{JI}^\bullet$. If for any $I \in {\mc I}$, the shrunk $F_I^\square = \emptyset$, then we just delete $C_I^\bullet$ from the collection of charts and redefine the set ${\mc I}$. Then by Definition \ref{defn_atlas}, the collection 
\beqn
{\mc A}^\square:= \Big( \{ C_I^\square\ |\ I \in {\mc I}\},\ \{T_{JI}^\square\ |\ I \preq J \in {\mc I} \} \Big)
\eeqn
form a virtual orbifold atlas on $\ov{\mc M}$ in the sense of Definition \ref{defn_atlas}. 

Lastly, from ${\mc A}^\square$ one can obtain a good coordinate system. More precisely, choose precompact open subsets
\beqn
F_I \sqsubset F_I^\square,\ \forall I \in {\mc I}
\eeqn
which still cover $\ov{\mc M}$. Then by Theorem \ref{thm618}, there exists a precompact shrinking ${\mc A}$ of ${\mc A}^\square$ which is a good coordinate system (see Definition \ref{defn_gcs}), denoted by 
\beqn
{\mc A} = \Big( \{ C_I\ |\ I \in {\mc I}\},\ \{T_{JI}\ |\ I \preq J \in {\mc I}\} \Big).
\eeqn

\subsection{The virtual fundamental class}

We first notice that there is a strongly continuous map on the good coordinate system we constructed which extends the evaluation map and the forgetful map. Indeed, as in each chart $C_I$, points in $U_I$ are represented by genuine gauged maps (with extra data) over $r$-spin curves, there are continuous maps
\beqn
e_I: U_I \to \ov{\mc M}_\Gamma \times X^n.
\eeqn
These maps are compatible under coordinate changes, hence induce a strongly continuous map (see Definition \ref{defn_strong_continuous})
\beqn
\mf{e}: {\mc A} \to \ov{\mc M}_\Gamma \times X^n \subset \ov{\mc M}{}_{g,n}\times X^n.
\eeqn

\begin{lemma}
The map $\mf{e}$ is weakly submersive.
\end{lemma}

\begin{proof}
We chose the obstruction spaces such that for each $I$, each point $p\in \psi_I^{-1}(0)$, and each representing solution ${\bm v}$ of $p$, the linearized operator ${\mc D}_{\bm v}$, when restricted to deformations that vanish at markings and nodes, is transverse to the obstruction space ${\bm E}_I$. Therefore, by the implicit function theorem and the gluing construction, the evaluation map on this chart $e_I: U_I \to \ov{\mc M}{}_{g,n} \times X^n$ is transverse to every point.
\end{proof}

Then by the discussion of Subsection \ref{subsection67}, one obtains well-defined virtual fundamental classes 
\beqn
[\ov{\mc M}_{\sf \Gamma}]^{\rm vir} \in H_*( \ov{\mc M}_\Gamma \times X^n; {\mb Q}).
\eeqn

Before verifying that this class satisfies the expected properties, we prove that the virtual cycle is independent of the various choices we made during the construction. The main reason for this independence is that the good coordinate systems coming from any two different systems of choices can be put together to a good coordinate system over the space $\ov{\mc M}_{\sf \Gamma} \times [0,1]$. 

First, from the general argument, if we take a shrinking of the good coordinate system (i.e., a ``reduction'' in the sense of McDuff--Wehrheim), then the resulting virtual cycle remains in the same homology class. 

Then we compare two different systems of choices made during the inductive construction of this section. One can stratify $\ov{\mc M}_{\sf \Gamma} \times [0,1]$ by 
\beq\label{eqn1012}
{\mc M}_{\sf \Pi} \times \{0\},\ \ \ \ {\mc M}_{\sf \Pi} \times \{1\},\ \ \ \ {\mc M}_{\sf\Pi} \times (0,1) 
\eeq
Suppose we have two systems of structure as stated in Proposition \ref{prop101}, say basic charts 
\begin{align*}
&\ C_{p_i}^{[0]},\ i = 1, \ldots, N_0,\ &\ C_{q_j}^{[1]},\ j=1, \ldots, N_1,
\end{align*}
sum charts 
\begin{align*}
&\ C_{I_0}^{[0]},\ I_0 \in {\mc I}_0: = 2^{\{p_1, \ldots, p_{N_0}\}} \setminus \{\emptyset\},\ &\ C_{I_1}^{[1]},\ I_1 \in {\mc I}_1:= 2^{\{q_1, \ldots, q_{N_1}\}}\setminus \{ \emptyset \}
\end{align*}
and weak coordinate changes. Then by multiplying them with $[0, \frac{1}{4})$ (resp. $(\frac{3}{4}, 1]$), one obtains a system of basic charts indexed by elements in $\{ p_1, \ldots, p_{N_0}\} \sqcup \{ q_1, \ldots, q_{N_1}\}$. They can be viewed as a system of basic charts that cover the first two types of strata of \eqref{eqn1012}. Then by repeating the previous inductive construction with $\ov{\mc M}_{\sf \Gamma}$ replaced by $\ov{\mc M}_{\sf \Gamma} \times [0,1]$. One can complete the induction process to construct a weak virtual orbifold atlas $\tilde {\mc A}$ over the product $\ov{\mc M}_{\sf \Gamma} \times [0,1]$, or more precisely, prove an analogue of Proposition \ref{prop101} for $\ov{\mc M}_{\sf\Gamma} \times [0,1]$. Then from the weak atlas $\tilde {\mc A}$ one can shrink it to make a good coordinate system, such that when restrict to the two boundary strata, gives shrunk good coordinate systems of $\ov{\mc M}_{\sf \Gamma}$. Then by the same argument as \cite[Appendix A.8]{Tian_Xu_2021}, the intersection numbers coming from perturbations on the two sides are equal. Hence the virtual fundamental classes defined from the two good coordinate systems are the same. There were other choices that we made for the construction, for example, the choice on a Riemannian metric in order to define the gauge fixing condition. Different such choices can all be compared by using the cobordism argument. We leave such comparison to the reader.

\section{Properties of the Virtual Cycles}\label{section11}

In the last section we verify the properties of the virtual fundamental cycles of the moduli spaces of gauged Witten equation listed in Theorem \ref{thm71}.

\subsection{The dimension property}

The {\bf (Dimension)} property of the virtual cycles follows directly from the construction and the index calculation given in Section \ref{section8}.

\subsection{Disconnected graphs}

It suffices to consider the case that $\sf \Gamma = {\sf \Gamma}_1 \sqcup {\sf \Gamma}_2$. On the other hand, this property is not as easy as it seems, as  there is no obvious notion of ``products'' of good coordinate systems or virtual orbifold atlases. Here we prove the following proposition about the product of two moduli spaces. Abbreviate
\beqn
P = X^{n_1}\times \ov{\mc M}_{\Gamma_1},\ Q = X^{n_2} \times \ov{\mc M}_{\Gamma_2},\ Z = P \times Q = X^{n_1+ n_2} \times \ov{\mc M}_{\Gamma_1\sqcup \Gamma_2}
\eeqn
and 
\beqn
{\mc X}_P = \ov{\mc M}_{{\sf \Gamma}_1},\ \  {\mc X}_Q = \ov{\mc M}_{{\sf \Gamma}_2},\ \  {\mc X}_Z =  {\mc X}_P \times {\mc X}_Q.
\eeqn

\begin{prop}\label{prop111}
There exist the following objects.
\begin{enumerate}

\item Good coordinate systems ${\mc A}_P$, ${\mc A}_Q$, and ${\mc A}_Z$ on ${\mc X}_P$, ${\mc X}_Q$, and ${\mc X}_Z$ respectively. Let the charts be indexed by ${\mc I}$, ${\mc J}$ and ${\mc K}$ respectively.

\item Strongly continuous maps
\beqn
\eev_P: {\mc A}_P \to P,\ \eev_Q: {\mc A}_Q \to Q,\ \eev_Z: {\mc A}_Z \to Z.
\eeqn

\item An injective map 
\beqn
\iota: {\mc K} \to {\mc I} \times {\mc J}.
\eeqn

\item For each $K \in {\mc K}$, a precompact open embedding 
\beqn
\eta_K: C_{Z, K}  \hookrightarrow C_{P, I} \times C_{Q, J},\ {\rm where}\ \iota(K) = (I, J) \in {\mc I} \times {\mc J}.
\eeqn

\end{enumerate}
They satisfy the following properties.

\begin{enumerate}

\item If $\iota(K) = (I, J)$, $\iota(K') = (I', J')$ with $K \preq K'$, then $I \preq I'$ and $J \preq J'$.

\item For each $K$ with $\iota(K) = (I, J)$, if we view $C_{Z, K}$ as a precompact shrinking of $C_{P, I} \times C_{Q, J}$ via the open embedding $\eta_K$ then for all pairs $K \preq K'$, the coordinate changes $T_{K' K}$ are the induced coordinate changes.

\item The evaluation maps are compatible. Namely, if $\iota(K) = (I, J)$, then
\beqn
\eev_{Z,K} = ( \eev_{P, I} \times \eev_{Q, J}) \circ \eta_K.
\eeqn
\end{enumerate}
\end{prop}

Then we can construct a perturbation on (a shrinking of) ${\mc A}_Z$ which is of ``product type.'' Indeed, suppose ${\mf t}_P$ and ${\mf t}_Q$ be transverse perturbations of ${\mc A}_P$ and ${\mc A}_Q$ respectively such that the perturbed zero loci are compact oriented weighted branched manifolds contained in the virtual neighborhoods $|{\mc A}_P |$ and $|{\mc A}_Q |$. The perturbations consist of chartwise multi-valued sections 
\begin{align*}
&\ t_{P, I}: U_{P, I} \overset{m}{\to} E_{P, I},\ &\ t_{Q, J}: U_{Q, J} \overset{m}{\to} E_{Q, J}.
\end{align*}
Via the embeddings $\eta_K$, one can restrict the product of ${\mf t}_P $ and ${\mf t}_Q $ to a perturbation ${\mf t}_Z$ defined by 
\beqn
t_{Z,K}: U_{Z,K} \overset{m}{\to} E_{Z, K},\ t_{Z, K} = t_{P, I} \boxtimes t_{Q, J} \circ \eta_K,\ {\rm where}\ \iota(K) = (I, J).
\eeqn
The collection obvious satisfies the compatibility condition with respect to coordinate changes, hence a valid perturbation of ${\mf t}_Z$ on ${\mc A}_Z$. We call it the {\it induced} perturbation, or the {\it product} perturbation. Moreover, if both ${\mf t}_P$ and ${\mf t}_Q$ are transverse, so is ${\mf t}_Z$. However, it is not obvious whether the perturbed zero locus is compact or not. 

\begin{prop}\label{prop112}
There exist transverse perturbations ${\mf t}_P$, ${\mf t}_Q$ of ${\mc A}_P$ and ${\mc A}_Q$ such that via the inclusion $|{\mc A}_Z| \hookrightarrow |{\mc A}_P | \times |{\mc A}_Q |$, one has (as topological spaces equipped with the quotient topology)
\beq\label{eqn111}
|\tilde {\mf s}_Z^{-1}(0)| = |\tilde {\mf s}_P^{-1}(0)| \times |\tilde {\mf s}_Q^{-1}(0)|.
\eeq
\end{prop}

\begin{proof}
Take precompact shrinkings 
\begin{align*}
&\ F_{P, I}' \sqsubset F_{P, I},\ \forall I \in {\mc I},\ &\ F_{Q, J}' \sqsubset F_{Q, J},\ \forall J \in {\mc J}
\end{align*}
which still cover ${\mc X}_P$ and ${\mc X}_Q$ respectively. We can take two sequences of precompact shrinkings 
\beqn
\begin{split}
&\ F_{P, I} \sqsubset \cdots  \sqsubset U_{P, I}^{k+1} \sqsubset U_{P, I}^k \sqsubset \cdots \sqsubset F_{P, I},\ \forall I \in {\mc I},\\
&\ F_{Q, J} \sqsubset \cdots \sqsubset U_{Q, J}^{k+1} \sqsubset U_{Q, J}^k \sqsubset \cdots \sqsubset F_{Q, J},\ \forall J \in {\mc J},
\end{split}
\eeqn
such that
\begin{align}\label{eqn112}
&\ \bigcap_{k \geq 1} U_{P, I}^k = \ov{F_{P, I}'},\ &\ \bigcap_{k \geq 1} U_{Q, J}^k = \ov{F_{Q, J}'}.
\end{align}
Define 
\begin{align*}
&\ |{\mf U}_P^k|:= \bigsqcup_{I \in {\mc I}} U_{P, I}^k / \curlyvee,\ &\ | {\mf U}_Q^k|:= \bigsqcup_{J \in {\mc J}} U_{Q, J}^k/\curlyvee.
\end{align*}
Also choose a normal thickening ${\mc N}_P = \{ N_{P, I'I}\ |\ I \preq I'\}$ of ${\mc A}_P$ and a normal thickening ${\mc N}_Q = \{ N_{Q, J' J}\ |\ J \preq J' \}$ of ${\mc A}_Q$ (see Definition \ref{defn621}).\footnote{Since the bundles $E_{P, I}$ (resp. $E_{Q, J}$) naturally split as direct sums, a normal thickening only contains a collection of tubular neighborhoods).} By Theorem \ref{thm624}, there exist two sequence of perturbations $\tilde {\mf s}_{P, k}$ and $\tilde {\mf s}_{Q, k}$ satisfying 
\begin{align*}
&\ |\tilde {\mf s}_{P, k}^{-1}(0)| \subset |{\mf U}_P^k|,\ &\ |\tilde {\mf s}_{Q, k}^{-1}(0)| \subset |{\mf U}_Q^k|.
\end{align*}
Let the induced perturbation on ${\mc A}_Z$ be $\tilde {\mf s}_{Z, k}$. We claim that for $k$ sufficiently large, \eqref{eqn111} holds and $\tilde {\mf s}_{Z, k}^{-1}(0)$ is sequentially compact. 

First we show \eqref{eqn111} in the set-theoretic sense. Suppose it is not the case. Then there exist a subsequence (still indexed by $k$) and two sequences of points
\begin{align*}
&\ p_k \in |\tilde {\mf s}_{Q, k}^{-1}(0)|,\ &\ q_k \in |\tilde {\mf s}_{Q, k}^{-1}(0)|
\end{align*}
such that $(p_k, q_k ) \notin |\tilde {\mf s}_{Z, k}^{-1}(0)|$. By the precompactness of $|{\mf U}_P^k|$, $|{\mf U}_Q^k|$, and \eqref{eqn112}, there exist subsequences (still indexed by $k$) such that 
\begin{align*}
&\ \lim_{k \to \infty} p_k = p_\infty \in {\mc X}_P,\ &\ \lim_{k \to \infty} q_k = q_\infty \in {\mc X}_Q.
\end{align*}
Then there exist some $K\in {\mc K}$ with $\iota(K) = (I, J)$, and $u_{Z,\infty} \in U_{Z, K}$, $u_{P, \infty} \in U_{P, I}$, $u_{Q, \infty} \in U_{Q, J}$ such that
\beqn
(p_\infty, q_\infty) =  \psi_{Z, K} ( u_\infty) = ( \psi_{P, I}  (u_{P, \infty}), \psi_{Q, J} (u_{Q, \infty})),\ \ \eta_K( u_{Z,\infty} ) = (u_{P, \infty}, u_{Q, \infty}).
\eeqn

\begin{claim}
For $k$ sufficiently large, there exists $u_{P, k} \in U_{P, I}$ (resp. $u_{Q, k} \in U_{Q, J}$) which represents $p_k$ (resp. $q_k$) in the virtual neighborhood, namely
\begin{align*}
&\ \pi_P (u_{P,k}) = p_k,\ &\ \pi_Q ( u_{Q, k}) = q_k.
\end{align*}
\end{claim}

\begin{proof}[Proof of the claim] 
Since ${\mc I}$ and ${\mc J}$ are finite sets, by taking a subsequence, we may assume that there is $I' \in {\mc I}$ (resp. $J' \in {\mc J}$) such that sequence $p_k$ (resp. $q_k$) is represented by a sequence of points $u_{P, k}' \in U_{P, I'}^1$ (resp. $u_{Q, k}' \in U_{Q, J'}^1$), i.e.
\begin{align*}
&\ \pi_P ( u_{P, k}') = p_k,\ &\ \pi_Q (u_{Q , k}') = q_k.
\end{align*}
Since $U_{P, I'}^1$ and $U_{Q, J'}^1$ are precompact, by taking a further subsequence, we may assume that 
\begin{align*}
&\ u_{P, k}' \to u_{P, \infty}' \in \ov{F_{P, I'}'} \subset U_{P, I'},\ &\ u_{Q, k}' \to u_{Q, \infty}' \in \ov{F_{Q, J'}'} \subset U_{Q, J'}.
\end{align*}
Then $\psi_{P, I'}(u_{P, \infty}' ) = p_\infty$, $\psi_{Q, J'}(u_{Q, \infty}' ) = q_\infty$. Then by the {\bf (Overlapping Condition)} of the atlases ${\mc A}_P$ and ${\mc A}_Q$, we know $I \preq I'$ or $I' \preq I$ (resp. $J \preq J'$ or $J' \preq J$).

If $I' \preq I$, then $u_{P, \infty} \in U_{P, II'}$. Since $U_{P, II'} \subset U_{P, I'}$ is an open subset, the convergence $u_{P, k}' \to u_{P, \infty}'$ implies that for $k$ sufficiently large, $u_{P, k}' \in U_{P, II'}$. Hence 
\beqn
p_k = \pi_P (u_{P, k}') = \pi_P ( u_{P, k}),\ {\rm where}\ u_{P, k}:= \phi_{P, II'} (u_{P, k}').
\eeqn

On the other hand, if $I \preq I'$ and ${\rm dim} U_{P, I} = {\rm dim} U_{P, I'}$, the embedding $\phi_{P, I'I}$ is an open embedding hence is invertible. Then $u_{P, \infty} \in U_{P, I' I}$ and $u_{P, \infty}' = \phi_{P, I'I}(u_{P, \infty})$. Since $\phi_{P, I'I}^{-1}(U_{P, I'I})$ is an open subset of $U_{P, I'}$ and contains the element $u_{P, \infty}'$, for $k$ sufficiently large, there exists $u_{P, k} \in U_{P, I' I}$ with $u_{P, k}' = \phi_{P, I' I}(u_{P, k})$. Hence
\beqn
p_k  = \pi_P (u_{P, k}') = \pi_P (u_{P, k}).
\eeqn

Lastly, assume that $I  \preq I'$ but ${\rm dim} U_{P, I} < {\rm dim} U_{P, I'}$. Then the embedding $\phi_{P, I' I}$ has positive codimension. Let $N_{P, I' I} \subset U_{P, I'}$ be the tubular neighborhood given by ${\mc N}_P$. If $u_{P, k}' \notin N_{P, I' I}$, then the limit $u_{P, \infty}' \notin N_{P, I' I}$. This contradicts the fact that $p_\infty \in F_{P, I'} \cap F_{P, I}$. Hence $u_{P, k}' \in N_{P, I' I}$. We also know that the perturbations $\tilde {\mf s}_{P, k}$ are ${\mc N}_P$-normal. Hence $u_{P, k}' \in N_{P, I' I} \cap \tilde s_{P, I, k}^{-1}(0)$ implies that $u_{P, k}' \in {\rm Image}(\phi_{P, I' I})$. Hence there exists $u_{P, k} \in U_{P, I' I}$ with $u_{P, k}' = \phi_{P, I' I}(u_{P, k})$ and hence
\beqn
p_k = \pi_P (u_{P, k}). 
\eeqn
The situation for $u_{Q, k}'$ is completely the same. \end{proof}

Therefore, in the topology of the chart $U_{P, I}$ and $U_{Q, J}$, one has 
\beqn
(u_{P, k}, u_{Q, k}) \to (u_{P, \infty}, u_{Q, \infty})  \in \eta_K ( s_{Z, K, k}^{-1}(0)) \subset s_{P, I, k}^{-1}(0) \times s_{Q, I, k}^{-1}(0). 
\eeqn
Since $\eta_K (U_{Z, K})$ is an open subset of $U_{P, I} \times U_{Q, J}$, for $k$ sufficiently large, there exist unique $u_{Z, k} \in U_{Z, K}$ such that $\eta_K (u_{Z, k} ) = (u_{P, k}, u_{Q,k})$, and $u_{Z,k}$ converges to $u_{Z,\infty}$. Hence 
\beqn
\tilde s_{Z, K}(u_{Z, k}) = 0 \Longrightarrow (p_k, q_k) \in |\tilde {\mf s}_{Z, k}^{-1}(0)|.
\eeqn
This is a contradiction. Hence for $k$ sufficiently large, as sets,
\beqn
|\tilde {\mf s}_{Z, k}^{-1}(0) | = |\tilde {\mf s}_{P, k}^{-1}(0)| \times |\tilde {\mf s}_{Q, k}^{-1}(0)|.
\eeqn

Fix a large $k$ and abbreviate $\tilde {\mf s}_{P, k}= \tilde {\mf s}_P$, $\tilde {\mf s}_{Q, k} = \tilde {\mf s}_Q$, $\tilde {\mf s}_{Z, k} = \tilde {\mf s}_Z$. We prove that $|\tilde {\mf s}_Z^{-1}(0)|$ is sequentially compact in the quotient topology. Given any sequence $z_l \in |\tilde {\mf s}_Z^{-1}(0)|$ identified with a pair $(p_l, q_l) \in |\tilde {\mf s}_P^{-1}(0)| \times |\tilde {\mf s}_Q^{-1}(0)|$, by the compactness of $|\tilde {\mf s}_P^{-1}(0)| \times |\tilde {\mf s}_Q^{-1}(0)|$, we may assume a subsequence (still indexed by $l$) converges to $(p_{\infty}, q_{\infty}) \in |\tilde {\mf s}_P^{-1}(0)| \times |\tilde {\mf s}_Q^{-1}(0)| = |\tilde {\mf s}_Z^{-1}(0)|$. Then there exists $K \in {\mc K}$ with $\iota(K) = I \times J$ and $u_\infty \in U_K$, $u_{P, \infty} \in U_{P, I}$, $u_{Q, \infty} \in U_{Q, J}$ such that 
\beqn
p_\infty = \pi_{P, I}(u_{P, \infty}),\ q_\infty = \pi_{Q, J}(u_{Q, \infty}),\ \eta_K (u_\infty) = (u_{P, \infty}, u_{Q, \infty}). 
\eeqn
Similar to the proof of the previous claim, for $l$ sufficiently large, there exists $u_{P, l} \in U_{P, I}$ and $u_{Q, l} \in U_{Q, J}$ such that 
\beqn
p_l = \pi_{P, I}  (u_{P, l}),\ q_l = \pi_{Q, J}(u_{Q, l});\ u_{P, l} \to u_{P, \infty},\ u_{Q, l} \to u_{Q, \infty}. 
\eeqn
Since $\eta_K (U_{Z, K} ) \subset U_{P, I} \times U_{Q, J}$ is an open subset, for $l$ sufficiently large, $(u_{P, l}, u_{Q, l}) \in \eta_K (U_{Z, K})$. Denote $u_l = \eta_K^{-1}(u_{P, l}, u_{Q, l})$. Then $u_l \to u_\infty$. Since the map 
\beqn
s_{Z, K}^{-1}(0) \to |\tilde {\mf s}_Z^{-1}(0)|
\eeqn
is continuous we see 
\beqn
z_l = \pi_{Z, K}(u_l)  \to \pi_{Z, K}(u_\infty). 
\eeqn
Hence $|\tilde {\mf s}_Z^{-1}(0)|$ is sequentially compact. 

Lastly, when $\tilde {\mf s}_P$ and $\tilde {\mf s}_Q$ are transverse, so is $\tilde {\mf s}_Z$. Then $|\tilde {\mf s}_Z^{-1}(0)|$ is second countable. Hence sequential compactness is equivalent to compactness. Since the identity map 
\beqn
|\tilde {\mf s}_Z^{-1}(0)| \to \| \tilde {\mf s}_Z^{-1}(0)\|
\eeqn
is continuous and bijective, and the subspace topology is Hausdorff, this map is actually a homeomorphism. Hence $|\tilde {\mf s}_Z^{-1}(0)|$ is compact Hausdorff and second countable. Moreover, $|\tilde {\mf s}_P^{-1}(0)|$ and $|\tilde {\mf s}_Q^{-1}(0)|$ are also compact, Hausdorff and second countable. It is easy to see that the map $|\tilde {\mf s}_Z^{-1}(0)| \to |\tilde {\mf s}_P^{-1}(0)| \times |\tilde {\mf s}_Q^{-1}(0)|$ sends converging sequences to converging sequences, hence is continuous. Since the domain of the map is compact and the target is Hausdorff, this map is a homeomorphism.
\end{proof}

It remains to prove the equality between the virtual fundamental classes. By the K\"unneth theorem, it suffices to prove that for homology classes of $Z$ of the form $\beta_1 \otimes \beta_2 \in H_{k_1} ( P) \otimes H_{k_2}( Q)$, one has 
\beq\label{product}
[{\mc X}_Z]^{\rm vir} \cap (\beta_1 \otimes \beta_2) = ( [{\mc X}_P]^{\rm vir} \cap \beta_1) ([{\mc X}_Q]^{\rm vir} \cap \beta_2).
\eeq
Recall the definition using piecewise smooth collared cubical cycles. If $\beta_1$ and $\beta_2$ are represented by two such cycles, then their tensor product on the chain level is such a representative of $\beta_1 \otimes \beta_2$, which is a linear combination of products of a pair of cubes. The transversality condition implies that for each product of a pair of cubes, the intersection number is the product of the intersection numbers of the two factors. Therefore \eqref{product} follows and hence we proved the {\bf (Disconnected Graph)} property of Theorem \ref{thm71}.

\subsubsection{Proof of Proposition \ref{prop111}}

We first warn the reader that in the construction, the roles of $P$ and $Q$ are not symmetric. 

Recall the construction virtual orbifold atlases (see Proposition \ref{prop101}) and good coordinate systems on moduli spaces. Suppose we can construct collections of charts 
\beqn
\Big\{ C_{P, i} = (U_{P, i}, E_{P, i}, S_{P, i}, \psi_{P, i}, F_{P, i}\ |\ i = 1, \ldots, M_P \Big\},
\eeqn
\beqn
\Big\{ C_{Q, j} = (U_{Q, j}, E_{Q, j}, S_{Q, j}, \psi_{Q, j}, F_{Q, j}\ |\ j = 1, \ldots, M_Q \Big\}
\eeqn
such that all $F_{P, i}$ cover $P$ and all $F_{Q, j}$ cover $Q$. By the method of Section \ref{section10}, we have a collection of charts
\beqn
C_{P, I}^\square = ( U_{P, I}^\square, E_{P, I}^\square, S_{P, I}^\square, \psi_{P, I}^\square, F_{P, I}^\square), \ \ I \in {\mc I}: = 2^{\{1, \ldots, M_P \}} \setminus \emptyset
\eeqn
and
\beqn
C_{Q, J}^\square = ( U_{Q, J}^\square, E_{Q, J}^\square, S_{Q, J}^\square, \psi_{Q, J}^\square, F_{Q, J}^\square), \ \ J \in {\mc J}: = 2^{\{1, \ldots, M_Q \}} \setminus \emptyset.
\eeqn
Their footprints (which could be empty) are
\begin{align*}
&\ F_{P, I}^\square = \bigcap_{i \in I} F_{P, i},\ &\ F_{Q, J}^\square = \bigcap_{j \in J} F_{Q, j}.
\end{align*}
Let $T_{Q, I_2 I_1}^\square$ (resp. $T_{Q, J_2 J_1}^\square$) be the coordinate change from $C_{P, I_1}^\square$ to $C_{P, I_2}^\square$ (resp. from $C_{Q, J_1}^\square$ to $C_{Q, J_2}^\square$). We denote the two collections of data ${\mc A}_P^\square$ and ${\mc A}_Q^\square$, but keep in mind that, although the coordinate changes satisfy the {\bf (Cocycle Condition)}, they are not virtual orbifold atlases since they may not satisfy the {\bf (Overlapping Condition)}. 

Order elements of ${\mc I}$ and ${\mc J}$ as 
\begin{align*}
&\ I_1, \ldots, I_m;\ &\ J_1, \ldots, J_n
\end{align*}
such that 
\begin{align*}
&\  I_k \preq I_l \Longrightarrow k \leq l;\ &\ J_k \preq J_l' \Longrightarrow k \leq l.
\end{align*}
Choose precompact shrinkings of all $F_{P, i}$ as 
\beqn
G_{P, i, 1} \sqsubset F_{P, i, 1} \sqsubset \cdots \sqsubset G_{P, i, m} \sqsubset F_{P, i, m} = F_{P, i};
\eeqn
and 
\beqn
G_{Q, j}^1 \sqsubset F_{Q, j}^1 \sqsubset \cdots \sqsubset G_{Q, j}^m \sqsubset F_{Q, j}^m = F_{Q, j}.
\eeqn
such that 
\begin{align*}
&\ {\mc X}_P = \bigcup_{i= 1}^{M_P} G_{P, i, 1},\ &\ {\mc X}_Q = \bigcup_{j  = 1}^{M_Q} G_{Q, j}^1.
\end{align*}
Furthermore, for all $k = 1, \ldots, m$, choose precompact shrinkings as 
\beqn
G_{Q, j}^k=: G_{Q, j, 1}^k \sqsubset F_{Q, j, 1}^k \sqsubset \cdots \sqsubset G_{Q, j, n}^k  \sqsubset F_{Q, j, n}^k:= F_{Q, j}^k.
\eeqn
Then for $I = I_k \in {\mc I}$, define
\beqn
F_{P, I_k}^\bullet:= \Big( \bigcap_{i\in I_k} F_{P, i, k} \Big) \setminus \Big( \bigcup_{i \notin I_k} \ov{G_{P, i, k}} \Big);
\eeqn
for $1 \leq k \leq m$ and $J = J_l \in {\mc J}$, define
\beq\label{eqn113}
F_{Q, J_l}^{k, \bullet}:= \Big( \bigcap_{j \in J_l}  F_{Q, j, l}^k \Big) \setminus \Big( \bigcup_{j \notin J_l} \ov{G_{Q, j, l}^k} \Big).
\eeq

One has the following covering properties of the above open sets.

\begin{lemma}\label{lemma113}
One has 
\beqn
{\mc X}_P  = \bigcup_{I \in {\mc I}} F_{P, I}^\bullet;
\eeqn
and for all $k \in \{1, \ldots, m\}$, 
\beqn
{\mc X}_Q = \bigcup_{J \in {\mc J}} F_{Q, J}^{k, \bullet}.
\eeqn
\end{lemma}

\begin{proof}
The same as the proof of Lemma \ref{lemma105}.
\end{proof}

The product set ${\mc K} = {\mc I} \times {\mc J}$ is equipped with the natural product partial order. For $K = (I_k, J_l)$, define 
\beqn
F_K^\bullet = F_{I_k \times J_l}^\bullet:= F_{P, I_k}^\bullet \times F_{Q, J_l}^k \subset {\mc X}_P \times {\mc X}_Q. 
\eeqn
Lemma \ref{lemma113} implies that the collection of $F_K^\bullet$ for all $K \in {\mc K}$ is an open cover of ${\mc X}_P \times {\mc X}_Q$. 

\begin{lemma}\label{lemma114} The collection $F_K^\bullet$ satisfy the {\bf (Overlapping Condition)} of Definition \ref{defn_atlas}. Namely, for each pair $K_1, K_2 \in {\mc K}$, one has
\beqn
\ov{F_{K_1}^\bullet} \cap \ov{F_{K_2}^\bullet} \neq \emptyset \Longrightarrow K_1 \preq K_2\ {\rm or}\ K_2 \preq K_1.
\eeqn
\end{lemma}

\begin{proof}
Suppose
\begin{align*}
&\ K_1 = I_{k_1} \times J_{l_1},\ &\ K_2 = I_{k_2} \times J_{l_2}.
\end{align*}
Suppose $(p, q) \in \ov{F_{K_1}^\bullet} \cap \ov{F_{K_2}^\bullet}$. Then $p \in \ov{F_{P, I_{k_1}}^\bullet} \cap \ov{F_{P, I_{k_2}}^\bullet}$. Then by the {\bf (Overlapping Condition)} of the collection $F_{P, I}^\bullet$, either $I_{k_1} \preq I_{k_2}$ or $I_{k_2} \preq I_{k_1}$. Without loss of generality, assume the former is true. Then $k_1 \leq k_2$. We then need to prove that either $J_{l_1} \preq J_{l_2}$, or $k_1 = k_2$, and in the latter case, either $J_{l_1} \preq J_{l_2}$ or $J_{l_2} \preq J_{l_1}$. 

If $k_1 = k_2$, then one has 
\beqn
q \in \ov{F_{Q, J_{l_1}}^{k_1, \bullet}} \cap \ov{F_{Q, J_{l_2}}^{k_1, \bullet}}.
\eeqn
Then it is similar to the proof of Lemma \ref{lemma105}, that either $J_{l_1} \preq J_{l_2}$ or $J_{l_2} \preq J_{l_1}$. It remains to consider the case that $k_1 < k_2$. Suppose in this case $J_{l_1} \preq J_{l_2}$ does not hold. Then there exists $j^* \in J_{l_1} \setminus J_{l_2}$. One one hand, one has 
\beqn
q \in \ov{ \bigcap_{j  \in  J_{l_1} } F_{Q, j, l_1}^{k_1}} \subset \ov{ F_{Q, j^*, l_1}^{k_1}} \subset G_{Q, j^*, l_2}^{k_2}.
\eeqn
The last inclusion follows from the property of $G_{Q, j, l}^k$ and $F_{Q, j, l}^k$ and the fact that $k_1 < k_2$. On the other hand, one has
\beqn
q  \in \ov{ {\mc X}_Q \setminus \bigcup_{j \notin J_{l_2}}  \ov{G_{Q, j, l_2}^{k_2}}} \subset {\mc X}_Q \setminus G_{Q, j^*, l_2}^{k_2}.
\eeqn
This is a contradiction. Hence $J_{l_1}  \preq J_{l_2}$. 
\end{proof}

Now we define 
\beqn
\tilde {\mc J}:= \Big\{ (k, J)\ |\ k \in \{1, \ldots, m\},\ J \in {\mc J} \Big\}.
\eeqn
Define the partial order on $\tilde {\mc J}$ by $(k, J) \preq (l, J') \Longleftrightarrow k \leq l,\ J \preq J'$. The above construction also provides a covering of ${\mc X}_Q$ by 
\beqn
F_{Q, (k, J)}^\bullet:= F_{Q, J}^{k, \bullet},\ (k, J) \in \tilde J.
\eeqn
Recall the latter is defined by \eqref{eqn113}. Clearly these open sets cover $Q$. Moreover, this open cover satisfies the {\bf (Overlapping Condition)}.

\begin{lemma}\label{lemma115}
The open cover $\{ F_{Q, (k,J)}^\bullet\ |\ (k,J) \in {\mc J}\}$ of ${\mc X}_Q$ satisfies {\bf (Overlapping Condition)} of Definition \ref{defn_atlas}.
\end{lemma}

\begin{proof}
It is similar to the proof of Lemma \ref{lemma114}. The details are left to the reader. 
\end{proof}

Now we start to construct good coordinate systems ${\mc A}_P$, ${\mc A}_Q$ and ${\mc A}_Z$. The collection $\{ F_{P, I}^\bullet\}$ satisfies the {\bf (Overlapping Condition)} property of Definition \ref{defn_atlas}. Therefore, by the same method as in Section \ref{section10}, one can construct a good coordinate system 
\beqn
{\mc A}_P = \Big( \{ C_{P, I}\ |\ I \in {\mc I} \},\ \{ T_{I I'}\ |\ I'\preq I \in {\mc I} \} \Big)
\eeqn
of ${\mc X}_P$ whose collection of footprints are arbitrary precompact shrinkings $F_{P, I} \sqsubset F_{P, I}^\bullet$ the union of which still equals ${\mc X}_P$. We take the shrinkings $F_{P, I}$ so close to $F_{P, I}^\bullet$ such that 
\beqn
{\mc X}_Z = {\mc X}_P \times {\mc X}_Q = \bigcup_{I_k \in {\mc I},\ J \in {\mc J}} F_{Q, I_k} \times F_{Q, J}^{k, \bullet},\ \forall k.
\eeqn
On the other hand, for an arbitrary collection of precompact shrinkings
\beqn
F_{Q, \tilde J} \sqsubset F_{Q, \tilde J}^\bullet,\ \forall \tilde J = (k,J) \in \tilde {\mc J},
\eeqn
one can construct a good coordinate system on ${\mc X}_Q$
\beqn
{\mc A}_Q:= \Big( \big\{ C_{Q, \tilde J } = (U_{Q, \tilde J}, E_{Q, \tilde J}, S_{Q, \tilde J}, \psi_{Q, \tilde J}, F_{Q, \tilde J})\ |\ \tilde J \in \tilde {\mc J} \big\},\ \big\{ T_{Q, \tilde J' \tilde J}\ |\ \tilde J \preq \tilde J' \big\} \Big)
\eeqn
Notice that one has the inclusion 
\beqn
U_{Q, \tilde J} \sqsubset U_{Q, \tilde J}^\square.
\eeqn
Moreover, one can take the shrinkings $F_{Q, \tilde J}$ so close to $F_{Q, \tilde J}^\bullet$ such that the collection 
\beqn
F_{Z, K}^\vee:= F_{P, I_k} \times F_{Q, \tilde J},\ {\rm where}\ K = I_k \times J \in {\mc K} = {\mc I} \times {\mc J}\ {\rm and}\ \tilde J = (k, J)
\eeqn
still cover the product ${\mc X}_Z$. The {\bf (Overlapping Condition)} remains true for $F_{Z, K}^\vee$. Then choose precompact shrinkings $F_{Z, K} \sqsubset F_{Z, K}^\vee$ for all $K \in {\mc K}$ that still cover ${\mc X}_Z$, one can construct a good coordinate system ${\mc A}_Z$ on ${\mc X}_Z$ by shrinking the product charts 
\beqn
C_{Z, K}^\square = C_{P, I}^\square \times C_{Q, \tilde J}^\square,\ {\rm where}\ K = I \times J \in {\mc K}
\eeqn
to a chart 
\beqn
C_{Z, K}  = (U_{Z, K}, E_{Z, K}, S_{Z, K}, \psi_{Z, K}, F_{Z, K}).
\eeqn
Moreover, we can make the shrinking so small that for all $K$, 
\beqn
U_{Z, K} \sqsubset U_{P, I} \times U_{Q, J}\subset U_{P, I}^\square \times U_{Q, J}^\square,\ {\rm where}\ K = I\times J.
\eeqn
Then for the three good coordinate systems ${\mc A}_P$, ${\mc A}_Q$ and ${\mc A}_Z$, we see the map 
\beqn
\iota: {\mc K} \to {\mc I} \times \tilde {\mc J}
\eeqn
and the precompact open embeddings 
\beqn
\eta_K: C_{Z, K} \to C_{P, I}\times C_{Q, \tilde J},\ {\rm where}\ \iota(K) = I \times \tilde J
\eeqn
are obviously defined. The properties of $\iota$, $\eta_K$ and the evaluation maps are very easy to check. This finishes the proof of Proposition \ref{prop111}.

\subsection{Cutting edge}

In the situation of Theorem \ref{thm71} \eqref{thm71c}, $\ov{\mc M}_\Pi \subset \ov{\mc M}_\Gamma$ is an oriented smooth closed suborbifold. Then the {\bf (Cutting Edge)} property follows by applying Proposition \ref{prop626}.

\subsection{Composition}

In the situation of Theorem \ref{thm71} \eqref{thm71d}, consider the oriented smooth submanifold $\Delta_X \subset X \times X$ hence induces an embedding
\beqn
\iota: \Delta_X \times X^n \times \ov{\mc M}{}_{\tilde\Pi}\to X^{n+2} \times \ov{\mc M}{}_{\tilde\Pi}.
\eeqn
By Proposition \ref{prop626}, one has 
\beqn
[{\mc X}_{\sf \Pi}]^{\rm vir} = [{\mc X}_{\tilde{\sf\Pi}}]^{\rm vir} \cap \iota_* [\Delta_X \times \ov{\mc M}_{\tilde \Pi}] = [{\mc X}_{\tilde{\sf\Pi}}]^{\rm vir} \setminus {\rm PD}(\Delta_X).
\eeqn

\bibliography{mathref}

\providecommand{\bysame}{\leavevmode\hbox to3em{\hrulefill}\thinspace}
\providecommand{\MR}{\relax\ifhmode\unskip\space\fi MR }
\providecommand{\MRhref}[2]{%
  \href{http://www.ams.org/mathscinet-getitem?mr=#1}{#2}
}
\providecommand{\href}[2]{#2}
\begin{thebibliography}{{Mun}03}

\bibitem[Abo22]{Abouzaid_axiomatic}
Mohammed Abouzaid, \emph{An axiomatic approach to virtual chains},
  \url{http://arxiv.org/abs/2201.02911}, 2022.

\bibitem[CGMS02]{Cieliebak_Gaio_Mundet_Salamon_2002}
Kai Cieliebak, Ana Gaio, Ignasi {Mundet i Riera}, and Dietmar Salamon,
  \emph{The symplectic vortex equations and invariants of {H}amiltonian group
  actions}, Journal of Symplectic Geometry \textbf{1} (2002), no.~3, 543--645.

\bibitem[CGS00]{Cieliebak_Gaio_Salamon_2000}
Kai Cieliebak, Ana Gaio, and Dietmar Salamon, \emph{${J}$-holomorphic curves,
  moment maps, and invariants of {H}amiltonian group actions}, International
  Mathematics Research Notices \textbf{16} (2000), 831--882.

\bibitem[CWW17]{Chen_Wang_Wang}
Bohui Chen, Bai-Ling Wang, and Rui Wang, \emph{The asymptotic behavior of
  finite energy symplectic vortices with admissible metrics},
  \url{https://arxiv.org/abs/1706.07027}, 2017.

\bibitem[FJR13]{FJR_annals}
Huijun Fan, Tyler Jarvis, and Yongbin Ruan, \emph{The {W}itten equation, mirror
  symmetry and quantum singularity theory}, Annals of {M}athematics
  \textbf{178} (2013), 1--106.

\bibitem[FO99]{Fukaya_Ono}
Kenji Fukaya and Kaoru Ono, \emph{Arnold conjecture and {G}romov--{W}itten
  invariants for general symplectic manifolds}, Topology \textbf{38} (1999),
  933--1048.

\bibitem[FOOO16]{FOOO_2016}
Kenji Fukaya, Yong-Geun Oh, Hiroshi Ohta, and Kaoru Ono, \emph{Shrinking good
  coordinate systems associated to {K}uranishi structures}, Journal of
  Symplectic Geometry \textbf{14} (2016), no.~4, 1295--1310.

\bibitem[FOOO20]{FOOO_Kuranishi}
\bysame, \emph{Kuranishi structures and virtual fundamental chains}, Springer,
  2020.

\bibitem[Giv96]{Givental_96}
Alexander Givental, \emph{Equivariant {G}romov--{W}itten invariants},
  International Mathematics Research Notices \textbf{1996} (1996), 613--663.

\bibitem[GS05]{Gaio_Salamon_2005}
Ana Gaio and Dietmar Salamon, \emph{Gromov--{W}itten invariants of symplectic
  quotients and adiabatic limits}, Journal of Symplectic Geometry \textbf{3}
  (2005), no.~1, 55--159.

\bibitem[GS18]{Gu_Sharpe_2018}
Wei Gu and Eric Sharpe, \emph{A proposal for nonabelian mirrors},
  \url{https://arxiv.org/abs/1806.04678}, 2018.

\bibitem[HV00]{Hori_Vafa}
Kentaro Hori and Cumrun Vafa, \emph{Mirror symmetry},
  \url{http://arxiv.org/abs/hep-th/0002222}, 2000.

\bibitem[HWZ07]{HWZ1}
Helmut Hofer, Krzysztof Wysocki, and Eduard Zehnder, \emph{A general {F}redholm
  theory. {I}. {A} splicing-based differential geometry}, Journal of European
  Mathematical Society \textbf{9} (2007), 841--876.

\bibitem[HWZ09a]{HWZ2}
\bysame, \emph{A general {F}redholm theory {I}{I}: implicit function theorems},
  Geometric and Functional Analysis \textbf{19} (2009), 206--293.

\bibitem[HWZ09b]{HWZ3}
\bysame, \emph{A general {F}redholm theory {I}{I}{I}: {F}redholm functors and
  polyfolds}, Geometry and Topology \textbf{13} (2009), 2279--2387.

\bibitem[LLY97]{LLY_1}
Bong Lian, Kefeng Liu, and Shing-Tung Yau, \emph{Mirror principle {I}}, Asian
  Journal of Mathematics \textbf{1} (1997), no.~4, 729--763.

\bibitem[LT98a]{Li_Tian_2}
Jun Li and Gang Tian, \emph{Virtual moduli cycles and {G}romov--{W}itten
  invariants of algebraic varieties}, Journal of American Mathematical Society
  \textbf{11} (1998), no.~1, 119--174.

\bibitem[LT98b]{Li_Tian}
\bysame, \emph{Virtual moduli cycles and {G}romov--{W}itten invariants of
  general symplectic manifolds}, Topics in symplectic $4$-manifolds (Irvine,
  CA, 1996), First International Press Lecture Series., no.~I, International
  Press, Cambridge, MA, 1998, pp.~47--83.

\bibitem[Mil64]{Milnor_micro_1}
John Milnor, \emph{Microbundles {Part I}}, Topology \textbf{3} (1964),
  no.~Suppl. 1, 53--80.

\bibitem[MS04]{McDuff_Salamon_2004}
Dusa McDuff and Dietmar Salamon, \emph{${J}$-holomorphic curves and symplectic
  topology}, Colloquium Publications, vol.~52, American Mathematical Society,
  2004.

\bibitem[{Mun}99]{Mundet_thesis}
Ignasi {Mundet i Riera}, \emph{Yang--{M}ills--{H}iggs theory for symplectic
  fibrations}, Ph.D. thesis, Universidad Aut\'onoma de Madrid, 1999.

\bibitem[{Mun}03]{Mundet_2003}
\bysame, \emph{Hamiltonian {G}romov--{W}itten invariants}, Topology \textbf{43}
  (2003), no.~3, 525--553.

\bibitem[MW17a]{MW_2}
Dusa McDuff and Katrin Wehrheim, \emph{Smooth {K}uranishi atlases with
  isotropy}, Geometry and Topology \textbf{21} (2017), 2725--2809.

\bibitem[MW17b]{MW_3}
\bysame, \emph{The topology of {K}uranishi atlases}, Proceedings of London
  Mathematical Society \textbf{115} (2017), no.~3, 221--292.

\bibitem[Par16]{Pardon_virtual}
John Pardon, \emph{An algebraic approach to virtual fundamental cycles on
  moduli spaces of pseudo-holomorphic curves}, Geometry and Topology
  \textbf{20} (2016), 779--1034.

\bibitem[RS06]{Robbin_Salamon_2006}
Joel Robbin and Dietmar Salamon, \emph{A construction of the
  {D}eligne--{M}umford orbifold}, Journal of the European Mathematical Society
  \textbf{8} (2006), no.~4, 611--699.

\bibitem[Sat56]{Satake_orbifold}
Ichiro Satake, \emph{On a generalization of the notion of manifold},
  Proceedings of National Academy of Sciences \textbf{42} (1956), 359--363.

\bibitem[TF19]{Tehrani_Fukaya}
Mohammad Tehrani and Kenji Fukaya, \emph{Gromov--{W}itten theory via
  {K}uranishi structures}, pp.~111--252, American Mathematical Society, 2019.

\bibitem[TX16]{Tian_Xu_2}
Gang Tian and Guangbo Xu, \emph{Correlation functions in gauged linear
  $\sigma$-model}, Science China. Mathematics \textbf{59} (2016), 823--838.

\bibitem[TX17]{Tian_Xu_2017}
\bysame, \emph{The symplectic approach of gauged linear $\sigma$-model},
  Proceedings of the G\"okova Geometry-Topology Conference 2016 (Selman
  Akbulut, Denis Auroux, and Turgut \"Onder, eds.), 2017, pp.~86--111.

\bibitem[TX18a]{Tian_Xu}
\bysame, \emph{Analysis of gauged {W}itten equation}, Journal f\"ur die reine
  und angewandte Mathematik \textbf{740} (2018), 187--274.

\bibitem[TX18b]{Tian_Xu_geometric}
\bysame, \emph{Gauged linear sigma model in geometric phases. {I}.},
  \url{https://arxiv.org/abs/1809.00424}, 2018.

\bibitem[TX21]{Tian_Xu_2021}
\bysame, \emph{Virtual fundamental cycles of gauged {W}itten equation}, Journal
  f\"ur die reine und angewandte Mathematik \textbf{771} (2021), 1--64.

\bibitem[Ven15]{Venugopalan_quasi}
Sushmita Venugopalan, \emph{Vortices on surfaces with cylindrical ends},
  Journal of Geometry and Physics \textbf{98} (2015), 575--606.

\bibitem[Wit93]{Witten_LGCY}
Edward Witten, \emph{Phases of ${N}=2$ theories in two dimensions}, Nuclear
  Physics \textbf{B403} (1993), 159--222.

\bibitem[Xu16]{Xu_VHF}
Guangbo Xu, \emph{Gauged {H}amiltonian {F}loer homology {I}: definition of the
  {F}loer homology groups}, Transactions of the American Mathematical Society
  \textbf{368} (2016), 2967--3015.

\end{thebibliography}

\bibliographystyle{amsalpha}

\end{document}